\def\confversion{0}
\def\ifconf{\ifnum\confversion=1}
\def\ifnotconf{\ifnum\confversion=0}

\documentclass[11 pt]{article}
\def\showauthornotes{0}
\def\showkeys{0}
\def\showdraftbox{0}

\usepackage{xspace,xcolor,enumerate}
\usepackage{amsmath,amsfonts,amssymb}
\usepackage{color,graphicx}

\ifnum\showkeys=1
\usepackage[color]{showkeys}
\fi

\definecolor{darkred}{rgb}{0.5,0,0}
\definecolor{darkgreen}{rgb}{0,0.5,0}
\definecolor{darkblue}{rgb}{0,0,0.5}

\usepackage[pdfstartview=FitH,pdfpagemode=None,colorlinks,linkcolor=darkred,filecolor=blue,citecolor=darkred,urlcolor=darkred,pagebackref]{hyperref}

\usepackage{dsfont}

\ifnum\showauthornotes=1
\newcommand{\Authornote}[2]{{\sf\small\color{red}{[#1: #2]}}}
\newcommand{\Authorcomment}[2]{{\sf \small\color{gray}{[#1: #2]}}}
\newcommand{\Authorfnote}[2]{\footnote{\color{red}{#1: #2}}}
\else
\newcommand{\Authornote}[2]{}
\newcommand{\Authorcomment}[2]{}
\newcommand{\Authorfnote}[2]{}
\fi

\ifnum\showdraftbox=1
\newcommand{\draftbox}{\begin{center}
  \fbox{%
    \begin{minipage}{2in}%
      \begin{center}%
        \begin{Large}%
          \textsc{Working Draft}%
        \end{Large}\\
        Please do not distribute%
      \end{center}%
    \end{minipage}%
  }%
\end{center}
\vspace{0.2cm}}
\else
\newcommand{\draftbox}{}
\fi

\newtheorem{theorem}{Theorem}[section]

\newtheorem{definition}[theorem]{Definition}
\newtheorem{lemma}[theorem]{Lemma}
\newtheorem{remark}[theorem]{Remark}

\newtheorem{claim}[theorem]{Claim}

\def\FullBox{\hbox{\vrule width 6pt height 6pt depth 0pt}}

\def\qed{\ifmmode\qquad\FullBox\else{\unskip\nobreak\hfil
\penalty50\hskip1em\null\nobreak\hfil\FullBox
\parfillskip=0pt\finalhyphendemerits=0\endgraf}\fi}

\def\qedsketch{\ifmmode\Box\else{\unskip\nobreak\hfil
\penalty50\hskip1em\null\nobreak\hfil$\Box$
\parfillskip=0pt\finalhyphendemerits=0\endgraf}\fi}

\ifnotconf
\newenvironment{proof}{\begin{trivlist} \item {\bf Proof:~~}}
   {\qed\end{trivlist}}
\fi

\def\to{\rightarrow}
\def\eps{\varepsilon}
\def\epsilon{\varepsilon}

\def\eps{\epsilon}

\def\phi{\varphi}
\def\cal{\mathcal}

\newcommand{\sub}{\ensuremath{\subseteq}}
 
\newcommand{\defeq}{\stackrel{\mathrm{def}}=}

\renewcommand{\bar}{\overline} 

\newcommand{\ie}{i.e.,\xspace}

\newcommand{\mper}{\,.}
\newcommand{\mcom}{\,,}

\newcommand{\R}{{\mathbb R}}

\newcommand{\C}{{\mathbb C}}
\newcommand{\N}{{\mathbb{N}}}
\newcommand{\Z}{{\mathbb Z}}

\newcommand{\B}{\{0,1\}\xspace}
\newcommand{\pmone}{\{-1,1\}\xspace}

\newcommand{\indicator}[1]{\mathds{1}_{\{#1\}}}

\newcommand{\gauss}[2]{{\cal N(#1, #2)}}

\usepackage{nicefrac}
\newcommand{\abs}[1]{\ensuremath{\left\lvert #1 \right\rvert}}

\newcommand{\norm}[1]{\ensuremath{\left\lVert #1 \right\rVert}}

\newcommand{\mydot}[2]{\ensuremath{\left\langle #1, #2 \right\rangle}}

\newcommand{\ip}[1]{\left\langle #1 \right\rangle}

\def\bfx {{\bf x}}
\def\bfy {{\bf y}}

\def\bfw{{\bf w}}
\def\bfa{{\bf a}}

\newcommand{\vone}[1]{\mathbf{u}_{#1}}
\newcommand{\voneempty}{\mathbf{u}_{\emptyset}}
\newcommand{\vtwo}[2]{\mathbf{v}_{(#1, #2)}}
\newcommand{\vtwoempty}{\mathbf{v}_{(\emptyset, \emptyset)}}

\newcommand{\vartwo}[2]{x_{(#1, #2)}}
\newcommand{\vartwoempty}{x_{(\emptyset,\emptyset)}}

\newcommand{\Esymb}{\mathbb{E}}
\newcommand{\Psymb}{\mathbb{P}}

\DeclareMathOperator*{\ExpOp}{\Esymb}

\DeclareMathOperator*{\ProbOp}{\Psymb}
\renewcommand{\Pr}{\ProbOp}

\newcommand{\prob}[1]{\Pr\left[{#1}\right]}
\newcommand{\Prob}[2]{\Pr_{{#1}}\left[{#2}\right]}
\newcommand{\ex}[1]{\ExpOp\left[{#1}\right]}
\newcommand{\Ex}[2]{\ExpOp_{{#1}}\left[{#2}\right]}

\newfont{\inhead}{eufm10 scaled\magstep1}
\newcommand{\deffont}{\sf}

\newcommand{\suchthat}{{\;\; : \;\;}}

\newcommand{\opt}{{\sf OPT}\xspace}
\newcommand{\sdpopt}{{\sf FRAC}\xspace}
\newcommand{\lpopt}{{\sf FRAC}\xspace}
\newcommand{\round}{{\sf ROUND}\xspace}

\newcommand{\val}{{\sf val}}

\newcommand{\problemmacro}[1]{\textsf{#1}}

\newcommand{\maxkcsp}{\problemmacro{MAX k-CSP}\xspace}

\newcommand{\inparen}[1]{\left(#1\right)}             %\inparen{x+y}  is (x+y)
\newcommand{\inbraces}[1]{\left\{#1\right\}}           %\inbrace{x+y}  is {x+y}
\newcommand{\insquare}[1]{\left[#1\right]}             %\insquare{x+y}  is [x+y]
\newcommand{\Mnote}{\Authornote{M}}

\usepackage[latin1]{inputenc}
\usepackage[latin1]{inputenc}
\usepackage[toc,page]{appendix}

\RequirePackage{lineno}
\numberwithin{equation}{section}

\newcommand{\G}{\cal{G}}

\newcommand{\HF}{\hat{f}}
\newcommand{\hf}{\hat{f}}
\renewcommand{\C}{\cal{C}}

\newcommand{\y}{{\bf y}}
\newcommand{\z}{{\bf z}}
\newcommand{\bbox}{{[-1,1]^d}}
\newcommand{\tu}{\tilde{\bf u}}
\renewcommand{\gauss}[3]{\gamma_{#1}\inparen{#2, #3}}
\newcommand{\values}{V}
\newcommand{\Rp}{\cal{R}_p}
\renewcommand{\maxkcsp}{\problemmacro{CSP}}
\newcommand{\sat}{{\sf sat}}

\newcommand{\chars}{\cal{A}_s}
\newcommand{\charp}{\cal{A}_p}
\newcommand{\charlp}{\cal{A}_l}

\newcommand{\Ang}{{\sf Alice}\xspace}
\newcommand{\Dev}{{\sf Harry}\xspace}
\newcommand{\pay}{{\sf PayOff}}
\newcommand{\opay}{\overline{\sf PayOff}}

\newcommand{\indic}{\mathcal I}
\renewcommand{\defeq}{:=}

\renewcommand{\mapsto}{\to}

\title{A Characterization of  Approximation Resistance}
\author{
Subhash Khot\thanks{%
NYU.
{\tt khot@cims.nyu.edu}
Research partly supported by NSF Expeditions grant CCF-0832795, NSF Waterman
Award and BSF grant 2008059. Part of the research carried out while the author was at U.Chicago. }
\and
Madhur Tulsiani\thanks{%
    Toyota Technological Institute at Chicago.  {\tt madhurt@ttic.edu.} Research supported by NSF
    Career Award CCF-1254044.
  }%
\and
Pratik Worah\thanks{%
  NYU.
  {\tt pworah@cims.nyu.edu} Research partly supported by Subhash Khot's NSF Waterman Award. Part of the research carried out while the author was at U.Chicago.
}%
}

\begin{document}
\sloppy
\maketitle
\draftbox
\begin{abstract}
A predicate $f:\{-1,1\}^k \mapsto \{0,1\}$  with  $\rho(f) = \frac{|f^{-1}(1)|}{2^k}$ is called
%and an instance of CSP$(f)$, the fraction of constraints satisfied by a
%random $\{-1,1\}$ assignment to its variables is in the range
%$[\rho(f)-o(1), \rho(f)+o(1)]$ with high probability.
 {\it
approximation resistant} if given a near-satisfiable instance of CSP$(f)$,
it is computationally  hard to find an assignment that satisfies at least  $\rho(f)+\Omega(1)$  fraction of the constraints. 
%satisfied
%%by it
%is at least $\rho(f)+\Omega(1)$.
% and  in particular the assignment is significantly different from a random assignment.

We present a complete characterization of
approximation resistant predicates under the Unique Games Conjecture.
We also present characterizations in the {\it mixed} linear and semi-definite programming hierarchy
and the Sherali-Adams linear programming hierarchy. In the former case, the characterization coincides with the
one based on UGC. Each of the two characterizations is in terms of existence of a
probability measure with certain symmetry properties on
a natural convex polytope associated with the predicate.

%The predicate is called {\it approximation resistant} if given a near-satisfiable instance of CSP$(f)$,
%it is computationally {\it hard} to find an assignment such that the fraction of constraints satisfied
%is at least $\rho(f)+\Omega(1)$. When the predicate is odd, i.e. $f(-z)=1-f(z),\forall z\in \{-1,1\}^k$,
%it is easily observed that the notion of approximation resistance coincides
%with that of strong approximation resistance. Hence for odd predicates,
%in all the above settings, our characterization of strong approximation resistance is also a
%characterization of approximation resistance.

\newpage

%Given a $k$-ary boolean predicate $f$, we show that either
%there exists a canonical rounding scheme on $k$ rounds
%of the natural Sherali-Adams LP that can beat the
%random $0$-$1$ assignment for all $\maxkcsp(f)$ instances
%or there exists an almost optimal integrality gap even
%after many rounds of Sherali-Adams hierarchy.
%Furthermore, the previous result extends to the
%mixed hierarchy as well and the canonical rounding
%scheme therein requires at most $k+1$-dimensional
%hyperplane rounding. In particular, the second result
%is a dichotomy theorem for the approximability of
%boolean predicates modulo the UGC.
%Our proof techniques also show that $1$-dimensional
%hyperplane rounding is optimal for Grothendieck's
%inequality when the rounding algorithm is not oblivious
%to the SDP value of the instance.
%
%$\Mnote{Will write a detailed abstract later explaining each result as a separate item.}
%
\end{abstract}

\section{Introduction}
%\Mnote{Intro should define the quantity $\rho(f)$ for a predicate $f$}

Constraint satisfaction problems (CSPs) are some of the most well-studied NP-hard problems. Given a predicate
$f:\{-1,1\}^k \mapsto \{0,1\}$, an instance of CSP$(f)$ consists of $n$ $\{-1,1\}$-valued\footnote{It is more convenient
to work with $\{-1,1\}$-valued variables than $\{0,1\}$-valued ones. Here $-1$ corresponds to logical TRUE and $+1$ to
logical FALSE. Multiplication of variables in the \{-1,1\}-world  corresponds to XOR-ing them in the \{0,1\}-world. }
variables
and $m$ constraints where each constraint is the predicate $f$ applied to an ordered subset of $k$ variables, possibly in negated form.
For example, the OR predicate on $k$ variables corresponds to the $k$-SAT problem whereas the PARITY predicate (i.e.
whether the product of the variables is $+1$) on $k$ variables corresponds to the $k$-LIN problem.
The satisfiability problem for CSP$(f)$
asks whether there is an assignment that satisfies {\it all} the constraints. A well-known {\it dichotomy} result of Schaefer \cite{schfr} shows that
for every predicate $f$, the satisfiability problem for CSP$(f)$ is either in P or NP-complete and moreover his characterization
explicitly gives a (short) list of predicates for which the problem is in P.

An instance of CSP$(f)$ is called $\alpha$-satisfiable if there is
an assignment that satisfies at least  $\alpha$ fraction of the constraints. The focus of
this paper is whether given a $(1-o(1))$-satisfiable instance, there is an efficient algorithm
with a {\it non-trivial} performance.
The density of the predicate
$\rho(f) = \frac{|f^{-1}(1)|}{2^k}$ is the probability that a uniformly
random assignment to its variables satisfies the predicate. Given an instance of CSP$(f)$, 
 a naive algorithm that assigns
random $\{-1,1\}$ values to
its variables yields an assignment  that satisfies $\rho(f)$   fraction of the constraints  
 in expectation.

%and {\it with high probability} is in the
%range $\left[ \rho(f) - o(1), \rho(f)+ o(1) \right]$  if the instance is {\it reasonable} (e.g. if every variable
%appears in at most $o(m)$ constraints).

With this
observation in mind, we consider the well-studied notion of {\it approximation resistance}. The
 instance is promised to be $(1-o(1))$-satisfiable and the
algorithm is considered non-trivial if
it finds an assignment such that the
fraction of assignments satisfied is at least $ \rho(f)+\Omega(1)$,
i.e. the algorithm has to do something more clever than outputting a random assignment.
If such an efficient algorithm exists, the predicate is called {\it  approximable}
and {\it  approximation resistant} otherwise.
%In the second notion, one that is more well-studied, an
% algorithm is considered non-trivial if it finds an assignment that satisfies
%at least $\rho(f)+\Omega(1)$ fraction of the constraints.  If such an efficient
%algorithm exists, the predicate is called {\it approximable} and  {\it approximation resistant} otherwise. Note that an approximable predicate is also weakly approximable and as a
%contra-positive, a strongly approximation resistant predicate is also approximation
%resistant. Also, it is easily observed that for an odd predicate, i.e. $f(z)=1-f(z) \forall z \in
%\{-1,1\}^k$, the two notions are equivalent. For an odd predicate, the non-constant part of its
%Fourier representation has only odd degree monomials and the constant term (as always)
%is $\rho(f)$.  Flipping the sign of all variables simultaneously
%if necessary, a weak approximation (i.e. {\it deviating} from $\rho(f)$) is easily turned into
%a standard approximation (i.e. {\it exceeding} $\rho(f)$).
%At the risk of jumping ahead a little, we also mention here that all prior works
%showing approximation resistance of specific predicates, with possibly one exception,
%in fact show strong approximation resistance either implicitly or explicitly.

Towards the study of approximation resistance, it is  convenient to define the gap version of the problem.  GapCSP$(f)_{c,s}$ is a promise problem such that the instance
is guaranteed to be either $c$-satisfiable or at most $s$-satisfiable. Thus a predicate is approximation resistant if
GapCSP$(f)_{1-o(1), ~\rho(f)+o(1)}$ is not in P.  For resistant predicates,
 one would ideally like to show that the corresponding gap problem is NP-hard, or as is often the case, settle for a weaker notion of hardness
 such as UG-hardness (i.e. NP-hard assuming the Unique Games Conjecture \cite{khot}) or hardness, a.k.a. {\it integrality gap}, for a specific family of linear or semidefinite programming relaxation.
 %We now give an overview of prior works.  which concern only with the notion of
% approximation resistance, though as we mentioned, the notion of strong approximation
% resistance is implicit in the hardness works.

Until early 1990s, very little, if anything, was known regarding whether any {\it interesting} predicate is approximable or approximation
resistant. By now we have a much better understanding of this issue thanks to a sequence of spectacular results. Goemans and Williamson
\cite{gw,hastad} showed that 2SAT and 2LIN are approximable.\footnote{The result is more famously known for the MAX-CUT problem, but MAX-CUT is
not a CSP in our sense of the definition as it does not allow variable negations. Once variable negations are allowed, MAX-CUT is
same as 2LIN.} The discovery of the PCP Theorem \cite{fglss, as, almss}, aided by works such as \cite{bgs, razrep}, eventually led to
H\aa stad's result that 3SAT and 3LIN are approximation resistant and in fact that the appropriate gap versions are NP-hard! Since
then, many predicates have been shown to be approximation resistant (see e.g. \cite{glst, st, khot3col, ehr}, all NP-hardness) and most
recently, a remarkable result of Chan \cite{chan} shows the approximation resistance of the Hypergraph Linearity Predicate (he shows
NP-hardness whereas UG-hardness was shown earlier in \cite{st}). Also, a general result of Raghavendra \cite{ragh} shows that if a predicate is approximable, then it is so via
 a natural SDP relaxation of
the problem followed by a {\it rounding} of the solution (the result is more general than stated: it applies to every $(c,s)$-gap).

In this paper, our focus is towards obtaining a complete characterization of  approximation resistance for all predicates, in the
spirit of Schaefer's theorem. There has been some progress in this direction that we sketch now. Every predicate of arity $2$ is approximable as follows from Goemans and Williamson's algorithm \cite{gw}.\footnote{H\aa stad \cite{hastad2} shows the same for $2$-ary predicates over larger alphabet as well. We restrict to boolean alphabet in this paper.}
A complete classification of
predicates of arity $3$ is known \cite{hastad, zwick}: a predicate of arity $3$ is approximation resistant (NP-hard) if it is implied by PARITY up to
variable negations and approximable otherwise. For predicates of arity $4$, Hast \cite{hast}
gives a partial classification.
Austrin and Mossel \cite{am} show that a predicate
is approximation resistant (UG-hard) if the set $f^{-1}(1)$ of its satisfying assignments supports a pairwise independent distribution (for a somewhat more general sufficient
condition see \cite{ah2}).
Using this sufficient condition, Austrin and H\aa stad \cite{ah} show that a vast majority of $k$-ary predicates for large $k$ are approximation resistant.
Hast \cite{hast2}
shows that a $k$-ary predicate with at most $k-1$ satisfying assignments is approximable.

In spite of all these works, a complete characterization of approximation resistance
remained elusive.
A recent result of Austrin and Khot \cite{ak} gives a complete characterization of approximation
resistance (UGC-based) when the CSP is restricted to be $k$-partite\footnote{Meaning the set of variables is partitioned into $k$ layers and for every constraint, the $i^{th}$ variable is from the $i^{th}$ layer.} and the predicate is even.\footnote{Meaning $f(-z)= f(z) \ \forall z \in
\{-1,1\}^k$.} Given an even predicate $f$, the authors therein associate with it a convex polytope $\cal{C}(f)$ consisting of all vectors
of dimension ${k}\choose{2}$ that arise as the {\it second moment vectors} $\left( \Ex{z \sim \nu}{z_i z_j}|1 \leq i < j \leq k\right)$
of distributions $\nu$ supported on $f^{-1}(1)$. It is shown that the $k$-partite version of CSP$(f)$ is approximation
resistant (UG-hard) if and only if $\cal{C}(f)$ supports a distribution (a probability measure to be more
precise)
 with a certain (difficult to state) property. The $k$-partiteness condition is rather restrictive
 and without the evenness condition, one would need to take into account the {\it first moment vector} $\left( \Ex{z \sim \nu}{z_i} | 1 \leq i \leq k\right)$ as well
and it is not clear how to incorporate this in \cite{ak}.

\subsubsection*{Characterizing Approximation Resistance}
In this paper, we indeed give a complete characterization of
approximation resistance, via an approach that is entirely different than \cite{ak}.  %As we noted, for odd predicates, this is same as characterizing
%approximation resistance: one gets strong approximation resistance on the hardness side and
%standard approximation on the algorithmic side, i.e. best of both the worlds. One interesting family
%of odd predicates is balanced linear threshold functions.
Before stating the characterization, we point out that  the characterization is not as
{\it simple} as one may wish and we do not yet know whether it is decidable, both these features also shared by the result in \cite{ak}.
\footnote{The characterization in \cite{ak} is recursively enumerable, i.e. there is a procedure that on a predicate that is
approximable, terminates and declares so. Our characterization is also recursively enumerable though it is not clear from its
statement and one has to work through the proof. We omit this aspect from the current version of the paper.}

Roughly speaking our characterization
states that a predicate $f:\{-1,1\}^k \mapsto \{0,1\}$ is approximation resistant (UG-hard) if and only if a convex polytope $\cal{C}(f)$
associated with it supports a probability measure with certain symmetry properties.\footnote{The characterization
in \cite{ak}, in hindsight, may also be stated in terms of similar symmetry properties,
and we do so in this paper.} Specifically, let  $\cal{C}(f)$ be the convex polytope
consisting of all vectors of dimension $k+ \binom{k}{2}$ that arise as the {\it first and second moment vectors}
$$\left( \left(\Ex{z \in \nu}{z_i}|1 \leq i \leq k\right), \left( \Ex{z \sim \nu}{z_i z_j}|1 \leq i < j \leq k\right) \right)$$
of distributions $\nu$ supported on $f^{-1}(1)$. For a measure $\Lambda$ on $\cal{C}(f)$ and a
 subset $S \subseteq [k]$, let $\Lambda_S$ denote the projection of $\Lambda$ onto the co-ordinates in $S$. For a permutation
 $\pi: S \mapsto S$ and a choice of signs $b \in \{-1,1\}^S$, let $\Lambda_{S,\pi,b}$ denote the measure $\Lambda_S$
 after permuting the indices in $S$ according to $\pi$ and then (possibly) negating the co-ordinates according to multiplication
 by $\{b_i\}_{i\in S}$. We are now ready to state our characterization.

\begin{definition}\label{def:chars}  Let $\chars$ be the family of all predicates (of all arities)
$f:\{-1,1\}^k \mapsto
\{0,1\}$ such that there is a probability measure
$\Lambda$ on $\cal{C}(f)$ such that
 for every $1 \leq t \leq k$, the {\it signed measure}
 \begin{equation}
 \Lambda^{(t)} ~:=~
\ExpOp_{|S| = t} ~\ExpOp_{\pi : [t] \to [t]} ~\Ex{b \in \{-1,1\}^t }{
 \left( \prod_{i = 1}^t b_i \right) \cdot\HF(S) \cdot \Lambda_{S,\pi,b}}\label{intro:main:eqn}
 \end{equation}
vanishes identically. If so, $\Lambda$ itself is said to vanish.
\end{definition}
\Mnote{I changed $\pi$ to be a permutation on $[t]$ and $b \in \pmone^t$ in the above expression, to
be consistent with the calculations later. We write $\y_1,\ldots,\y_t$ later, so it's better to also
write $b_1,\ldots,b_t$.}

Much elaboration is in order. In the above expression, the expectation is over a random subset of $[k]$ of size $t$,
a random permutation $\pi$ of $S$ and a random choice of signs $b$ on $S$. The coefficients $\HF(S)$ are the Fourier coefficients of
the predicate $f$, namely, the coefficients in the Fourier representation:
$$f(x_1,\ldots,x_k) = \rho(f) + \sum_{S \not= \emptyset} \HF(S) \prod_{i\in S} x_i. $$
A {\it signed measure} is allowed to take negative values as well (as is evident from the possibly negative sign of
$\HF(S)$ and $\prod_{i = 1}^t b_i$ in the above expression). An equivalent way to state the condition is that if one writes the
Expression \eqref{intro:main:eqn}
 as a difference of two non-negative measures $\Lambda^{(t),1}$ and  $\Lambda^{(t),2}$ by grouping the
terms with positive and negative coefficients respectively, then the two non-negative measures are identical.

Our characterization states that if $f \in \chars$, then $f$ is  approximation resistant (UG-hardness)
and otherwise  approximable. In the former case, the vanishing measure $\Lambda$ is a
{\it hard to round} measure (in fact any proposed hard to round measure must be a vanishing measure).
In the latter case, we can in fact conclude that the predicate is  approximable via
a natural SDP relaxation followed by a {\it $(k+1)$-dimensional rounding algorithm}. A $(k+1)$-dimensional rounding algorithm samples a
$(k+1)$-dimensional rounding function  $\psi: \R^{k+1} \mapsto \{-1,1\}$ from an appropriate distribution, projects the SDP vectors onto a
random $(k+1)$-dimensional subspace and then rounds using $\psi$. We find this conclusion rather surprising. As mentioned earlier, it
follows from Raghavendra \cite{ragh} that if a predicate is approximable then it is so via (the same) SDP relaxation followed by a
rounding. However his rounding
(and/or the one in \cite{raghsteu}) is high dimensional in the sense that one first projects onto a random $d$-dimensional
subspace and then rounds using an appropriately sampled function $\psi: \R^d \mapsto \{-1,1\}$ and there is no a priori upper bound on the dimension $d$ required.

It is instructive to check that our characterization generalizes the sufficient
condition for approximation resistance due to Austrin and Mossel \cite{am}.
Suppose that a predicate supports a pairwise independent distribution. This amounts
to saying that the $k+ \binom{k}{2}$ dimensional all-zeroes vector lies in the polytope $\cal{C}(f)$. It is immediate that
the measure $\Lambda$ concentrated at this single vector is vanishing (the all-zeroes vector and its projections onto subsets
$S$ remain unchanged under sign-flips via $b \in \{-1,1\}^S$ and these terms cancel each other out due to the sign $\prod_{i\in S} b_i$
in the expression) and hence the predicate is  approximation resistant. It is also instructive to check the case
$t=1$. In this case, the condition implies, in particular, that
  \begin{equation*}
 \Ex{\zeta \sim \Lambda}{
\sum_{i=1}^k  \HF(\{i\}) \cdot \zeta(i)   }  = 0.
 \end{equation*}
Here $\zeta(i)$ denotes the $i^{th}$ first moment (i.e. bias) in the vector $\zeta \in \cal{C}(f)$.
For all the predicates that are known to be approximation resistant so far in literature, there is always a single {\it hard to
round point} $\zeta$, i.e. the measure $\Lambda$ is concentrated at a single point $\zeta$. In that case, the above condition specializes
to $\sum_{i=1}^k \HF(\{i\}) \cdot \zeta(i)   = 0$ and this condition is known to be necessary (as a folklore among the experts at least).
This is because otherwise a rounding that simply rounds each variable according to its bias given by the LP relaxation
(and then flipping signs of all variables simultaneously if necessary) will strictly exceed
the threshold $\rho(f)$. The term $\sum_{i=1}^k \HF(\{i\}) \cdot \zeta(i)$ represents the contribution to the
{\it advantage over $\rho(f)$} by the level-$1$ Fourier coefficients and a standard trick allows one to ignore the (potentially
troublesome) interference from higher order Fourier levels.
%These considerations also apply in the case $t=2$ where our
%condition amounts to saying that a simple rounding based only on level-$2$ Fourier coefficients and pairwise
%interactions between SDP vectors does not succeed.
The conditions for $t \geq 2$ intuitively rule out
successively more sophisticated rounding strategies and taken together for all $t \in [k]$ form a complete set of necessary {\it and}
sufficient conditions for strong approximation resistance.

It seems appropriate to point out another aspect in which our result differs from \cite{ragh, raghsteu}.
It can be argued (as also discussed in \cite{ak})  that \cite{raghsteu} also gives a characterization of approximation
resistance in the following sense. The authors therein propose a brute force search over all instances and their potential
SDP solutions on $N = N(\eps)$ variables which
determines the hardness threshold up to an additive $\eps$. Thus if a predicate is approximable with an advantage of say $2\eps$ over the
trivial $\rho(f)$ threshold and if $\eps$ were known a priori, then the brute force search will be able to affirm this.
However, there is no a priori lower bound on $\eps$ and
thus this characterization is not known to be decidable either. Moreover, it seems somewhat of a stretch to call it a
characterization because of the nature of the search involved. On the other hand, our characterization  is in terms of concrete symmetry properties of
a measure supported on the explicit and natural polytope  $\cal{C}(f)$.
%sense that it depends purely on the predicate $f$ and the corresponding polytope $\cal{C}(f)$.
  The characterization
does not depend on the topology (i.e. the hyper-graph structure) of the CSP instance. We find this conclusion rather surprising
as well. A priori, what might make a predicate hard is both a {\it hard to round} measure over local LP/SDP distributions (i.e.
a measure $\Lambda$ on $\cal{C}(f)$) as well as the topology of the constraint hyper-graph (i.e. how the variables and constraints  {\it fit
together}). Our conclusion is that the latter aspect is not relevant, not in any direct manner at least. This conclusion may be contrasted
against Raghavendra's result. He shows that any SDP integrality gap {\it instance} can be used as a gadget towards proving a UG-hardness
result with the same gap. The {\it instance} here refers to both the variable-constraints topology
and the local LP/SDP distributions and from his result, it is not clear whether one or the other or both the aspects are required to make the CSP hard.

When CSP instances are restricted to be $k$-partite as in
\cite{ak}, we are able to obtain a complete characterization. For the family
$\charp$ defined below, if $f \in \charp$ then the partite version is  approximation
resistant and otherwise the partite version is approximable.

\begin{definition} \label{def:charp} Let $\charp$ be the family of all predicates (of all arities)
$f:\{-1,1\}^k \mapsto
\{0,1\}$ such that there is a probability measure
$\Lambda$ on $\cal{C}(f)$ such that
 for every $S \subseteq [k], S \not= \emptyset$, the {\it signed measure}
 \begin{equation}
 \Lambda^{S} ~:=~
 \HF(S) \cdot  \Ex{b \in \{-1,1\}^S }{
 \left( \prod_{i\in S} b_i \right) \cdot \Lambda_{S, b}}\label{intro:mainptt:eqn}
 \end{equation}
vanishes identically.
\end{definition}

The difference from Definition \ref{def:chars} is that each non-empty set $S$ is considered
separately and there are no permutations of the set. We note that for even predicates,
the first $k$ co-ordinates in the body $\cal{C}(f)$ corresponding to the first moments (i.e.
``biases") can be assumed to be identically zero and then the characterization boils down
to one in \cite{ak} (though there it is stated differently).

 We point out some directions left open by the discussion so far
 (we do not consider these as the focus of the current paper). Firstly,
it would be nice to show that our characterization is decidable. Secondly,
we are not aware of an approximation resistant predicate where one needs a combination of more than
one {\it hard to round} points
in $\cal{C}(f)$. In other words, it might be the case that for every  approximation
resistant predicate, there exists a vanishing
measure $\Lambda$ on $\cal{C}(f)$ that is concentrated on a single point or on a bounded number
of points with an a priori bound. If this were the case, our
characterization will be decidable (we omit the proof). In this regard, it would be interesting to investigate the
 example of an arity $4$ predicate in \cite{ah2}, Example 8.7 therein. The
authors show that the predicate is approximation resistant by presenting a {\it hard to round} point. However, the
approximation resistance is shown in an ad hoc manner that, as far as we see, does not immediately give a vanishing measure
$\Lambda$. Such a measure must exist by our results and would perhaps require a combination of more than one point.
%(this concept is elaborated later). The same point is not good enough to show strong approximation
%resistance, but it is possible that there is another one or a {\it probability measure} on points
%which is good enough (we have not investigated this
%possibility yet).
Thirdly, it will be interesting to show that
for some special classes of
predicates our characterization takes a much simpler form. For instance,
\cite{chis}
asks whether there is a linear threshold
predicate that is approximation resistant. It would be nice if for such predicates our
characterization takes a simpler form and leads to a resolution of the question. Finally, for predicates that
do not satisfy our characterization and hence not approximation resistant, our result suggests that there are sophisticated rounding algorithms
whose analysis may require looking at terms at level $3$ and above in the Fourier representation (as opposed to only using terms at first and
second level). We do not yet have explicit examples and leave it as an exciting open question.

\subsubsection*{Results for Linear and Semidefinite Relaxations}
We now move onto a discussion about our results concerning
the notion of  approximation resistance in the context of
linear and/or semi-definite programming relaxations. A CSP instance can be formulated as an
integer program and its variables may be {\it relaxed} to assume real values (in the case
of LP relaxation) or vector values (in the case of SDP relaxation). The {\it integrality gap}
of a relaxation is the maximum gap between the optimum of the integer program
and the optimum of the relaxed program. An integrality gap instance is a concrete instance
of a CSP whose LP/SDP optimum is {\it high} and the integer optimum is {\it low}. Constructing
such gap instances is taken as evidence that the LP/SDP based approach will not achieve
good approximation to the CSP. The LP/SDP relaxation may be {\it ad hoc} or may be obtained
by systematically adding inequalities, in successive {\it rounds}, each additional round
yielding a potentially {\it tighter} relaxation. The latter method is referred to as
an LP or SDP hierarchy and several such hierarchies have been proposed and well-studied \cite{ch-tul}.

In this paper, we focus on one ad hoc relaxation that we call {\it basic  relaxation}
and two hierarchies, namely the {\it mixed hierarchy} and the {\it Sherali-Adams LP hierarchy}.
We refer to Section \ref{sec:prelims} for their formal definitions, but provide a quick sketch here.
Consider a CSP$(f)$ instance with a $k$-ary predicate $f$,
a set of variables $V = \{x_1,\ldots,x_n\}$ and
constraints $C_1,\ldots,C_m$. We think of the number of rounds $r$ as $k$ or more.
The $r$-round Sherali-Adams LP is required to provide, for
every set $S \subseteq V, |S| \leq r$, a {\it local}
distribution $D(S)$ over assignments to the set $S$, namely $\{-1,1\}^S$. The local distributions
must be consistent in the sense that for any two sets $S, T$ of size at most $r$ and
$S \cap T \not= \emptyset$, the local distributions to $S$ and $T$ have the same marginals on $S
\cap T$. The $r$-round mixed hierarchy, in addition, is supposed to assign unit vectors $\vone{i}$ to
variables $x_i$ such that the pairwise inner products of these vectors match the second moments of the
local distributions: $\left<\vone{i}, \vone{j} \right> = \Ex{\sigma \sim D(\{i,j\}) }{\sigma(i)\sigma(j)}$
(this is a somewhat simplified view). The {\it basic relaxation} is a reduced form of the
$k$-round mixed hierarchy where a local distribution over a set $S$ needs to be specified  only if
$S$ is a set of $k$ variables of some constraint $C_\ell$. The only consistency requirements
are that $\left<\vone{i}, \vone{j} \right> = \Ex{\sigma \sim D(S) }{\sigma(i)\sigma(j)}$ if variables
$i, j$ appear together inside some constraint $C_\ell$ on set $S$. Finally, the objective function
for all three programs is the same: the probability that an assignment sampled from the local
distribution over a constraint satisfies the predicate (accounting for variable negations),
averaged over all constraints.

A $(c,s)$-integrality gap for a relaxation is an instance that is at most $s$-satisfiable, but has a
feasible LP/SDP solution with objective value at least $c$. A predicate is approximation
resistant w.r.t. a given relaxation if the relaxation has $(1-o(1), \rho(f)+o(1))$ integrality
gap. The general result of Raghavendra referred to before shows that for any gap location
$(c,s)$, UG-hardness is equivalent to integrality gap for the basic relaxation. Moreover,
the general results of Raghavendra and Steurer \cite{raghsteu2} and Khot and Saket \cite{ks09} show that
the integrality gap for basic relaxation is equivalent to that for a super-constant
number of rounds of the mixed hierarchy.

%All our integrality gap constructions achieve strong resistance, namely that the constructed
%CSP instance has LP or SDP value $1-o(1)$ and for any (integral) assignment to the instance,
%the fraction of satisfied constraints is in the range $\left[\rho(f)-o(1), \rho(f)+o(1) \right]$. We
%call this strong $(1-o(1), \rho(f) \pm o(1))$ integrality gap.
Our characterization of  approximation resistance
for the basic relaxation and the mixed hierarchy is the same and coincides
with one in Definition \ref{def:chars} whereas that for the Sherali-Adams LP is different and
presented below.

When $f \in \chars$ as in Definition \ref{def:chars}, we construct a  $(1-o(1), \rho(f) + o(1))$
integrality gap
for the basic relaxation. From the general results \cite{ragh,
raghsteu2, ks09} mentioned before, integrality gap for basic relaxation can
be translated into the same gap for mixed hierarchy and into UG-hardness.
%For us, in the
%NO case, we need to be more careful since we want to preserve the strong resistance,
%namely that every (integral) assignment satisfies between $\rho(f)\pm o(1)$ fraction of assignments.
%Still, these translations are by now standard and well-understood and are omitted from the
%current paper.
When $f \not\in \chars$, we know that the predicate is approximable and moreover
the algorithm is a rounding of the basic relaxation. When $f \in \charp$ as in
Definition \ref{def:charp}, the UG-hardness as well as integrality gap constructions can be
ensured to be on $k$-partite instances, as in \cite{ak}.

Finally we focus on the characterization of approximation resistance in Sherali-Adams LP hierarchy.
Here the situation is fundamentally different at a conceptual level. The
difference is illustrated by the (arguably the simplest) predicate 2LIN.
Goemans and Williamson show that 2LIN is approximable via an SDP relaxation, namely the
basic relaxation according to our terminology. In fact the approximation is really
close: on an $(1-\eps)$-satisfiable instance, the relaxation finds
$( 1- \frac{\arccos(1-\eps)}{\pi} )$-satisfying assignment (which is asymptotically
$1-O(\sqrt{\eps})$). It is also known that this is precisely the integrality gap as well as
UG-hardness gap \cite{fsch, kkmo}.
However, the predicate turns out to be approximation resistant in
the Sherali-Adams LP hierarchy as shown by de la Vega and Mathieu \cite{delavega}!  They show
$(1-o(1), \frac{1}{2}+o(1))$ integrality gap for $\omega(1)$ rounds of the
Sherali-Adams hierarchy, which is subsequently improved to $n^{\Omega(1)}$ rounds in \cite{yury}.

Even though the approximation resistance in Sherali-Adams LP hierarchy is fundamentally different,
our characterization of  resistance here looks {\it syntactically} similar to the ones before,
once we ignore the second moments (which are not available in the LP case).

\begin{definition}\label{def:charlp}  Let $\charlp$ be the family of all predicates (of all arities)
$f:\{-1,1\}^k \mapsto
\{0,1\}$ such that there is a probability measure
$\Lambda^*$ on $\cal{C}^*(f)$ such that
 for every $1 \leq t \leq k$, the {\it signed measure}
 \begin{equation}
 \Lambda^{*,(t)} ~:=~
\ExpOp_{|S| = t} ~\ExpOp_{\pi : [t] \to [t]} ~\Ex{b \in \{-1,1\}^t }{
 \left( \prod_{i = 1}^t b_i \right) \cdot\HF(S) \cdot \Lambda^*_{S,\pi,b}}\label{intro:mainlp:eqn}
 \end{equation}
vanishes identically. Here $\cal{C}^*(f)$ is the projection of
the polytope $\cal{C}(f)$ to the first $k$ co-ordinates corresponding to the first moments and
$\Lambda^*_{S,\pi,b}$ are as earlier, but for the projected polytope $\cal{C}^*(f)$.
\end{definition}

We show that if $f \in \charlp$, then there is a $(1-o(1), \rho(f) + o(1))$  integrality
gap for a super-constant number of rounds of Sherali-Adams hierarchy. Otherwise there is an
approximation given by  $k$-rounds of the hierarchy.  For the class of symmetric $k$-ary predicates,
our characterization takes a simple form. If $f$ is symmetric then $f \in \charlp$ if and only if
there are inputs $x, y \in \{-1,1\}^k$ such that
$f(x)=f(y)=1$ and $\sum_{i=1}^k x_i \geq 0$, $\sum_{i=1}^k y_i \leq 0$.

%For Sherali-Adams hierarchy, we also get a complete characterization of strong approximation resistance,
%though not as explicit as in Definition \ref{def:charlp}.
%% What we can say is that firstly,
%% if the integrality gap (in the usual notion) is $(1-o(1), \rho(f)+o(1))$ for $k$ rounds of
%% the hierarchy (or even for the basic version that requires local distributions over constraints
%% only), then it also stays the same for even a super-constant number of rounds.
%% Secondly,
%Also, if a predicate is weakly approximable in this hierarchy, then this is so via a generic algorithm
%that solves a $k$-round LP and then
%% some odd function $\psi: [-1,1] \mapsto [-1,1]$ sampled
%% from a fixed distribution,
%rounds
%every variable using its bias in the LP solution.
% % (as suggested by the LP) using a bias $\psi(p)$.
%Thus
%only the LP biases (and their consistency with local distributions over constraints)
%are {\it useful} towards algorithmic purpose.
%As far as we know, these conclusions were not known before.

\subsubsection*{Equivalence of Approximation Resistance and Strong Approximation Resistance}

In a previous version of this paper \cite{KTW-SAR}, we obtained a characterization of a related notion that we called {\it strong approximation
resistance}. In this notion, an algorithm is considered non-trivial if on a near-satisfiable instance of CSP$(f)$, it finds an assignment such that the
fraction of constraints satisfied is outside the range $[\rho(f)-\Omega(1), \rho(f)+\Omega(1)]$. If such an efficient algorithm exists, the predicate is called
{\it weakly approximable} and strongly approximation resistant otherwise.

We are now able to show that our characterization for strong approximation resistance applies to approximation resistance as well, i.e.
the two notions of resistance are equivalent! In other words, every predicate is either approximable (as opposed to weakly approximable) or is
strongly approximation resistant (as opposed to just approximation resistant), i.e. the best in both the worlds. To emphasize further, the equivalence
means that
for an approximation resistant
predicate $f$, there is a reduction from the Unique Games problem to GapCSP$(f)_{1-o(1), ~\rho(f)+o(1)}$ such that in the NO case,
the instance has an additional property that every assignment to its variables satisfies between $\rho(f)-o(1)$ and $\rho(f)+o(1)$ fraction of the
constraints, i.e. not more and {\it not less} than the threshold $\rho(f)$ by a non-negligible amount. Similarly, all the LP/SDP integrality
gap instances also share this additional property.

We show that a predicate is either approximable or there exists a vanishing measure $\Lambda$. In the latter case, the measure $\Lambda$ is used to
construct $(1-o(1), \rho(f)+o(1))$ integrality gap for the mixed hierarchy which then implies approximation resistance via Raghavendra's result. Though we do not
present it here, it is also possible to use $\Lambda$ to directly construct a dictatorship test and prove approximation resistance (i.e. without
going through the integrality gap construction). In either case, the fact that $\Lambda$ is a vanishing measure ensures that in the {\it soundness
analysis}, the Fourier terms that are potentially responsible for {\it deviating}  from the threshold $\rho(f)$ precisely cancel each other out 
(up to $o(1)$ error). Thus, the integrality gap instance we well as the NO instance in the hardness reduction have the property that for every assignment, the fraction of
satisfied constraints cannot even {\it deviate} from $\rho(f)$ (and therefore cannot exceed $\rho(f)$ either) by a non-negligible amount, yielding
strong approximation resistance.

The notion of strong approximation resistance has been considered in literature before, albeit implicitly. In fact, almost all known proofs of
approximation resistance actually show strong resistance, either implicitly or explicitly,
or by a minor modification or possibly switching from NP-hardness to UG-hardness.
This is because the soundness analysis of
these constructions shows that the Fourier terms that are potentially responsible for  deviating   from the threshold $\rho(f)$ are all bounded by $o(1)$ in {\it magnitude} (our analysis has the novelty that these Fourier terms cancel each other out, without each  term necessarily being $o(1)$ in magnitude).
%, thus showing that it is
%hard to not just {\it exceed} $\rho(f)$, but even to {\it deviate} from it.
The only possible exception we
are aware of is an arity $4$ predicate in \cite{ah2}, Example 8.7 therein, that we mentioned before. The predicate is shown to be approximation
resistant therein and now our result implies that it is also strongly approximation resistant.

We will avoid referring to the notion of strong approximation resistance henceforth and refer an interested reader to the previous version of this paper
\cite{KTW-SAR} for relevant definitions.

\subsection{Overview of the Proof Techniques}
In this section we give an informal overview of the main ideas and techniques used in our results. A
significant ingredient in our results is the Von Neumann min-max theorem for zero-sum games which
 was also used by O'Donnell and Wu \cite{odwu} towards characterizing the {\it
approximability curve} for the MAX-CUT problem.
The game-theoretic and measure-theoretic framework  we develop is likely to find other applications. For instance, it is possible to give an
exposition to Raghavendra's result in our framework, providing a clear explanation in terms of duality of the min-max theorem and
(in our opinion) demystifying the result.

We first focus on the main result in the paper, namely that a predicate $f$ is  approximation
resistant if and only if $f \in \chars$ as in Definition \ref{def:chars}.  Before we begin the
overview, we briefly comment  how the characterization in  Definition \ref{def:chars} comes about and how it makes
sense from the perspective of both the hardness and the algorithmic side. In hindsight, the characterization in Definition \ref{def:chars},
 in terms of the existence of a vanishing measure $\Lambda$, is tailor-made to prove the hardness result: given a vanishing measure
 $\Lambda$, it is straightforward, at least at a conceptual level, to design a dictatorship test and prove approximation resistance
 modulo the UGC. As we said earlier, the vanishing condition precisely ensures that the Fourier terms that are potentially responsible for deviating from
 $\rho(f)$ exactly cancel each other out and % (thus we in fact get strong approximation resistance).
 the novel feature here is that the Fourier terms cancel each other out without each term necessarily being
 $o(1)$ in magnitude. %, as is the case with almost all prior works.
 From the algorithmic side, one would want to show that if a
 vanishing measure does not exist, then the predicate is approximable. However, we do not actually know how to design such an algorithm directly!  Instead, our
 entire argument runs in reverse. We propose a family of (SDP rounding) algorithms and show that either some algorithm in this family works
 or else there exists a vanishing measure. Given a vanishing measure, it is relatively straightforward to prove the hardness result and
 construct integrality gaps (which are equivalent by Raghavendra's result), as mentioned earlier. Our argument is non-constructive on both the algorithmic and the
 hardness side: it yields neither  an explicit algorithm nor  an explicit vanishing measure.

 With the benefit of hindsight, there might be an
 intuitive explanation why non-existence of a vanishing measure implies existence of an algorithm. Each point $\zeta$ in the body ${\cal C}(f)$ corresponds to
 a Gaussian density function $\gamma_d(\zeta, \cdot)$ on ${\mathbb R}^d$.
 Non-existence of a vanishing measure on ${\cal C}(f)$, by duality in an appropriate setting,
 might be interpreted as {\it linear independence} of these Gaussian density functions in the following sense: w.r.t. {\it any} probability measure on ${\cal C}(f)$, the integral of
 the Gaussian density function $\gamma_d(\zeta, \cdot)$ (with Fourier coefficients of the predicate and sign flips thrown in so as to have both positive and negative terms), say a signed density  $\phi(\cdot)$  on ${\mathbb R}^d$, is a non-zero density  with a positive lower bound on its norm. The rounding algorithm is then
 a function $\psi: {\mathbb R}^d \mapsto \{-1,1\}$ that distinguishes $\phi(\cdot)$ from the zero density; the algorithm however has to work for all possible
 $\phi(\cdot)$ simultaneously. One might be able to translate this intuition into a formal proof with an appropriate setting, but we haven't yet investigated this possibility.

We now begin our overview. We
make several simplifying assumptions and use informal mathematically imprecise language
as we proceed (for the sake of a cleaner overview only).
Let $f:\{-1,1\}^k \mapsto \{0,1\}$ be the predicate under
consideration with $\rho(f) = \frac{|f^{-1}(1)|}{2^k}$. We make a simplifying assumption that the predicate $f$ is even, i.e.
$f(-z)=f(z) \ \forall z \in \{-1,1\}^k$. This allows us to assume that the first moments (i.e. ``biases")   $\Ex{z \sim \nu}{z_i}$  are all zero for any distribution $\nu$
supported on $f^{-1}(1)$ and can be safely ignored. \footnote{See \cite{ak} for elaboration
  where the same assumption is used.}
Therefore we let the polytope $\cal{C}(f)$ to be the set of all ${k}\choose{2}$-dimensional
{\it second moments vectors} $\zeta(\nu) = \big( \Ex{z \sim \nu}{z_i z_j} \ | \ 1 \leq i < j \leq k
\big)$ over all distributions $\nu$ supported on
$f^{-1}(1)$. Our main concern is whether there is an efficient algorithm for CSP$(f)$
that achieves a non-trivial approximation, i.e. on an $1-o(1))$ satisfiable instance obtains an assignment
such that the fraction of satisfied constraints is at least $\rho(f)+\Omega(1)$. We make the simplifying assumption that the CSP instance is in fact perfectly
satisfiable. This implies
that the basic relaxation yields, for every constraint $C$ that depends on variables say $x_1, \ldots, x_k$,
a distribution $\nu(C)$ over the set of satisfying assignments $f^{-1}(1)$ and unit vectors $\vone{1},
\ldots, \vone{k}$ such that
$\ip{\vone{i}, \vone{j}} = \Ex{z \sim \nu}{z_i \cdot z_j}$. As noted, $\zeta(\nu(C))$ then is a
$\binom{k}{2}$-dimensional vector of the second moments (which equal $\ip{\vone{i}, \vone{j}}$).
The uniform distribution over the vectors
$\zeta(\nu(C))$ over all constraints $C$ is then a probability measure $\lambda$ on $\cal{C}(f)$. We
regard the measure
$\lambda$ as essentially representing the given CSP instance (a priori, we seem to be losing information by
ignoring the topology of the instance, but as we will see this doesn't matter).

Note that in the relaxed solution, the vector assignment is global in the sense that the vector
assigned to each CSP variable is fixed, independent of the constraint $C$ in which the variable participates in
whereas the distribution $\nu(C)$ is local in the sense that it depends on the specific constraint $C$.

Our main idea, as hinted to before, is to propose a family of algorithms based on ``$d$-dimensional roundings"
of the SDP solution for $d=k+1$
and to show that either one such algorithm
achieves a non-trivial approximation  or else the polytope $\cal{C}(f)$
supports a probability measure $\Lambda$ as in Definition \ref{def:chars} (note again  that we are ignoring
the first moments).
In the latter case, the existence and symmetry of $\Lambda$
leads naturally to a  $(1-o(1), \rho(f) + o(1))$ integrality gap for the basic relaxation
(and therefore
mixed hierarchy)
and a UG-hardness result for GapCSP$(f)_{1-o(1), \rho(f) + o(1)}$, showing that the predicate
is  approximation resistant.

The proposed family of $d$-dimensional roundings is easy to describe: any function
$\psi: \R^d \mapsto \{-1,1\}$ serves as a candidate rounding algorithm where the SDP vectors
$\{\vone{i}\}$ are projected onto a random $d$-dimensional subspace inducing
$\vone{i} \mapsto \y_i \in \R^d$ and then the $i^{th}$ variable is assigned a boolean value $\psi(\y_i)$.
From the algorithmic viewpoint, one seeks a rounding function $\psi$ (more generally a distribution over $\psi$)
such that its ``performance" on {\it every} instance $\lambda$\footnote{We recall again that for any CSP instance,
$\lambda$ is the uniform distribution over $\zeta(\nu(C))$ over all constraints $C$ and thus a probability measure on
$\cal{C}(f)$. The measure $\lambda$ now represents the whole instance.}
significantly exceeds $\rho(f)$ (in average, if a distribution over $\psi$ is used).
From the hardness viewpoint, a natural goal then
would be to come up with a ``hard-to-round measure" $\lambda$ on $\cal{C}(f)$ such that
the ``performance" of {\it every} rounding function $\psi$ is at most $\rho(f) + o(1)$.

These considerations lead naturally to a two-player zero-sum game between \Dev, the ``hardness
player" and \Ang, the ``algorithm player" (we view  \Dev as the row player and \Ang as the column
player). The pure strategies of \Dev are the probability measures  $\lambda$ on
$\cal{C}(f)$ to be rounded and the pure strategies of \Ang are the rounding functions $\psi: \R^d
\mapsto \{-1,1\}$.
The payoff to \Ang when the two players play $(\lambda, \psi)$ respectively is the ``advantage over $\rho(f)$"
achieved by rounding $\lambda$ using $\psi$. More precisely, consider the scenario where the set of local distributions on CSP constraints is represented by the measure $\lambda$. The local distribution on a
randomly selected constraint is a sample $\zeta \sim \lambda$ along with vectors $\vone{1},\ldots,\vone{k}$ whose
pairwise inner products match $\zeta$.
During the rounding process, the vectors $\vone{1}, \ldots, \vone{k}$ are projected onto a random
$d$-dimensional subspace,  generating a sequence of $k$
points $\y_1, \ldots, \y_k \in \R^d$  that  are standard $d$-dimensional Gaussians with correlations
$\zeta$. The CSP variables are then rounded to boolean values $\psi(\y_1),\ldots,\psi(\y_k)$. Whether these
values satisfy the constraint or not is determined by plugging them in the Fourier representation of the predicate $f$. The
``advantage over $\rho(f)$" is precisely this Fourier expression without the constant term (which is
$\rho(f)$). Given this intuition, we define the payoff to \Ang as the expression:
\Mnote{Strictly speaking the above is slightly incorrect since if we obtain $\y_1,\ldots,\y_k$ by
  simply projecting onto a random subspace, then they will have mean zero. I can explain the process
$\cal{N}_d(\zeta)$ more precisely in the preliminaries.}
\begin{equation}\label{payoff-overview}
\pay(\lambda,\psi)  ~:=~  \ExpOp_{\zeta \sim \lambda }
~\Ex{\y_1, \ldots, \y_k \sim \cal{N}_d(\zeta)}{\sum_{S \neq \emptyset} \HF(S) \cdot \prod_{i\in S}\psi(\y_i)},
\end{equation}
where $\cal{N}_d(\zeta)$ denotes a sequence of $k$ standard $d$-dimensional Gaussians with correlations $\zeta$.
%\footnote{One can also place
%the absolute value signs outside of the whole expression,
%but we don't have a good answer other than saying that this works.}
We apply Von Neumann's min-max theorem and conclude that there exists a number $L$, namely the ``value" of the game, a mixed equilibrium strategy $\Gamma$ (a distribution over $\psi$) for \Ang and an equilibrium strategy $\Lambda$ (a pure one as we will observe!) for \Dev. Actually
Von Neumann's theorem applies only to games where the sets of strategies for both players are finite, but we ignore this issue
for now. Depending on whether the value of the game $L$ is strictly positive or zero
(it is non-negative since \Ang can always choose a {\it random function} $\psi$ and achieve a zero payoff), we get the ``dichotomy" that the
predicate $f$ is  approximable or  approximation resistant (modulo UGC).

The conclusion when $L > 0$ is easy: in this case \Ang has a mixed strategy $\Gamma$ such that her payoff (expected over $\Gamma$) is
at least $L$ for {\it every} pure
strategy $\lambda$ of \Dev. This is same as saying that if a rounding function $\psi \sim \Gamma$ is sampled and then used to round
the relaxed solution, it achieves an advanateg $L$ over $\rho(f)$ for every CSP instance $\lambda$.

The conclusion when $L = 0$ is more subtle: in this case in general \Dev has a mixed strategy,
 say $\cal{D}$,
 such that for {\it every} pure
strategy $\psi$ of \Ang, her expected payoff (expected over $\lambda \sim \cal{D}$)
is at most  zero. We observe that \Dev may replace his mixed strategy $\cal{D}$ by a pure strategy $\Lambda$. Noting that
$\cal{D}$ is a distribution over  measures $\lambda$, we let $\Lambda$ be the single {\it averaged} measure
informally written as $\Lambda :=
\ExpOp_{\lambda \sim \cal{D}} [\lambda]$. Thus the expectations over $\lambda \sim \cal{D}$ and $\zeta \sim \lambda$ may be merged into the expectation
over $\zeta \sim \Lambda$.
%Since for every $\psi$,
%Expression \eqref{payoff-overview} is supposed to be zero averaged over $\lambda \sim \cal{D}$,
%it must also be zero for $\Lambda$ itself!
 We may therefore conclude that for the measure $\Lambda$
over $\cal{C}(f)$, for every $\psi: \R^d \mapsto \{-1,1\}$:
\begin{equation} \label{eqn:Lambda}
\ExpOp_{\zeta \sim \Lambda}  ~\Ex{\bfy_1, \ldots, \bfy_k \sim \cal{N}_d(\zeta)}{\sum_{S \neq \emptyset}
  \HF(S) \cdot \prod_{i\in S}\psi(\bfy_i)} ~\leq ~ 0.
\end{equation}
Now we view this
expression as a multi-linear  polynomial in (uncountable number of) variables $\{ \psi(\y) \ | \ y \in \R^d\}$. We observe that if a multi-linear polynomial
in  finitely many $\{-1,1\}$-valued variables with no constant term is upper bounded by zero, then it must be identically zero (see Lemma \ref{multi-linear:lem}). We pretend, for now, that the same
conclusion holds to the ``polynomial" above and hence that  it is identically zero and we may equate every ``coefficient" of this polynomial to zero.

Fix any $1 \leq t \leq k$.
For every $\y_1, \ldots, \y_t \in \R^d$, we are interested in the coefficient of the monomial
$\prod_{i=1}^t \psi(\y_i)$. Firstly,
this coefficient can arise from precisely the sets $S$ with $|S|=t$. Secondly, for a fixed set $S, |S|=t$, the coefficient is
really the joint density of  $t$ standard $d$-dimensional Gaussians with correlations $\zeta_S$ at the sequence
$(\y_1,\ldots,\y_t)$, where $\zeta_S$ is same as $\zeta$ restricted to indices in $S$.  Thirdly, for any permutation $\pi: [t] \mapsto [t]$, we must consider all sequences $(\y_{\pi(1)}, \ldots, \y_{\pi(t)})$
and add up their coefficients (i.e. Gaussian densities) since they all correspond to the same
monomial  $\prod_{i=1}^t \psi(\y_i)$. Finally,
we did not mention this so far, but we need to allow only odd rounding functions $\psi$,
i.e. $\psi(-\y)= -\psi(\y)$, to account for
the issue of variable negations in CSPs. This has the effect that the monomials $\prod_{i=1}^t
\psi(b_i \cdot \y_i)$ are same as
$\prod_{i=1}^t b_i \cdot \prod_{i=1}^t \psi(\y_i)$ for a choice of signs $b_i \in \{-1,1\}$, and
hence their coefficients (i.e. Gaussian
densities) must be added up together. With all these considerations, the coefficient of the monomial
$\prod_{i=1}^t \psi(\y_i)$ can be written as:
$$ \Ex{\zeta \sim \Lambda}{\sum_{S, |S|=t} \ \ \sum_{\pi: [t] \mapsto [t]} \ \  \sum_{b \in
    \{-1,1\}^t}  \ \HF(S) \cdot \left(\prod_{i = 1}^t b_i\right)
   \gamma_{t,d}\big( (\y_1,\ldots,\y_t), \zeta_{S,\pi,b}\big)}. $$
Here $\zeta_{S,\pi,b}$ is the sequence of correlations between the indices in $S$ after accounting
for the permutation of indices according to $\pi$ and the sign-flips according to $b \in \{-1,1\}^t$. Also
$\gamma_{t,d} \big( (\y_1,\ldots,\y_t), \xi \big)$ is the joint density of
$t$ standard $d$-dimensional Gaussians with
correlations $\xi$.
Defining the ``signed measure" $\Lambda^{(t)}$ as in
Equation \eqref{intro:main:eqn}, the conclusion that the above coefficient is zero (for every $(\y_1,\ldots,\y_t)$), can be written as:
$$\forall \y_1,\ldots,\y_t \in \R^d, \ \ \  \int  \gamma_{t, d}\big( (\y_1,\ldots,\y_t), \xi \big) \ d\Lambda^{(t)}(\xi) = 0.$$
In words, w.r.t. the signed measure $\Lambda^{(t)}$ on $[-1,1]^{{t}\choose{2}}$ (corresponding to
all possible correlation vectors
between $t$ standard $1$-dimensional Gaussians), the integral of every function  $\gamma_{t,d}(
(\y_1,\ldots,\y_t), \cdot )$ vanishes
(there is one such function for every fixed choice of $(\y_1,\ldots, \y_t)$). The class of these
functions is rich enough that, after jumping through several hoops, we are able to conclude that the signed measure $\Lambda^{(t)}$ itself
must identically vanish.

This proves the existence of the measure $\Lambda$ as in Definition \ref{def:chars}. After this, the
construction of the
$(1-o(1), \rho(f)+ o(1))$ integrality gap for the CSP is obtained by generalizing the
construction for MAX-CUT due to Feige and
Schechtman \cite{fsch}. We describe the construction in the continuous setting and ignore the
discretization step here.
The variables in the CSP instance correspond to points in $\R^N$ for a high enough
dimension $N$ and the variables for $\y$ and $-\y$ are designated as negations of each other. The constraints of the
CSP are defined by sampling $\zeta \sim \Lambda$ and then sampling $k$ Gaussian points $\y_1,\ldots,
\y_k \in \R^N$ with correlations
 $\zeta$  and placing a constraint on these variables. For the completeness part, one observes that
 for large $N$ the space $\R^N$ with
the Gaussian measure is (up to $o(1)$ errors) same as the unit sphere ${\mathbb S}^{N-1}$ towards our purpose and we may
assume that all the CSP variables lie on the unit sphere.
Each point on the sphere is assigned a vector that is itself and for every constraint, the local distribution equals
$\nu$ if $\zeta= \zeta(\nu)$ is used towards that constraint. For the soundness part, an assignment to the CSP corresponds to
a function $\psi: \R^N \mapsto \{-1,1\}$ and the ``advantage over $\rho(f)$" is precisely the expression on the l.h.s. of Equation \eqref{eqn:Lambda},
if $\y_1,\ldots, \y_k$ were chosen from $\R^N$ instead of $\R^d$ ($d=k+1$ therein). The symmetry properties of
$\Lambda$ (i.e. that the signed measure $\Lambda^{(t)}$ vanishes for every $1 \leq t \leq k$) ensure  that this expression
vanishes identically and hence no CSP assignment can exceed or even deviate from $\rho(f)$.
We would like to emphasize here that the existence of $\Lambda$ was deduced only assuming that no $(k+1)$-dimensional rounding
exceeds $\rho(f)$, but once the existence of $\Lambda$ is established, it automatically implies that no higher
dimensional rounding can deviate from $\rho(f)$.

Once the integrality gap is established, the  UG-hardness
of GapCSP$(f)_{(1-o(1), \rho(f) + o(1))}$
 follows  automatically from the general result of Raghavendra and the same
integrality gap for a super-constant number of rounds of the mixed  hierarchy follows
 automatically from the general results of Raghavendra and Steurer~\cite{raghsteu2},
and Khot and Saket \cite{ks09}.

As we said, this is a simplified and informal view and we actually need
to work around all the simplifying assumptions we made, formalize all the arguments, and
address many issues that we hid under the carpet, e.g.
setting $d= k+1$ and the reason say $d=1$ does not work, handling the first moments,
handling the possibility that a Gaussian density is degenerate, etc.
Also, we cannot apply Von Neumann's min-max theorem to infinite
games. In principle, one might be able to use min-max theorems for infinite games such as Glicksberg's theorem, but then one
has to ensure that the strategy spaces are compact. Instead, we find it easier to work with a sequence of finite approximations
to the infinite game and then use limiting arguments everywhere (this is easier said than done and
this is where much of the work lies in).

Another tricky issue is to ensure that the polynomial, obtained as a  discretized  finite analogue of the expression
\eqref{eqn:Lambda}, stays multi-linear. We ensure this by
modifying the $\pay(\cdot,\cdot)$ function so as to delete the non-multi-linear terms from the very start. In general, it seems
difficult to argue that the norm on the terms so deleted is negligible compared to the norm on the linear combination of
the remaining terms (which might suffer heavily due to cancellations). We observe however that the norm on the deleted terms
needs to be negligible only compared to the value of the game $L$ in the case $L >0$ and this  is indeed the case when the discretization is fine enough.
The reason is that deleting certain terms changes the $\pay(\cdot,\cdot)$ function by a corresponding amount, but as long as this amount is negligible compared to
$L$, in the case $L > 0$, the algorithm player has a strategy with value at least say $L/2$ even w.r.t. the original payoff function
and hence still gets an  advantage of $L/2$ over $\rho(f)$.
\footnote{We missed this trick before and were able to get a characterization of only
 strong approximation resistance in the previous version of the paper \cite{KTW-SAR}.}

\subsubsection*{Approximation Resistance for LP Hierarchies}
Now we give an overview of the characterization of  approximation resistance (i.e.
Definition \ref{def:charlp})
for a super-constant number of rounds of Sherali-Adams LP. We proceed along a similar line
as earlier with one difference: we work with a different
body $\tilde{\cal{C}}(f)$ instead of $\cal{C}(f)$.

In the LP case, the second moments are not available at all and
the first moments are all one has. We will nevertheless pretend that the second moments are available
by using their {\it dummy} setting. For any distribution $\nu$ supported on $f^{-1}(1)$,
let the vector $\zeta = \zeta(\nu)$ consist of the $k$ first moments $\zeta(i) = \Ex{z\sim \nu}{z_i}$ and in addition, {\it dummy} second moments corresponding to those of $k$ {\it independent}
unit $\ell_2$-norm Gaussians $g_1,\ldots,g_k$ with the given first moments, i.e.
$\ExpOp[g_i] = \zeta(i)$ and $\ExpOp[g_i^2]=1$. The body $\tilde{\cal{C}}(f)$ is defined as the set of all
vectors $\zeta(\nu)$ over all distributions $\nu$ supported on $f^{-1}(1)$. Note that
$\tilde{\cal{C}}(f)$ is different than the polytope $\cal{C}(f)$ and not necessarily convex
(we never used convexity), but its
projection onto the first $k$ co-ordinates is the same as that of $\cal{C}(f)$, namely
$\cal{C}^*(f)$ as in Definition \ref{def:charlp}.

Once the polytope $\cal{C}(f)$ is replaced by the body $\tilde{\cal{C}}(f)$,
our argument proceeds as before. Note that since the second moments reflect independent Gaussians,
our rounding is really using only the first moments, as ought to be the case with LPs.
We conclude that either the predicate is approximable or there is a probability measure
$\Lambda$ on $\tilde{\cal{C}}(f)$ that satisfies characterization in Definition \ref{def:chars}.
Projecting $\Lambda$ onto the first $k$ co-ordinates gives a measure $\Lambda^*$ on $\cal{C}^*(f)$
satisfying the characterization in
Definition \ref{def:charlp}.

Once the existence of $\Lambda^*$ is established, we proceed to constructing the
 $(1-o(1), \rho(f) + o(1))$ integrality gap in the Sherali-Adams hierarchy. This step however
turns out to be more involved than before since general results as in \cite{ragh,raghsteu2,ks09} are not available
in the LP setting. Instead, we are able to rework the MAX-CUT construction of de la Vega and Kenyon \cite{delavega} for any predicate $f \in \charlp$.

An intuitive way of looking at the construction is as follows. The variables of the CSP are
points in the interval $[-1,1]$ and the variables for $x$ and $-x$ are negations of each
other (called {\it folding}). Constraints are defined by sampling $\zeta \sim \Lambda^*$ and
then placing the constraint on variables $(\zeta(1),\ldots,\zeta(k))$. The local distribution
for this constraint is $\nu$ such that $\zeta = \zeta(\nu)$. The LP-bias of a variable $x$ is
$x$ itself. The vanishing condition in Definition \ref{def:charlp} implies
that any (measurable) $\{-1,1\}$-assignment to this CSP instance satisfies exactly
$\rho(f)$ fraction (measure) of the constraints. This conclusion also holds for
$[-1,1]$-valued assignments appropriately interpreted.

This continuous instance only has a {\it basic} LP solution, i.e. the local distributions
are defined only for constraints. We now construct the actual instance as follows.
We discretize the interval $[-1,1]$ by picking equally spaced points $x_1,\ldots,x_s$ with
fine enough granularity (and ensuring that a point and its negation are both included
and are {\it folded}).
Each variable $x_i$ is now {\it blown up} into a block of $n/s$ variables for a
large $n$ (so the total number of variables is $n$).
Whenever a constraint is generated in the continuous setting by sampling
$\zeta \sim \Lambda^*$, we first round $\zeta(j)$ to nearest $x_{i_j}$ and then the
constraint is
actually placed on randomly chosen variables
from blocks corresponding to $x_{i_1},\ldots, x_{i_k}$ respectively. This is the way
{\it one} constraint is randomly introduced and the process is repeated independently
$m$ times for $m \gg n$. This defines the CSP instance as a $k$-uniform hyper-graph.
By deleting a small fraction of the constraints, one ensures that the hyper-graph has
super-constant girth. Finally, de la Vega and Kenyon \cite{delavega} construction is reworked
to construct local distributions for all $r$-sets of variables, i.e. for the $r$-round
Sherali-Adams LP. Our presentation is somewhat different than that in \cite{delavega}: we find
it easier to first construct a nearly correct LP solution and then correct it as in
\cite{raghsteu2, ks09}.

One interesting and novel feature of our construction is how the CSP instance is constructed and how
the ``soundness" is proved as opposed to a standard construction of random CSPs.

A standard construction, in one step, generates a constraint by uniformly selecting a
$k$-subset of variables and then randomly selecting the {\it polarities} (i.e. whether
a variable occurs in a negated form or not). This step is then repeated independently
to generate $m \gg n$ constraints. Since the polarities are randomly chosen in each
step, for any fixed global assignment, the probability that the assignment satisfies
the constraint is precisely $\rho(f)$, and then one uses the
 Chernoff bound and the union bound to conclude that w.h.p. every global assignment to the instance
 satisfies between $\rho(f) \pm o(1)$ fraction of the constraints.

In our case, the one step of generating a constraint is different. In particular, the
$k$-subset of variables chosen is not necessarily uniformly random (it depends on $\Lambda^*$
since $\zeta \sim \Lambda^*$) and the polarities are not necessarily random either
(they depend on signs of $\zeta(1),\ldots,\zeta(k)$ due to {\it folding}). However it is
still true that for any fixed global assignment, the probability that the assignment
satisfies the constraint is precisely $\rho(f)$ (up to $o(1)$ errors introduced by
discretization)! This property is simply inherited from the continuous setting by
viewing the global assignment as a function $\psi: \{x_1,\ldots,x_s\} \mapsto [-1,1]$
where $\psi(x_i)$ is the average of the global values to variables in block $x_i$! This
concludes our overview.

\section{Preliminaries and Our Results}\label{sec:prelims}

In this section, we present formal definitions and statements of our results and a preliminary background
on mathematical tools used.

\subsection{Constraint Satisfaction Problems}

%For an instance $\Phi$ of  \maxkcsp, we denote the variables by $\{x_1, \ldots, x_n\}$ and
%the constraints by $C_1, \ldots, C_m$. Each constraint is a function of the form $C_i : \pmone^{T_i}
%\rightarrow \B$  depending only on the values of the variables in the ordered tuple $T_i$ with
%$|T_i| \leq k$.

%We shall prove results for constraint satisfaction problems where every constraint is specified by the same %Boolean
%predicate $f: \pmone^k \rightarrow \B$.
%We denote the set of assignments for which the predicate evaluates to 1 by $f^{-1}(1)$.
%A  CSP instance for such a problem is a collection of constraints of the form of $f$ applied to
%$k$-tuples of \emph{literals}. For a variable $x$ with domain $\pmone$, we take a literal to be
%$x\cdot b$ for any $b \in \pmone$. More formally,

\begin{definition}
For a predicate $f : \pmone^k \rightarrow \{0,1\}$, an instance
$\Phi$ of \maxkcsp$(f)$ consists of a set of variables $\{x_1,\ldots,x_n\}$ and a set of constraints $C_1, \ldots, C_m$ where each
constraint $C_i$ is over a $k$-tuple of variables $\{x_{i_1}, \ldots, x_{i_k}\}$ and is of the
form
\[
C_i ~\equiv~ f(x_{i_1} \cdot b_{i_1}, \ldots, x_{i_k} \cdot b_{i_k})
\]
for some $b_{i_1}, \ldots, b_{i_k} \in \pmone$.
For an assignment $A: \{x_1,\ldots,x_n\} \mapsto \pmone$, let $\sat(A)$ denote the fraction of
constraints satisfied by $A$. The instance is called $\alpha$-satisfiable if there exists an assignment
$A$ such that $\sat(A) \geq \alpha$. The
maximum fraction of constraints that can be simultaneously satisfied is denoted by $\opt(\Phi)$, i.e.
 $$ \opt(\Phi) = \max_{A: \{x_1,\ldots,x_n\} \mapsto \pmone}  \sat(A). $$
The density of the predicate is $\rho(f) = \frac{|f^{-1}(1)|}{2^k}$.
\end{definition}

For a constraint $C$ of the above form, we use $x_C$ to denote the tuple of variables $(x_{i_1},
\ldots, x_{i_k})$ and $b_C$ to denote the tuple of bits $(b_{i_1}, \ldots, b_{i_k})$. We then write
the constraint as $f(x_C \cdot b_C)$. We also denote by $S_C$ the set of indices
$\{i_1,\ldots, i_k\}$ of the variables participating in the constraint $C$.

\begin{definition}
A predicate  $f : \pmone^k \rightarrow \{0,1\}$ is called {\deffont approximable} if there exists a constant $\eps > 0$ and a polynomial time algorithm, possibly randomized, that given an $(1-\eps)$-satisfiable
instance of \maxkcsp($f$), outputs an assignment $A$ such that  $\Ex{A}{~\sat(A)~} \geq \rho(f)+\eps$.
Here the expectation is over the randomness used by the algorithm.
%The predicate is called {\deffont weakly approximable} if the output of the algorithm deviates from $\rho(f)$ in expectation, i.e.
%$\Ex{A}{~\left| \sat(A) -\rho(f)  \right|~} \geq \eps$.
\end{definition}

%We define the notions of approximation resistance and strong approximation resistance.
Towards defining the notion of approximation resistance,
 it is convenient to define the gap version of the CSP. Though the gap version can be defined
w.r.t. any {\it gap location}, we do so only for the location that is of interest to us, namely $1-o(1)$
versus $\rho(f) + o(1)$. We say that a decision problem is UG-hard if there is polynomial time
reduction from the Unique Games Problem \cite{khot} to the problem under consideration (we will not
be directly concerned with the
Unique Games Problem and the Conjecture; hence their discussion is deferred to the end of the preliminaries section).

\begin{definition} Let $\eps > 0$ be a constant.

Let {\sf GapCSP}$(f)_{1-\eps, ~\rho(f)+\eps}$ denote the promise
version of \maxkcsp$(f)$ where the given instance $\Phi$ is promised to have either $~\opt(\Phi) \geq
1-\eps$ ~or ~$\opt(\Phi) \leq \rho(f)+\eps$. The predicate is called {\deffont approximation resistant}
if for every $\eps > 0$, {\sf GapCSP}$(f)_{1-\eps, ~\rho(f)+\eps}$ is UG-hard.

%Let {\sf GapCSP}$(f)_{1-\eps, ~\rho(f) \pm \eps}$ denote the promise
%version of \maxkcsp$(f)$ where the given instance $\Phi$ is promised to have either ~$\opt(\Phi) \geq
%1-\eps$ ~or that for {\deffont every} assignment $A$, $\left| \sat(A) -\rho(f)  \right| \leq \eps$. The predicate is called {\deffont strongly approximation resistant}
%if for every $\eps > 0$, {\sf GapCSP}$(f)_{1-\eps, ~\rho(f) \pm \eps}$ is UG-hard.

\end{definition}

\subsection{The LP and SDP Relaxations for Constraint Satisfaction Problems}\label{subsec:SDP}
Below we present three LP and SDP relaxations for the $\maxkcsp(f)$ problem that are relevant in this paper:
the Sherali-Adams LP relaxation,
mixed LP/SDP relaxation and finally the {\it basic relaxation}.

%as considered by Raghavendra
%\cite{Raghavendra08}. The relaxation is weaker than the one obtained by augmenting the relaxation
%obtained by $k$ levels of the Sherali-Adams hierarchy with the SDP constraints.

% as it is obtained by applying a level-$t$ Sherali-Adams
% tightening to the basic SDP formulation of some instance $\Phi$ of \maxkcsp. A well known fact states that
% the level-$n$ Sherali-Adams tightening (even starting from a linear program) provides a
% perfect formulation \ie the integrality gap is 1 (see \cite{SA90} or \cite{Lau03} for a proof).

We start with the $r$-round Sherali-Adams relaxation.
The intuition behind it is the following. Note that an integer
solution to the problem can be given by an assignment $A: [n] \to \pmone$.  Using this, we can
define $\{0,1\}$-valued variables $\vartwo{S}{\alpha}$ for each $S \subseteq [n], 1 \leq |S| \leq r$
and $\alpha \in \pmone^S$,
with the intended solution $\vartwo{S}{\alpha} = 1$ if $A(S) = \alpha$ and 0 otherwise.
We also introduce a variable $\vartwoempty$, which equals 1.
We relax the integer program and allow variables to take real values in $[0,1]$.
Now the variables $\{ \vartwo{S}{\alpha} \}_{\alpha \in \{-1,1\}^k}$ give a probability distribution
over assignments to $S$. We can enforce consistency between these {\it local} distributions
by requiring that for $T \subseteq S$, the distribution over assignments to $S$, when marginalized
to $T$, is precisely the distribution over assignments to $T$.
 The relaxation is shown in Figure \ref{fig:SA-lp}.

\begin{figure}[ht]
\hrule
\vline
\begin{minipage}[t]{0.99\linewidth}
\vspace{-5 pt}
{\small
\begin{align*}
\mbox{maximize} &~~
\Ex{C \in \Phi}{\sum_{\alpha \in \{-1,1\}^k} f(\alpha \cdot b_C) \cdot \vartwo{S_C}{\alpha} }&  \\
\mbox{subject to} \\ %\qquad \qquad \qquad \qquad ~~
%
%\mydot{\vtwo{i}{1}}{\vtwo{i}{-1}} &~=~ 0 &\forall i \in [n] \\
%
%\norm{\vtwoempty} &~=~ 1 \\
%
%\vtwo{i}{1} + \vtwo{i}{-1} &~=~ \vtwoempty &\forall i \in [n] \\
%
\sum_{\alpha \in \pmone^{S} \atop \alpha|_T = \beta}\vartwo{S}{\alpha} &~=~
\vartwo{T}{\beta}  &\forall T \subseteq S \subseteq [n], |S|\leq r, ~\forall \beta \in \pmone^T \\
 \vartwo{S}{\alpha} &~\geq~ 0
  & \forall S \subseteq [n], |S| \leq r,  ~\forall \alpha \in \pmone^{S} \\
   \vartwoempty & ~=~  1    \\
\end{align*}}
\vspace{-14 pt}
\end{minipage}
\hfill\vline
\hrule
\caption{$r$-round Sherali-Adams LP for \maxkcsp($f$)}
\label{fig:SA-lp}
\end{figure}

We can further strengthen the integer program by adding the quadratic constraints
\[
\vartwo{\{i_1, i_2\}}{(b_1,b_2)}
~=~ \vartwo{\{i_1\}}{b_1} \cdot \vartwo{\{i_2\}}{b_2} \mper
\]
As solving quadratic programs is NP-hard we then relax
these quadratic constraints to the existence
of vectors $\vtwo{i}{b}$ and a unit vector $\vtwoempty$, and impose the above constraints on inner
products of the corresponding vectors. Adding these SDP variables and constraints to the $r$-round
Sherali-Adams LP as above yields the $r$-round mixed relaxation as in Figure \ref{fig:mixed}.

\begin{figure}[ht]
\hrule
\vline
\begin{minipage}[t]{0.99\linewidth}
\vspace{-5 pt}
{\small
\begin{align*}
\mbox{maximize}  &
~~\Ex{C \in \Phi}{\sum_{\alpha \in \{-1,1\}^k} f(\alpha \cdot b_C) \cdot \vartwo{S_C}{\alpha} } \\
\mbox{subject to} \\ %\qquad \qquad \qquad \qquad ~~
\mydot{\vtwo{i}{1}}{\vtwo{i}{-1}} &~=~ 0  &\forall i \in [n] \\
\vtwo{i}{1} + \vtwo{i}{-1} &~=~ \vtwoempty &\forall i \in [n] \\
%\mydot{\vtwo{i}{1}}{\vtwo{i}{-1}} &~=~ 0 &\forall i \in [n] \\
%
%\norm{\vtwoempty} &~=~ 1 \\
%
%\vtwo{i}{1} + \vtwo{i}{-1} &~=~ \vtwoempty &\forall i \in [n] \\
%
\vartwo{\{i_1,i_2\}}{(b_1,b_2)}  &~=~
\mydot{\vtwo{i_1}{b_1}}{\vtwo{i_2}{b_2}} &\forall i_1 \neq i_2 \in [n], b_1,b_2 \in
\pmone\\
\sum_{\alpha \in \pmone^{S} \atop \alpha|_T = \beta}\vartwo{S}{\alpha} &~=~
\vartwo{T}{\beta}  &\forall T \subseteq S \subseteq [n], |S|\leq r, ~\forall \beta \in \pmone^T \\
 \vartwo{S}{\alpha} &~\geq~ 0
  & \forall S \subseteq [n], |S| \leq r,  ~\forall \alpha \in \pmone^{S} \\
 \norm{\vtwoempty}^2 ~=~ \vartwoempty & ~=~  1 &    \\
\end{align*}}
\vspace{-14 pt}
\end{minipage}
\hfill\vline
\hrule
\caption{$r$-round Mixed Relaxation for \maxkcsp($f$)}
\label{fig:mixed}
\end{figure}

Finally, the basic relaxation is a reduced form of the above mixed relaxation where only
those variables $\vartwo{S}{\alpha}$ are included for which $S  = S_C$ is the set of CSP variables
for some constraint $C$. The consistency constraints between pairs of vectors are included
only for those pairs that occur inside some constraint. The relaxation (after a minor rewriting)
is shown in Figure
\ref{fig:basic-sdp}.

\begin{figure}[ht]
\hrule
\vline
\begin{minipage}[t]{0.99\linewidth}
\vspace{-5 pt}
{\small
\begin{align*}
\mbox{maximize} &~~
\Ex{C \in \Phi}{\sum_{\alpha \in \{-1,1\}^k} f(\alpha \cdot b_C) \cdot \vartwo{S_C}{\alpha} }&  \\
\mbox{subject to}\\ % \qquad \qquad \qquad \qquad ~~
\mydot{\vtwo{i}{1}}{\vtwo{i}{-1}} &~=~ 0 &\forall i \in [n] \\
\vtwo{i}{1} + \vtwo{i}{-1} &~=~ \vtwoempty &\forall i \in [n] \\
\norm{\vtwoempty}^2 &~=~ 1 \\
%
%\vtwo{i}{1} + \vtwo{i}{-1} &~=~ \vtwoempty &\forall i \in [n] \\
%
\sum_{\alpha \in \pmone^{S_C} \atop \alpha(i_1) = b_1, \alpha(i_2) = b_2}\vartwo{S_C}{\alpha} &~=~
\mydot{\vtwo{i_1}{b_1}}{\vtwo{i_2}{b_2}} &\forall C \in \Phi, i_1 \neq i_2 \in S_C, b_1,b_2 \in
\pmone\\
 \vartwo{S_C}{\alpha} &~\geq~ 0
  & \forall C \in \Phi, ~\forall \alpha \in \pmone^{S_C}
\end{align*}}
\vspace{-14 pt}
\end{minipage}
\hfill\vline
\hrule
\caption{Basic Relaxation for \maxkcsp($f$)}
\label{fig:basic-sdp}
\end{figure}

For an LP/SDP relaxation of \maxkcsp, and for a given instance $\Phi$ of the problem, we denote by
$\sdpopt(\Phi)$ the LP/SDP (fractional) optimum.
For the particular instance $\Phi$, the integrality gap is defined as $\sdpopt(\Phi)/\opt(\Phi)$.
The integrality gap of the relaxation is the supremum of integrality gaps over all instances. The
integrality gap thus defined is in terms of a ratio whereas we are concerned with the specific {\it gap
location} $1-o(1)$ versus $\rho(f)+o(1)$. % and also with the {\it strong integrality gap} as defined below.

\begin{definition} Let $\eps > 0$ be a constant.

A relaxation is said to have a $(1- \eps, \rho(f)+\eps)$-integrality gap if there exists a
CSP instance $\Phi$ such that $\sdpopt(\Phi) \geq 1-\eps$ and $\opt(\Phi) \leq \rho(f)+\eps$.

%The relaxation is said to have a {\deffont strong}
%$(1- \eps, \rho(f)\pm \eps)$-integrality gap if there exists a
%CSP instance $\Phi$ such that $\sdpopt(\Phi) \geq 1-\eps$ and for {\deffont every} assignment
%$A$ to the instance, $\left| \sat(A) - \rho(f) \right| \leq \eps$.
\end{definition}

We will use known results showing that the integrality gap for the basic relaxation as in Figure
\ref{fig:basic-sdp} implies a UG-hardness result as well as integrality gap for the mixed relaxation
as in Figure \ref{fig:mixed} for a super-constant number of rounds, while essentially
preserving the gap. The first implication is by Raghavendra \cite{ragh} and the second by Raghavendra and Steurer \cite{raghsteu2} and Khot and Saket \cite{ks09}. We state these results in a form suitable for our purpose.
%(in particular making an additional observation that the results also apply in our
%setting of the {\it strong} integrality gap).

%Raghavendra \cite{} shows that integrality gap for the basic relaxation
%implies a UG-hardness result while preserving
%the gap. Observing that his result applies also to strong gaps we have:

\begin{theorem} {\rm \cite{ragh}} Let $\eps > 0$ be an arbitrarily small constant.

If the basic relaxation as in Figure \ref{fig:basic-sdp} has a $(1- \eps, \rho(f)+\eps)$-integrality gap,
then {\sf GapCSP}$(f)_{1-2\eps, \rho(f)+2\eps}$ is UG-hard.
%Moreover, if the gap is a strong $(1- \eps, \rho(f) \pm \eps)$-gap, then
%{\sf GapCSP}$(f)_{1-2\eps, \rho(f) \pm 2\eps}$ is UG-hard.
\end{theorem}

%Raghavendra and Steurer \cite{} and Khot and Saket \cite{} show that the integrality gap for the
%basic relaxation can be {\it lifted} to that for the mixed relaxation (we observe again that this
%too applies for a strong gap).

\begin{theorem}{\rm \cite{raghsteu2,ks09}} Let $\eps > 0$ be an arbitrarily small constant.

If the basic relaxation as in Figure \ref{fig:basic-sdp}  has a $(1- \eps, \rho(f)+\eps)$-integrality gap,
then the mixed relaxation as in Figure \ref{fig:mixed} has a  $(1-2\eps, \rho(f)+2\eps)$-integrality gap for a super-constant number of rounds.
%The same holds for a strong $(1- \eps, \rho(f)\pm \eps)$-integrality gap.
\end{theorem}

\subsection{Measure Theory and Probability}
We provide a basic background on relevant tools from measure theory and probability.
For further reference, please see ~\cite{met,lifshits}.

\subsubsection*{Measures, Weak$^*$ Convergence and Signed Measures}

\begin{definition}
Given a set $X$ along with a $\sigma$-algebra $\cal{F}$ (i.e. a non-empty collection of subsets of
$X$ that is closed under complementation and countable union), a {\deffont measure}
on $X$ is a function $m:\cal{F}\to[0,\infty]$ satisfying:
\begin{itemize}
\item $m(\emptyset)=0$.
\item If $\{R_j\}_{j=1}^\infty  \subseteq \cal{F}$ is a countable collection of disjoint sets, then $m\left(\cup_{j=1}^\infty R_j\right)=\sum_{j=1}^\infty m(R_j)$.
\end{itemize}
\end{definition}
We will consider only finite measures, i.e. those with $m(X) < \infty$. In particular, we will be
interested in {\it probability measures}, i.e. those with $m(X)=1$. The class of all probability
measures on $X$ is denoted as ${\sf Prob}(X)$.

We note that the Borel $\sigma$-algebra ${\cal B}$ on $\R^n$
is the smallest $\sigma$-algebra that contains all open balls w.r.t. the standard Euclidean
metric (and the sets in ${\cal B}$ are called measurable).
It can be restricted to $X \subseteq \R^n$ leading to the induced $\sigma$-algebra on $X$,
which will be the $\sigma$-algebra under consideration below. We note also that for subsets of
$\R^n$, being compact is same as being closed and bounded via the Heine-Borel Theorem.
It is also equivalent to being sequentially compact (existence of a convergent subsequence for every
infinite sequence) by the Bolzano-Weierstrass Theorem.

We state the main measure-theoretic result that we need in a form convenient to us:

\begin{theorem} Let $X \subseteq \R^n$ be a compact set and $\{\Lambda_i\}_{i=1}^\infty$ be a sequence
of probability measures on $X$. Then there exists a sub-sequence $\{\Lambda_{i_j}\}_{j=1}^\infty$
and a probability measure $\Lambda$ on $X$ such that
for any continuous function $h: X \mapsto \R$,
\begin{equation} \label{eq:weak-star}
    \lim_{j \rightarrow \infty} \int h ~d\Lambda_{i_j} =  \int h ~d\Lambda.
\end{equation}
\end{theorem}

This statement follows from the theorem stated below:

\begin{theorem}[Corollary 13.9 in \cite{met}]\label{wk*:thm}
Let $X$ be a compact metric space. Then the class of probability measures ${\sf Prob}(X)$ is compact
and metrizable in the weak$^*$ topology.
\end{theorem}

In words, the class ${\sf Prob}(X)$ can be endowed with a
suitable metric so that the metric topology
coincides with the weak$^*$ topology. Since ${\sf Prob}(X)$ is
compact and metrizable, it is also sequentially compact, i.e. every sequence has a convergent
subsequence. The convergence is w.r.t. the metric defined on  ${\sf Prob}(X)$ and as mentioned,
this is same as the convergence in the so-called weak$^*$ topology. The latter, by definition, is
precisely the statement that Equation \eqref{eq:weak-star} holds for every continuous function $h: X \mapsto \R$.

Let $X \subseteq \R^n, X' \subseteq R^{n'}$ be compact,
$\Lambda$ be a measure on $X$ and $\phi: X \mapsto X'$ be continuous. The
measure $\phi(\Lambda)$ on $X'$ is defined in a natural way as $\phi(\Lambda)(A') = \Lambda(\phi^{-1}(A'))$.
We will use this observation in two settings: (1) when $\phi$ is a projection of $X$ onto a subset
of co-ordinates $S \subseteq [n]$, the measure on $\R^{|S|}$ so obtained will be denoted as $\Lambda_S$
and refereed to as the projected measure.
(2) when $\phi$ is a bijection, we can pass back and forth between $\Lambda$ and $\phi(\Lambda)$,
regarding them as essentially the same.

Sometimes we will describe the construction of the measure $\phi(\Lambda)$ as above by informally
saying  ``sample $x \sim \Lambda$ and  take (or apply) $\phi(x)$''. When $h: X \mapsto \R$ is a
real valued function, we
will informally write $\Ex{x \in \Lambda}{h(x)}$, the ``expectation of $h(x)$ when
$x$ is sampled from $\Lambda$", to denote $\int h ~d\Lambda$.

We will also need the notion of a signed measure, which is
a generalization of the usual (non-negative) measure.
\begin{definition}
Given a set $X$ along with a $\sigma$-algebra $\cal{F}$, a {\deffont signed measure}
on $X$ is a function $m:\cal{F}\to[-\infty,\infty]$ allowed to take at most
one of the values in $\{-\infty, +\infty\}$ and satisfying:
\begin{itemize}
\item $m(\emptyset)=0$.
\item If $\{R_j\}_{j=1}^\infty  \subseteq \cal{F}$ is a countable collection of disjoint sets, then $m\left(\cup_{j=1}^\infty R_j\right)=\sum_{j=1}^\infty m(R_j)$
as long as the series $\sum_{j=1}^\infty m(R_j)$ is absolutely convergent.
\end{itemize}
\end{definition}

Let $\{\Lambda_i\}_{i=1}^q$ be a finite set of probability measures on $X$ with underlying $\sigma$-algebra ${\cal F}$ and $\{\alpha_i\}_{i=1}^q$ be
(possibly negative) reals. Then $\Lambda = \sum_{i=1}^q \alpha_i \Lambda_i$ is a signed measure. Formally,
for any $A \in {\cal F}$, $\Lambda(A) =
\sum_{i=1}^q \alpha_i \Lambda_i(A)$. We will consider only such signed measures, arising as finite linear
combinations of probability measures. Such a signed measure may identically vanish, i.e. $\Lambda(A)=0
~\forall A \in {\cal F}$. This is same as saying that if one writes $\Lambda = \Lambda' - \Lambda''$
as a difference of two non-negative measures (by grouping all $\Lambda_i$ with positive and negative coefficients respectively), then $\Lambda'$ and $\Lambda''$ are identical.

\subsubsection*{Gaussian Measures}

Let $\Sigma$ be an invertible, symmetric $t \times t$ matrix and $\mu$ be a $t$-dimensional vector.
The Gaussian measure of a (measurable) set
$A \subseteq \R^t$  w.r.t. {\it means}
$\mu = (\mu_1,\ldots, \mu_t)$ and the {\it covariance matrix} $\Sigma$ is defined as
 $$ \int_A   \gamma_t( y = (y_1,\ldots,y_t) , (\Sigma, \mu)) ~~ dy_1 dy_2 \ldots dy_t, $$
where $\gamma_t( \cdot, (\Sigma, \mu))$ is the Gaussian density function
 $$\gamma_t(y  , (\Sigma, \mu)) =
 \frac{1}{\sqrt{(2\pi)^t\mathrm{Det}(\Sigma)}} ~e^{-\frac{1}{2} \cdot (y-\mu)^T\Sigma^{-1}(y-\mu)}. $$
The random variables $y_1,\ldots,y_t$ then satisfy $\ExpOp[y_i] = \mu_i$ and
$\Sigma_{ij} = \ExpOp[y_i y_j ] - \mu_i \mu_j$.

With $\Sigma, \mu$ as above, one can also define a Gaussian measure on $(\R^d)^t$ that is a
product measure with the measure on each of the $d$ co-ordinates as above. Formally, if
${\bf y} = ({\bf y}_1,\ldots, {\bf y}_t)$ with ${\bf y}_i \in \R^d$ and one denotes
${\bf y}^{(\ell)} \in \R^t$ as the vector of $\ell^{th}$ co-ordinates of ${\bf y}_1,\ldots, {\bf y}_t$
respectively for $\ell \in [d]$, then the measure is given by a density
$\gamma_{t,d}({\bf y}, (\Sigma, \mu))$ defined as:
 $$\gamma_{t,d} ({\bf y}  , (\Sigma, \mu)) = \prod_{\ell=1}^d
 \frac{1}{\sqrt{(2\pi)^t\mathrm{Det}(\Sigma)}} ~ e^{-\frac{1}{2}\cdot ({\bf y}^{(\ell)}-\mu)^T\Sigma^{-1}({\bf y}^{(\ell)}-\mu)}. $$

\subsection{Fourier Representation of Functions on the Boolean Hypercube}

We note a basic fact that every function $f: \{-1,1\}^k \mapsto {\mathbb R}$ can be represented in the Fourier basis:
$$ f(x) =  \hat{f}(\emptyset) + \sum_{S \subseteq [k], S \not= \emptyset} \hat{f}(S) \cdot \prod_{i \in S} x_i. $$
When $f$ is a predicate, i.e. $\{0,1\}$-valued, the empty Fourier coefficient equals the density of the
predicate, i.e. $\hat{f}(\emptyset) = \rho(f)$.
We crucially use the following simple fact about multi-linear polynomials.  Multi-linearity is essential here as the polynomial      $-x^2$ shows.
\begin{lemma} \label{multi-linear:lem}
Let $g: \{-1,0,1\}^n \mapsto {\mathbb R}$ be a multi-linear polynomial with no constant term that is upper bounded by zero on all inputs. Then $g$ is
identically zero.
\end{lemma}
\begin{proof} Since $g$ has no constant term, $\Ex{}{g} = 0$ where the expectation is w.r.t. the uniform distribution on the inputs. Hence, if $g$ is not identically zero,
it must take a value that is strictly positive as well as a value that is strictly negative.
\end{proof}

\subsection{Our Results}
In this section we present formal statements of our results.
Given a predicate $f:\{-1,1\}^k\to\B$, let
%$\cal{P}(f)$ denote
%the convex polytope with the vertex set $f^{-1}(1)$ and let
$\cal{D}(f)$ denote the set of all probability distributions over
$f^{-1}(1)$.
%Observe that $\cal{D}(f)$ has a natural identification with
%$\cal{P}(f)$.
%
\begin{definition}\label{def:cf}
For $\nu \in \cal{D}(f)$, we let $\zeta(\nu)$ denote the
$(k+1) \times (k+1)$ symmetric {\deffont moment matrix}
for $\nu$ such that:
\begin{align*}
\forall i\in \{0\} \cup [k]:\ \zeta(i,i) &~=~ 1 \mcom \\
\forall i\in[k]:\ \zeta(0,i) &~=~ \Ex{x \sim \nu}{x_i} \mcom \\
\forall i,j\in[k], i\neq j:\ \zeta(i,j) &~=~ \Ex{x \sim \nu}{x_ix_j} \mper
\end{align*}
Also, let $\cal{C}(f) \subseteq \R^{(k+1)\times (k+1)}$
denote the compact, convex  set of all moment matrices:
\[ \cal{C}(f) ~:=~\{\zeta(\nu) \suchthat \nu \in\cal{D}(f)\}. \]
\end{definition}

Note that the definition of the polytope ${\cal C}(f)$  defers somewhat from that in
the introduction of the paper (it is now a $(k+1)\times (k+1)$ matrix as opposed to
$\left( k+ \binom{k}{2} \right)$-dimensional vector), but this difference is inconsequential.

For $S \subseteq [k]$, let $\zeta_S$ denote
$\zeta$ restricted to the rows and columns
of indices in $S \cup \{0\}$.
For a permutation $\pi:S\to S$ we use $\zeta_{S,\pi}$
to denote a permutation of the submatrix $\zeta_S$ with the
coordinates of $S$ permuted according to $\pi$.
Also, for a $|S|$-dimensional vector of signs $b \in \pmone^{|S|}$,
let $\zeta_{S,\pi,b}  = \zeta_{S,\pi}\circ ((1 ~b)(1 ~b)^T)$
i.e.,  the matrix obtained by taking the Hadamard product (entrywise product)
of the matrices $\zeta_{S,\pi}$ and $(1 ~b)(1 ~b)^T$.

\begin{definition}\label{cvxptppred:dfn}
Let $\Lambda$ be a probability measure supported on $\cal{C}(f)$.
Then, for $S \sub [k]$, let $\Lambda_S$ denote the measure on
$(|S|+1) \times (|S|+1)$ matrices obtained by sampling $\zeta \sim \Lambda$ and taking the matrix
$\zeta_S$.
Let $\pi : S \to S$ be any permutation and let
 $b \in \pmone^{|S|}$ be a vector of signs.
We denote by $\Lambda_{S,\pi,b}$ the measure on
$(|S|+1) \times (|S|+1)$ matrices obtained by sampling
$\zeta \sim \Lambda$ and taking the matrix $\zeta_{S,\pi,b}$.
\end{definition}

We define a generic family of algorithms based on $d$-dimensional rounding of the vector solution
to the basic relaxation, Figure \ref{fig:basic-sdp}. We choose to state an informal definition
here as the exact rounding process is a bit cumbersome, formally described in Subsection \ref{Lge0rnd:sbs}.

\begin{definition} \mbox{(Informal)}
A {\deffont $d$-dimensional rounding algorithm} is a polynomial time algorithm based on an odd
measurable function
$\psi:\R^d\to [-1, 1]$. The algorithm solves the basic relaxation for \maxkcsp$(f)$,
projects the SDP vectors onto a random $d$-dimensional subspace and then rounds them
to $\{-1,1\}$ values according to (biases given by) $\psi$.
The algorithm may draw the function $\psi$ itself from a certain (pre-determined) distribution.
% in which case we call the overall algorithm as a {\deffont $d$-dimensional rounding scheme.}
%
\end{definition}

Our main result appears below. It states that a predicate either admits a non-trivial approximation based on a
$(k+1)$-dimensional rounding algorithm or is  approximation resistant. This ``dichotomy" is characterized precisely by the existence of a measure $\Lambda$ on
${\cal C}(f)$ as in Definition \ref{def:chars}.

\begin{theorem}\label{sdp:thm}
Given $f : \pmone^k\to\{0,1\}$, the following ``dichotomy" holds:
\begin{itemize}
\item Either there is a constant $\eps > 0$ and a $(k+1)$-dimensional rounding algorithm that
given a $(1-\eps)$-satisfiable instance of \maxkcsp$(f)$, outputs an assignment $A$ such that
$\Ex{A}{ \sat(A)} \geq \rho(f) + \eps$ (i.e. achieves a non-trivial approximation),

\item Or there exists a probability measure $\Lambda$ on ${\mathcal C}(f)$,
such that for all $t\in[k]$, and a uniformly random choice of $S$ with  $|S|=t$,
$\pi:S\to S$ and $b\in\{\pm 1\}^{|S|}$, the following signed measure on $(t+1) \times (t+1)$ matrices:
\begin{equation}
\Lambda^{(t)} ~:=~
\ExpOp_{|S| = t} ~\ExpOp_{\pi : S \to S} ~\Ex{b \in \pmone^{|S|}}{ \left( \prod_{i\in S} b_i \right)   \cdot\HF(S) \cdot \Lambda_{S,\pi,b}}\label{main:eqn}
\end{equation}
is identically zero. In this case for every $\eps > 0$, the predicate has a  $(1-\eps,
\rho(f) + \eps)$ integrality gap for the basic relaxation and (hence) for the
mixed relaxation with a super-constant number of rounds and
is  approximation
resistant, i.e. {\sf GapCSP}$(f)_{1-\eps, \rho(f) + \eps}$ is UG-hard.
\end{itemize}
\end{theorem}
Similarly, we obtain a ``dichotomy" for the integrality gap in the Sherali-Adams LP hierarchy.
The characterization is syntactically similar once
the polytope ${\cal C}(f)$ is replaced by the polytope ${\cal C}^*(f)$ consisting of
only the first moment vectors of distributions supported on $f^{-1}(1)$ (and is therefore
the same as the convex hull of $f^{-1}(1)$). For a measure $\Lambda^*$ on ${\cal C}^*(f)$
and a subset $S \subseteq [k]$, the projected measure $\Lambda^*_S$ and the measure
$\Lambda^*_{S,\pi,b}$ for a permutation $\pi: S \mapsto S$ and signs $b \in \{-1,1\}^S$ are
defined in an analogous manner.
The family of generic algorithms is now defined w.r.t. only the first moments, i.e. the
algorithm can ``use" only the biases computed by the LP relaxation.

\begin{definition}
For $\nu \in \cal{D}(f)$, we let $\zeta(\nu)$ denote the
$k$-dimensional vector such that:
\[\forall i \in [k] : \ \zeta(i) ~:=~ \Ex{x \sim \nu}{x_i}.\]
Let $\cal{C}^*(f) \subseteq \R^k$ denote the convex, compact set:
\[\cal{C}^*(f) ~:=~ \{\zeta(\nu): \nu \in \cal{D}(f)\}.\]
\end{definition}

\begin{definition} \mbox{(Informal)}
A {\deffont $k$-round LP rounding algorithm} is a polynomial time algorithm based on an odd measurable function
$\psi: [-1,1] \to [-1, 1]  $. The algorithm solves the $k$-round
Sherali-Adams relaxation for \maxkcsp$(f)$ and then a CSP variable with bias
$p$ (as computed by the LP relaxation) is rounded to a $\{-1,1\}$ value with bias $\psi(p)$,
independently for different variables.
The algorithm may draw the function $\psi$ itself from a certain (pre-determined) distribution.
% in which case we call the overall algorithm as a {\deffont $d$-dimensional rounding scheme.}
%
\end{definition}

\begin{theorem}\label{lp:thm}
Given $f : \pmone^k\to\{0,1\}$, the following ``dichotomy" holds:
\begin{itemize}
\item Either there is a constant $\eps > 0$ and a $k$-round LP rounding algorithm that
given a $(1-\eps)$-satisfiable instance of \maxkcsp$(f)$, outputs an assignment $A$ such that
$\Ex{A}{ \sat(A) } \geq \rho(f) + \eps$ (i.e. achieves a non-trivial approximation),

\item Or there exists a probability measure $\Lambda^*$ on ${\mathcal C}^*(f)$,
such that for all $t\in[k]$, and a uniformly random choice of $S$ with  $|S|=t$,
$\pi:S\to S$ and $b\in\{\pm 1\}^{|S|}$, the following signed measure on $t$-dimensional
vectors:
\begin{equation}
\Lambda^{*,(t)} ~:=~
\ExpOp_{|S| = t} ~\ExpOp_{\pi : S \to S} ~\Ex{b \in \pmone^{|S|}}{ \left( \prod_{i\in S} b_i \right)   \cdot\HF(S) \cdot \Lambda^*_{S,\pi,b}}\label{mainlp:eqn}
\end{equation}
is identically zero. In this case for every $\eps > 0$, the predicate has a  $(1-\eps,
\rho(f) + \eps)$ integrality gap for a super-constant number of rounds of the Sherali-Adams
LP relaxation.
\end{itemize}
\end{theorem}

As described in the introduction, we obtain additional interesting observations and results,
e.g. approximation resistance for $k$-partite version of CSPs and in all cases, strong approximation
resistance instead of just approximation resistance. We skip
their formal statements and proofs from the current version of the paper. The proofs of our main
results, namely Theorem \ref{sdp:thm} and Theorem \ref {lp:thm}, appear in Section~\ref{sdpthm:sec} and
Section~\ref{lpthm:sec} respectively.

%\begin{theorem}\label{lp:thm}
%Given $f:\pmone^k \to \B$, the following ``dichotomy" holds:
%\begin{itemize}
%\item If there exists probability measure $\Lambda'$ over $\cal{C'}(f)$
%such that the following holds for any odd function $\psi:[-1,1]\to[-1,1]$:
%\begin{equation}
%\Ex{\eta \sim \Lambda'}{\sum_{S \neq \emptyset} \HF(S) \cdot \prod_{i\in S} \psi(\eta(i))} ~=~ 0,
%\end{equation}
%then the predicate $f$ has a gap of $(1-o(1), \rho(f)+o(1))$ even
%after arbitrarily large, but constantly many, rounds of the natural
%SA relaxation for $\maxkcsp(f)$.

%\item Otherwise, there is an approximation algorithm based on rounding the
%level $k$ SA LP for $\maxkcsp(f)$ which has approximation factor better than
%$\frac{1}{\rho(f)}$.
%\end{itemize}
%\end{theorem}

%As a consequence we obtain the following corollary.
%\begin{corollary}
%For $\maxkcsp(f)$, the worst case guarantee upon rounding
%the LP obtained after $k$ rounds of the SA hierarchy is the
%same as the guarantee obtained upon rounding the LP
%relaxation after (any) constant number of rounds of the
%SA hierarchy.
%\end{corollary}

%\subsection{Overview of the Proofs}
%

%\section{Basic Definitions}\label{defn:sec}
%

\subsection{The Unique Games Conjecture}

We present the definitions of the Unique Games problem and UG-hardness.

\begin{definition} A Unique Games instance ${\cal L}(G(V,E), [L], \{ \pi_{v,w} \}_{(v,w)\in E}  )$
consists of a graph $G(V,E)$, a set of labels $[L]$ and a set of permutations $\pi_{v,w}: [L] \mapsto
[L]$, one for each edge of the graph (the edges have an implicit direction). A labeling is an assignment
$A: V \mapsto [L]$. The labeling satisfies an edge $(v,w)$ if $\pi_{v,w}(A(v)) = A(w)$. $\opt({\cal L})$
is the maximum fraction of edges satisfied by any labeling.
\end{definition}

Let {\sf GapUG}$_{1-\delta, \delta}$ denote the gap version of the Unique Games problem where the
instance ${\cal L}$ is promised to have either ~$\opt({\cal L}) \geq 1-\delta$ ~or ~$\opt({\cal L})\leq \delta$.
Khot \cite{khot} conjectures that for an arbitrarily small constant $\delta > 0$,
{\sf GapUG}$_{1-\delta, \delta}$ is NP-hard on instances with $L$ labels where $L = L(\delta)$
may depend on $\delta$.

\begin{definition} A decision problem is said to be UG-hard if for a sufficiently small constant
$\delta > 0$, there is a polynomial time reduction from {\sf GapUG}$_{1-\delta, \delta}$ (with
the number of labels $L=L(\delta)$) to the problem under consideration.
\end{definition}

\section{Proof of the SDP Dichotomy Theorem}\label{sdpthm:sec}
In this section we present the proof of Theorem~\ref{sdp:thm}.
We begin by developing the game theory formalism that we will need. This formalism
is mostly common to both dichotomy theorems but we state it first with the SDPs in
mind.
\subsection{Game-Theoretic Formulation}\label{game:sbs}
We have two players: \Ang, the player trying to design an algorithm, and \Dev, the player trying to
prove a hardness result. Intuitively, \Ang wants
to show that  $\maxkcsp(f)$ admits a non-trivial efficient approximation
via rounding the natural SDP relaxation.
For this, \Ang will try to maximize the pay-off in the (zero-sum) game, which we shall define soon.
On the other hand, \Dev intends to minimize the pay-off.
Intuitively, \Dev wants to show that there exist $\maxkcsp(f)$ instances for which the integrality
gap is high.

The \emph{pure strategies} of \Dev will correspond to \emph{distributions} over
moment matrices. Recall that $\cal{C}(f)$ was the set of
$(k+1) \times (k+1)$ moment matrices for distributions in $\cal{D}(f)$. To ensure that our moment
matrices are non-singular, we will actually need to work with a slightly modified body
$\C_{\delta}(f)$ defined as
\[
\C_{\delta}(f) ~\defeq~ \inbraces{(1-\delta) \cdot \zeta + \delta \cdot \mathbb{I}_{k+1} \suchthat \zeta
  \in \C(f)} \mcom
\]
for a constant $\delta \in (0,1)$. Here, $\mathbb{I}_{k+1}$ denotes the $(k+1) \times (k+1)$
identity matrix.
Let $R_i$, such that $R_1\subseteq R_2 \ldots \subseteq R_p\subseteq ...$
denote a fixed sequence of finite subsets of $\C_{\delta}(f)$.
We assume that the above sequence is dense in $\C_{\delta}(f)$ in the limit.
Let $\Rp$ denote the class of distributions over $R_p$, such that all the probabilities are
integral multiples of $1/2^p$.

% For a fixed integer $d$, let $\bbox$ denote the box $[-1,1]^d$.
Let $\{\cal{P}_q\}_{q\in\N}$ denote a sequence of partitions, where $\cal{P}_q$ partitions $\bbox$ into
$2^{(q+1)d}$ boxes of equal size. We will choose $d=k+1$ for reasons that will become clear later.
Note that for each $q$, $\cal{P}_{q+1}$ is a refinement of $\cal{P}_q$
inside.
Let $\values$ denote the set of values $\inbraces{-1,0,1 }$
and let $\psi_{q} : \R^d \to \values$ denote an \emph{odd} function which takes values
in $\values$ in $\bbox$ and $0$ outside. We further assume that $\psi_{q}$ is
constant on each cell of $\cal{P}_q$. This can be ensured since the partitions ${\cal P}_q$
are symmetric with respect to 0.

For every fixed $p,q$ above, we will define a zero-sum game $\cal{G}_{p,q}$.
A pure strategy for \Dev corresponds to a distribution $\lambda \in \Rp$, and a mixed
strategy is a probability distribution $\Lambda_p$ over $\Rp$.
The pure strategies for \Ang are given by all possible functions $\psi_{q}$, and a mixed strategy
is a probability distribution $\Gamma_q$ over these.
For $\lambda \in \Rp$ and $\psi = \psi_q: \R^d \to V$ as described above, we first define a
function $\opay(\lambda, \psi)$ and then the
payoff of the (pure strategy) game $\G_{p,q}$ is defined as a function $\pay(\lambda, \psi)$ that closely
approximates $\opay(\lambda, \psi)$. We will elaborate soon why we need two separate functions. Let
\begin{equation}%\label{payzpsi:eqn}
\opay(\lambda,\psi)  ~:=~
\ExpOp_{\zeta \sim \lambda} ~\Ex{\y_1, \ldots, \y_k \sim \cal{N}_d(\zeta)}{\sum_{S \neq \emptyset} \HF(S) \cdot \prod_{i\in
      S}\psi(\y_i)    }  \mcom
\end{equation}
where $\{\y_i\}_{i\in [k]}$ are points in $\R^d$ sampled from a Gaussian process $\cal{N}_d(\zeta)$
different coordinates being independent and for each coordinate $l \in [d]$,
$\ex{(\y_i)_l} = \zeta(0,i)$ and $\ex{(\y_i)_l (\y_j)_l} = \zeta(i,j)$ $\forall i,j \in [k]$.
Note that this corresponds to the advantage over a random assignment that would be obtained by the
rounding algorithm corresponding to $\psi$, if the SDP vectors corresponding to variables in each
constraint had correlations given by $\zeta$. Note also that if some $\y_i$ lies outside $[-1,1]^d$, the
terms involving $\psi(\y_i)$ vanish since $\psi(\cdot)$ is zero outside $[-1,1]^d$.

% We shall also use the notation $\opay(\lambda,\psi)$ to
%denote the above expression, for an arbitrary distribution $\lambda$ over $\C_{\delta}(f)$ and
%any function $\psi: \R^d \to \values$. For a matrix $\zeta$, we use $\opay(\zeta,\psi)$ to
%denote $\opay(\lambda,\psi)$ for a distribution $\lambda$ concentrated on a single point
%corresponding to $\zeta$.

The payoff of the (pure strategy) game $\G_{p,q}$ is now defined as (noting that $\psi = \psi_q$):

\begin{equation}\label{payzpsi:eqn}
\pay(\lambda,\psi)  ~:=~
\ExpOp_{\zeta \sim \lambda} ~\Ex{\y_1, \ldots, \y_k \sim \cal{N}_d(\zeta)}{\sum_{S \neq \emptyset} \HF(S) \cdot \prod_{i\in
      S}\psi(\y_i) \cdot \indic_q\left( \{ \y_i | i \in S\}\right)   }  \mcom
\end{equation}
where  $\indic_q\left( \{ \y_i | i \in S\}\right)$ is an indicator of the event that all the $2\cdot |S|$ points in the set
$\{ \y_i , - \y_i | i \in S\}$ are inside $[-1,1]^d$ and lie in distinct cells of the partition ${\cal P}_q$.
For mixed strategies $\Lambda_p$ and $\Gamma_q$ the payoff is  given by:
\begin{eqnarray}
\pay(\Lambda_p,\Gamma_q) &  := &
\ExpOp_{\lambda \sim \Lambda_p} ~\ExpOp_{\psi \sim \Gamma_q} \left[ \pay(\lambda, \psi) \right] \nonumber \\
 & = & \ExpOp_{\lambda \sim \Lambda_p} ~\ExpOp_{\psi \sim \Gamma_q}  ~\ExpOp_{\zeta \sim \lambda} ~\Ex{\y_1, \ldots, \y_k \sim \cal{N}_d(\zeta)}{\sum_{S \neq
      \emptyset} \HF(S) \cdot \prod_{i\in S}\psi(\y_i)\cdot  \indic_q\left( \{ \y_i | i \in S\}\right)    } \label{payzpsi1:eqn}  \mper
\end{eqnarray}
\Ang plays to maximize the above payoff and \Dev plays to
minimize the same.
By Von Neumann's min-max theorem there exists a unique value for
the above game $\G_{p,q}$, for every $p$ and $q$.
Our next task will be to relate the value of this game (in the limit)
to the hardness of the predicate $f$. Before we begin, let us briefly comment on why we use two different payoff functions.
As it turns out, the function $\opay(\cdot,\cdot)$ is more suited towards designing an algorithm since it precisely
captures the advantage over $\rho(f)$.  On the other hand, the function  $\pay(\cdot,\cdot)$ is more suited towards inferring the
existence of a vanishing measure and proving approximation resistance. In the latter case, it is crucial to avoid non-multi-linear terms from certain
polynomials that we  encounter.  The polynomials are essentially the r.h.s. of  expressions \eqref{payzpsi:eqn}-\eqref{payzpsi1:eqn}. The variables
of the polynomials are $\psi_q(\y_i)$ and
since $\psi_q$ is constant on the cells of the partition ${\cal P}_q$, the number of variables is finite. If two or more of the
points in the set $\{ \y_i , - \y_i | i \in S\}$ lie in the same cell of the partition, this gives rise to a
non-multi-linear term  (note that $\psi_q$ is odd, so
$\psi_q(-\y_i) = - \psi_q(\y_i)$). The indicator function deletes these undesired non-multi-linear terms.

The following simple facts about the payoff functions will be used repeatedly. The second fact shows that the two payoff functions
are very close to each other for large enough $q$ and hence one may use one or the other depending on whether one intends to
design an algorithm or to infer the existence of a vanishing measure.

\begin{claim}\label{payoff-bound:clm}
For any $\lambda$ and $\psi = \psi_q: \R^d \to V$, both $\pay(\lambda,\psi)$  and $\opay(\lambda,\psi)$ are bounded in absolute value by $2^k$.
\end{claim}
\begin{proof}
This is because the payoff is expectation of a sum of at most $2^k - 1 $ terms, each with absolute value at most $1$.
\end{proof}

\begin{claim}  \label{payoffs-close:clm}  For any $\lambda$ and  $\psi = \psi_q: \R^d \to V$,
 $$\left| \pay(\lambda, \psi) - \opay(\lambda, \psi) \right| ~\leq ~c_{k,d,\delta,q},  $$
 where for every fixed $k,d$ and $\delta$, we have $c_{k,d,\delta,q} \to 0$ as $q \to \infty$.
\end{claim}
\begin{proof} The two payoff functions may differ only when some pair of points in the set $\{\y_i, -\y_i | i \in [k]\}$ lie in the same cell of
the partition ${\cal P}_q$. The claim follows since the cells of the partition ${\cal P}_q$ become arbitrarily small in size as $q \to \infty$,  the
points $\y_1,\ldots,\y_k$ are sampled from $\cal{N}_d(\zeta)$ with $\zeta \in {\cal C}_\delta(f)$ and hence any pair $(\y_i, \y_j)$ is at most
$(1-\delta)$ correlated.
\end{proof}

We will use the notation $\pay(\lambda, \psi)$  to also
denote the payoff for an arbitrary distribution $\lambda$ over $\C_{\delta}(f)$ (i.e. $\lambda$ need not necessarily be in
$\Rp$ for some $p$).
For a matrix $\zeta \in \C_{\delta}(f)$, we use $\pay(\zeta,\psi)$ to
denote $\pay(\lambda,\psi)$ for a distribution $\lambda$ concentrated on a single point
corresponding to $\zeta$.
The same goes for the function $\opay(\cdot, \psi)$. Moreover, for the function $\opay(\cdot, \psi)$, the notation makes sense
even if $\psi: \R^d \mapsto [-1,1]$ is an arbitrary measurable odd function.

%\begin{claim}\label{payoff-bound:clm}
%For any $\lambda$ and $\psi: \R^d \to [-1,1]$, $\pay(\lambda,\psi)\leq 1$.
%\end{claim}
%%
%\begin{proof}
%Since the pay-off function is
%\[
%\pay(\lambda,\psi)  ~:=~
%\abs{\ExpOp_{\zeta \sim \lambda} ~\Ex{\y_1, \ldots, \y_k \sim \cal{N}_d(\zeta)}{\sum_{S \neq \emptyset} \HF(S) \cdot \prod_{i\in
%    S}\psi(\y_i)}} \mcom
%\]
%it suffices to show that for any $k$ values $\psi(\y_1),\ldots, \psi(\y_k) \in \values$, we have
%$\abs{\sum_{S \neq \emptyset} \HF(S) \cdot \prod_{i\in S}\psi(\y_i)} \le 1$.
%%
%This is immediate if all the values are $\pm 1$, since we can write the expression inside the absolute value as
%$f(\psi(\y_1), \ldots,\psi(\y_k)) - \hf(\emptyset)$, which only takes values $-\hf(\emptyset)$ and
%$1-\hf(\emptyset)$ since $f$ is Boolean.
%%
%If a value, say $\psi(\y_i)$ is in $(-1,1)$, we can view the
%expression as an expectation over the cases where we choose the $i^{th}$ input variable of $f$ to be
%1 with  probability $(1+\psi(\y_i))/2$ and $-1$ with probability $(1-\psi(\y_i)/2)$, which
%shows that it is always between $-\hf(\emptyset)$ and $1-\hf(\emptyset)$.
%\end{proof}
%

%
%
We will also need another simple fact about a matrix $\zeta \in \C_{\delta}(f)$. Recall that for a
matrix $\zeta$, for $S \sub [k]$, $\pi: S \to S$ and $b \in \pmone^S$, we define the matrix
$\zeta_{S,\pi,b}$ by considering the submatrix given by rows and columns in $S$ (and the
first row and first column), permuting them
according to $\pi$ and multiplying each row $i$ by $b_i$ and each column $j$ by $b_j$ (thus, the
$(i,j)$ entry is multiplied by $b_i \cdot b_j$). Also, for each matrix $\zeta' = \zeta_{S,\pi,b}$,
we can define a \emph{covariance matrix} $\Sigma$, with $\Sigma_{ij} = \zeta'(i,j) - \zeta'(0,i)
\cdot \zeta'(0,j)$. Then we shall use the following fact repeatedly.
%%%%%%%%%%%%%%%%%%%%%%%%%%%%%%%%%%%
\newcommand{\tz}{\tilde{\zeta}}
%%%%%%%%%%%%%%%%%%%%%%%%%%%%%%%%%%%
\begin{claim}\label{eigenvalues-delta:clm}
Let $\zeta \in \C_{\delta}(f)$ and let $\zeta_{S,\pi,b}$ be as defined above for an arbitrary choice
of $S$, $\pi$ and $b$. Let $\Sigma$ be the covariance matrix corresponding to
$\zeta_{S,\pi,b}$. Then $\Sigma$ is a positive semidefinite matrix with all eigenvalues at least $\delta$.
\end{claim}
\begin{proof}
Note that it is sufficient to prove the claim with $\pi$ being the identity permutation and $b =
1^k$,  since permuting the rows and columns, or multiplying them with a sign
does not affect the eigenvalues.
Let us first consider the case when $S = [k]$ and $b = 1^k$ (and thus $\zeta_{S,\pi,b} = \zeta$).
For $\zeta \in \C_{\delta}(f)$, there is a distribution $\nu$ on $f^{-1}(1)$, and in
particular on $\pmone^{k}$, so that $\zeta = (1-\delta) \cdot \zeta(\nu) + \delta \cdot
\mathbb{I}_{k+1}$. Let $\tz$ denote $\zeta(\nu)$. If $\Sigma(\zeta)$ denotes the covariance matrix
corresponding to $\zeta$, then we can write
\[
\Sigma(\zeta)
~=~ (1-\delta) \cdot \Sigma(\tz) + \delta \cdot (1-\delta) \cdot M + \delta \cdot \mathbb{I}_k \mcom
\]
where $M$ is a positive semidefinite (PSD) matrix with $M_{ij} = \tz(0,i) \cdot \tz(0,j)$. Also,
note that $\Sigma(\tz)$ is a covariance matrix corresponding to a distribution $\nu$ on $\pmone^k$
and is hence PSD. Thus, all eigenvalues for $\Sigma(\zeta)$ are at least $\delta$.
Similarly, when $|S|=t$ for some $t \leq k$, we can consider $\nu_S$, the projection of $\nu$ to
$\pmone^S$. We can again write $\zeta_S = (1-\delta) \cdot \zeta(\nu_S) + \delta \cdot
\mathbb{I}_{|S|+1}$. The rest of the proof is same as above.
\end{proof}
%
% In Subsection~\ref{game:sbs} we have defined the game $\G_{p,q}$ for
% every pair of positive integers $p$ and $q$.
Let $\cal{V}$ denote the
(infinite) matrix over reals such that: $\cal{V}(p,q) := \val(\G_{p,q})$
(the equilibrium value for the game $\G_{p,q}$).
\begin{lemma}\label{lmt:lm}
The limit $L$ defined below exists, is finite and is non-negative.
\begin{equation}\label{L:eqn}
L ~:=~ \lim_{p,q\to\infty} \cal{V}(p,q) \mper
\end{equation}
Moreover, every row $p$ has a limit $r_p$ as $q \to\infty$ and every
column $q$ has a limit $c_q$ as $p\to\infty$.
\end{lemma}
\begin{proof}
First, observe that the entries in $\cal{V}$ are all non-negative (since the strategy $\psi_q$ that is an identically zero function is
always available to  \Ang) and bounded above by $2^k$ (using
Claim \ref{payoff-bound:clm}).
Also, for any fixed $q$, $\cal{V}(p,q)$ is non-increasing as $p$ increases
since $\Rp \sub \cal{R}_{p+1}$. Similarly, for any fixed $p$, $\cal{V}(p,q)$ is non-decreasing as $q$
increases. This follows from the fact that  ${\cal P}_{q+1}$ is a refinement of ${\cal P}_{q}$
and thus each strategy $\psi_q$ can also be implemented by a
function $\psi_{q+1}$.

Therefore, by the monotone convergence theorem every row (resp. column)
in $\cal{G}$ has a limit, say $r_p$ (resp. $c_q$).
Moreover, $r_p$ is non-increasing as $p$ increases and $c_q$ is
non-decreasing as $q$ increases. Therefore, again by the monotone convergence
theorem both these sequences have to converge. Also, the limits must coincide since for sufficiently
large $p,q$, $\cal{V}(p,q)$ must come arbitrarily close to both the limits. This common limit is
denoted by $L$ in our statement.
\end{proof}
We shall also need the following lemma. We will need its conclusion to hold when $\zeta, \zeta'$
are not necessarily in ${\cal C}_\delta (f)$, but are still ``sufficiently non-singular". A convenient
notation is to use the body ${\cal C}({\bf 1})$  corresponding to the predicate ${\bf 1}:\{-1,1\}^k \to \{0,1\}$ that is
constant $1$. The points in this body are moment matrices of
distributions supported on $\{-1,1\}^k$ as per Definition \ref{def:cf}. We now allow $\zeta, \zeta'$
to be in the ``noise-added" body ${\cal C}_\delta({\bf 1})$. The payoff functions are now allowed to
have such $\zeta \in {\cal C}_\delta({\bf 1})$ as their first argument.
%
%\begin{lemma}\label{payoff-continuous:lm}
%% For each $q \in \N$, the pay-off function $\pay(\zeta,\psi_q)$ is continuous in $\zeta$. Moreover
%% the modulus of continuity depends on $k$, $q$ and $\delta$, but it is independent of the actual
%% function $\psi_q$.
%If $\zeta, \zeta' \in \C_{\delta}(f)$ are such that $\norm{\zeta-\zeta'}_{\infty} \leq \eps$,
%then for any function $\psi: \R^d \to \values$, we have
%\[
%\abs{
%\Ex{\y_1, \ldots, \y_k \sim \cal{N}_d(\zeta)}{\sum_{S \neq \emptyset} \HF(S) \cdot \prod_{i\in
%    S}\psi(\y_i)}
%-
%\Ex{\y_1, \ldots, \y_k \sim \cal{N}_d(\zeta')}{\sum_{S \neq \emptyset} \HF(S) \cdot \prod_{i\in
%    S}\psi(\y_i)}
%}
%~=~
%O_{k,d,\delta}(\eps) \mper
%\]
%Hence, the function $\opay(\zeta,\psi)$ is $O_{k,d,\delta}(1)$-Lipschitz in the argument $\zeta$.
%% \[\abs{\pay(\zeta,\psi) - \pay(\zeta',\psi)} ~=~ O_{k,d,\delta}(\eps) \mper\]
%\end{lemma}
\begin{lemma}\label{payoff-continuous:lm}
% For each $q \in \N$, the pay-off function $\pay(\zeta,\psi_q)$ is continuous in $\zeta$. Moreover
% the modulus of continuity depends on $k$, $q$ and $\delta$, but it is independent of the actual
% function $\psi_q$.
If $\zeta, \zeta' \in \C_{\delta}({\bf 1})$ are such that $\norm{\zeta-\zeta'}_{\infty} \leq \eps$,
then for any function $\psi = \psi_q: \R^d \to \values$, we have
$$ \left| \opay(\zeta, \psi) - \opay(\zeta',\psi)   \right|  = O_{k,d,\delta}(\eps) \mper$$
%\[
%\abs{
%\Ex{\y_1, \ldots, \y_k \sim \cal{N}_d(\zeta)}{\sum_{S \neq \emptyset} \HF(S) \cdot \prod_{i\in
%    S}\psi(\y_i)}
%-
%\Ex{\y_1, \ldots, \y_k \sim \cal{N}_d(\zeta')}{\sum_{S \neq \emptyset} \HF(S) \cdot \prod_{i\in
%    S}\psi(\y_i)}
%}
%~=~
%O_{k,d,\delta}(\eps) \mper
%\]
I.e. the function $\opay(\zeta,\psi)$ is $O_{k,d,\delta}(1)$-Lipschitz in the argument $\zeta$. The same holds for
the function $\pay(\zeta, \psi)$.
% \[\abs{\pay(\zeta,\psi) - \pay(\zeta',\psi)} ~=~ O_{k,d,\delta}(\eps) \mper\]
\end{lemma}

\begin{proof}
% The expression for $\pay(\zeta,\psi)$ is of the form
% \[
% \pay(\zeta,\psi) ~=~ \abs{\Ex{\y_1, \ldots, \y_k \sim \cal{N}_d(\zeta)}{\sum_{S \neq \emptyset} \HF(S) \cdot \prod_{i\in
%     S}\psi(\y_i)}} ~=~ \abs{\Ex{\y_1, \ldots, \y_k \sim \cal{N}_d(\zeta)}{g(\y_1,\ldots,\y_k)}} \mcom
% \]
We prove the lemma for the $\opay(\cdot,\cdot)$ function. The proof for the $\pay(\cdot, \cdot)$ function is the same.
Let $g(\y_1,\ldots,\y_k)$ denote the expression $\sum_{S \neq \emptyset} \HF(S) \cdot \prod_{i\in
  S}\psi(\y_i)$ so that
  $$\opay(\zeta, \psi) = \Ex{\y_1, \ldots, \y_k \sim
    \cal{N}_d(\zeta)}{g(\y_1,\ldots,\y_k)} \mper $$
Since $g(\y_1,\ldots,\y_k)$ is bounded by $2^k$ in absolute value, it is clear that  $ \left| \opay(\zeta, \psi) - \opay(\zeta',\psi)   \right|  $
is bounded by $2^{k+1} \cdot \norm{\cal{N}_d(\zeta) - \cal{N}_d(\zeta')}_1  $ where  $\norm{\cal{N}_d(\zeta) - \cal{N}_d(\zeta')}_1$ denotes the total variation distance
between the two distributions. This can be bounded by $O_{k,d,\delta}(\eps)$ as below.

%We first note that proving the bound claimed above
%is also sufficient to show that $\opay(\zeta,\psi)$ is $O_{k,d,\delta}(1)$-Lipschitz
%in the argument $\zeta$, since $\opay(\zeta,\psi) = \abs{\Ex{\y_1, \ldots, \y_k \sim
%    \cal{N}_d(\zeta)}{g(\y_1,\ldots,\y_k)}}$
%and hence,
%\[
%\abs{\pay(\zeta,\psi) - \pay(\zeta',\psi)}
%~\le~
%\abs{\Ex{\y_1, \ldots, \y_k \sim \cal{N}_d(\zeta)}{g(\y_1,\ldots,\y_k)}
%- \Ex{\y_1, \ldots, \y_k \sim \cal{N}_d(\zeta')}{g(\y_1,\ldots,\y_k)}} \mper
%\]
%
%%By Claim
%%\ref{payoff-bound:clm},
%We have a trivial upper bound $\abs{g(\y_1,\ldots,\y_k)} \le 2^k$. Hence, we can bound the
%above expression as
%\begin{align*}
%% \abs{\pay(\zeta,\psi) - \pay(\zeta',\psi)}
%% &~\le~
%\abs{\Ex{\y_1, \ldots, \y_k \sim \cal{N}_d(\zeta)}{g(\y_1,\ldots,\y_k)}
%- \Ex{\y_1, \ldots, \y_k \sim \cal{N}_d(\zeta')}{g(\y_1,\ldots,\y_k)}}
%&~\le~ 2^{k+1} \cdot \norm{\cal{N}_d(\zeta) - \cal{N}_d(\zeta')}_1 \mcom
%\end{align*}
%where $\norm{\cal{N}_d(\zeta) - \cal{N}_d(\zeta')}_1$ denotes the total variation distance
%between the two distributions. This can be bounded by $O_{k,d,\delta}(\eps)$ as below.

By Pinsker's inequality one can bound the total variation distance by the Kullback-Leibler (KL)
divergence, denoted $D(\cal{N}_d(\zeta) || \cal{N}_d(\zeta'))$, as follows.
\begin{align*}
\norm{\cal{N}_d(\zeta) - \cal{N}_d(\zeta')}_1^2
&~\le~
\frac12 \cdot \insquare{D(\cal{N}_d(\zeta) \Vert \cal{N}_d(\zeta')) + D(\cal{N}_d(\zeta') \Vert
  \cal{N}_d(\zeta))}  \\
&~=~
\frac{d}{2} \cdot \insquare{D(\cal{N}(\zeta) \Vert \cal{N}(\zeta')) + D(\cal{N}(\zeta') \Vert \cal{N}(\zeta))} \mcom
\end{align*}
where the equality uses the fact that $\cal{N}_d(\zeta)$ and $\cal{N}_d(\zeta')$ are product
distributions of $d$ $k$-dimensional Gaussians and $D((P_1,P_2) \Vert (Q_1,Q_2)) = D(P_1 \Vert Q_1) +
D(P_2 \Vert Q_2)$ for product distributions $(P_1,P_2)$ and $(Q_1,Q_2)$. The sum on the right can
now be bounded by $O_{k,d,\delta}(\eps^2)$.

Let $\Sigma$ and $\Sigma'$ denote the covariance matrices for $\zeta$ and $\zeta'$.
Let $\mu$ and $\mu'$ denote the vector of means for $\zeta$ and $\zeta'$.
For a multivariate normal distribution in $k$ dimensions the sum of the two KL
divergences, as above, can be written as follows (for eg. see chapter 15 in \cite{wna})
%\cite{Wikipedia:Multivariate-Normal}).
\begin{equation*}
D(\cal{N}(\zeta) \Vert \cal{N}(\zeta')) + D(\cal{N}(\zeta') \Vert \cal{N}(\zeta)) ~=~
(\mu-\mu')^T\left(\frac{\Sigma^{-1}+\Sigma'^{-1}}{2}\right)(\mu-\mu') - \frac12 \cdot
(\Sigma-\Sigma') \bullet (\Sigma^{-1}-\Sigma'^{-1}) \mper
\end{equation*}

Since $\norm{\mu-\mu'}_{\infty} \leq \norm{\zeta - \zeta'}_{\infty} \leq \eps$ and all eigenvalues
of $\Sigma^{-1}$ and $\Sigma'^{-1}$ are at most $1/\delta$, the first term is bounded by
$O_{k,\delta}(\eps^2)$.
%%%%%%%%%%%%%%%%%%%%%%%%%%%%%%%%%%
\newcommand{\adj}{\mathrm{Adj}}
%%%%%%%%%%%%%%%%%%%%%%%%%%%%%%%%%%
For bounding Frobenius product in the second term, note that $\norm{\Sigma -\Sigma'}_{\infty} =
O(\eps)$. For the term, $\Sigma^{-1} -\Sigma'^{-1}$, using the fact that
$\Sigma^{-1} = \frac{\adj(\Sigma)}{\abs{\Sigma}}$, we can write
\[
\Sigma^{-1} -\Sigma'^{-1}
~=~ \frac{\adj(\Sigma) \cdot \abs{\Sigma'} - \adj(\Sigma) \cdot \abs{\Sigma'}}{\abs{\Sigma} \cdot
  \abs{\Sigma'}} \mper
\]
Since all eigenvalues of $\Sigma$ and $\Sigma'$ are at least $\delta$, the determinants
$\abs{\Sigma}$ and $\abs{\Sigma'}$ are at least $\delta^k$.
Each entry of the matrices in the numerator can be viewed as a difference between two multivariate
degree-$2k$ polynomials with $O_{k}(1)$ terms. The two polynomials are identical, except that each has been
perturbed by at most $\eps$ in its variables. Hence, their difference can be at most $O_{k}(\eps)$.

Thus, we obtain that $\norm{\Sigma^{-1} -\Sigma'^{-1}}_{\infty} = O_{k,\delta}(\eps)$ and hence
$(\Sigma-\Sigma') \bullet (\Sigma^{-1}-\Sigma'^{-1}) = O_{k,\delta}(\eps^2)$, which gives the
required bound on $\norm{\cal{N}_d(\zeta) - \cal{N}_d(\zeta')}_1$.
\end{proof}

%%%%%%%%%%%%%%%%%%%%%%%%%%%%%%%%%%%%%%%%%%%%%%%%%%%%%%%%%%%
% L > 0 implies good rounding
%%%%%%%%%%%%%%%%%%%%%%%%%%%%%%%%%%%%%%%%%%%%%%%%%%%%%%%%%%%

\subsection{A Rounding Scheme for Predicates when $L > 0$}\label{Lge0rnd:sbs}

We can now prove that if the value of the above games has a positive limit, then the predicate $f$
admits a non-trivial  approximation. For an instance $\Phi$ of $\maxkcsp(f)$ and a
$d$-dimensional rounding function $\psi: \R^d \to \values$, let $\round_{\psi}(\Phi)$ denote the
expected fraction of constraints satisfied by the rounding algorithm using the function
$\psi$. Recall that a $d$-dimensional \emph{rounding algorithm} is a distribution over functions
$\psi$. We will show that if the value of the game is positive, then there is a rounding scheme such
that $\ExpOp_{\psi}\left[ \round_{\psi}(\Phi) \right] $ is at least  $\rho(f)+\Omega(1)$.
\begin{theorem}\label{sdp-algo:thm}
If $L >0$, then there exists a $(k+1)$-dimensional rounding algorithm
for the basic SDP relaxation of $\maxkcsp(f)$,
such that given an instance $\Phi$ with $\sdpopt(\Phi) \ge 1-\eps$ (for sufficiently small
$\eps > 0$), we have $\ExpOp_{\psi} \left[ \round_{\psi}(\Phi) \right] \geq \rho(f) + L/2$.
\end{theorem}
\begin{proof}
Since $\lim_{p,q \to \infty} \val(\G_{p,q}) = L > 0$, for every $\beta > 0$, for all sufficiently large  $p, q \in
\N$, we have  $\val(\G_{p,q}) \geq L - \beta$. We will fix a sufficiently large $p$ and $q$ as we proceed. By definition of
$\val(\G_{p,q})$,
 there exists a distribution $\Gamma_q$ over functions $\psi_q$, such that for all $\lambda_p \in \Rp$
\[
\pay(\lambda_p, \Gamma_q) ~\ge~ L - \beta \mper
\]
We will use $\Gamma_q$ to design a
$d$-dimensional rounding strategy (recall that we choose $d=k+1$). Since it is really the
$\opay(\cdot,\cdot)$ function that captures the performance of the rounding strategy, we switch to it
via Claim \ref{payoffs-close:clm}.  Specifically, we choose $q$ sufficiently
large so that $c_{k,d,\delta,q} \leq \beta$. Thus the two pay-off functions differ by at most $\beta$ and hence for
all $\lambda_p \in \Rp$
\begin{equation} \label{eqn:algo-start}
\opay(\lambda_p, \Gamma_q) ~\ge~ L - 2 \beta \mper
\end{equation}

Given an instance $\Phi$ of
$\maxkcsp(f)$ and a solution to the SDP in Figure \ref{fig:basic-sdp}, we proceed as follows:
\begin{itemize}
\item[-] For all $i \in [n]$, define vectors:
\begin{align*}
\voneempty &~=~ \vtwoempty \\
\tu_i &~=~ \vtwo{i}{1} - \vtwo{i}{-1} \\
\vone{i} &~=~ \sqrt{1-\delta} \cdot \tu_i + \sqrt{\delta} \cdot {\bf e}_i \mcom
\end{align*}
where $\inbraces{{\bf e}_i}_{i \in [n]}$ form an orthonormal basis, orthogonal to all the vectors
$\inbraces{\vtwo{i}{b}}_{i \in [n], b \in \pmone}$.
\item[-] Sample vectors $\inbraces{{\bf g}_l}_{l \in [d]}$ such that each coordinate of each ${\bf g}_l$ is a standard normal
  variable. Define the vectors $\y_1', \ldots, \y_n' \in \R^d$ such that for each $l \in [d]$,
\[
(\y_i')_l ~=~ \ip{\inparen{\vone{i} - \ip{\vone{i}, \voneempty} \cdot \voneempty}, {\bf g}_l} +  \ip{\vone{i}, \voneempty}
\]
\item[-] Sample $\psi \sim \Gamma_q$. For each $i \in [n]$, assign the variable $x_i$ as
1 with probability $(1+\psi(\y_i'))/2$ and $-1$ with probability $(1-\psi(\y_i'))/2$.
\end{itemize}

For a constraint $C$, let $\sdpopt(C) \in [0,1]$ denote the contribution of the constraint $C$ to
the SDP objective function. For the given instance $\Phi$, we have
$\sdpopt(\Phi) = \Ex{C \in \Phi}{\sdpopt(C)} \geq 1-\eps$
and hence $\Prob{C \in \Phi}{\sdpopt(C) \geq 1-\sqrt{\eps}} \geq 1-\sqrt{\eps}$.
Let $C$ be a constraint such that $\sdpopt(C) \geq 1-\sqrt{\eps}$. Without loss of generality, we
can take $C$ to be on the variables $x_1,\ldots,x_k$ and of the form
$f(x_1\cdot b_1,\ldots,x_k \cdot b_k)$ for $b_1,\ldots,b_k \in \pmone$.

The probability that $C$ is satisfied by the assignment produced by a rounding function $\psi$,
chosen by our rounding scheme is given by
\begin{align*}
\round_{\psi}(C)
&~=~ \rho(f) + \Ex{\y_1',\ldots, \y_k'}{\sum_{S \sub [k] \atop S \neq
    \emptyset} \inparen{\prod_{i \in S}b_i} \cdot \hf(S) \cdot \inparen{\prod_{i \in S} \psi(\y_i')}} \\
&~=~ \rho(f) + \Ex{\y_1',\ldots, \y_k'}{\sum_{S \sub [k] \atop S \neq
    \emptyset} \hf(S) \cdot \inparen{\prod_{i \in S} \psi(b_i
    \cdot \y_i')}} \mcom
\end{align*}
where $b_i \cdot \y_i'$ denotes a vector with each coordinate multiplied by $b_i$, and the second
equality used the fact that the functions $\psi$ are odd.

Let $\zeta_C \in \R^{(k+1)\times (k+1)}$ be
the symmetric moment matrix with $\zeta(0,i) = \ip{b_i \cdot \vone{i}, \voneempty}$ and
$\zeta(i,j) = \ip{b_i \cdot \vone{i}, b_j \cdot \vone{j}}$. Then the variables $(\y_1,\ldots,\y_k) = (b_1 \cdot
\y_1',\ldots,b_k \cdot \y_k')$ are distributed according to the Gaussian process $\cal{N}_d(\zeta_C)$. Thus,
we can write
\begin{align*}
\round_{\psi}(C)
&~=~ \rho(f) + \Ex{\y_1, \ldots, \y_k \sim \cal{N}_d(\zeta_C)}{\sum_{S \sub [k] \atop S \neq
    \emptyset} \hf(S) \cdot \inparen{\prod_{i \in S} \psi(\y_i)}} &~=~ \rho(f) + \opay(\zeta_C, \psi)\mper
\end{align*}
%
% We now relate the above to the value of the games $\G_{p,q}$.
The variables $\vartwo{[k]}{\alpha}$
define a probability distribution, say $\nu_0$ on $\pmone^k$. Let $\nu$ be the distribution on
$\pmone^k$ such that for any $x \in \pmone^k$,
\[
\nu(x_1,\ldots,x_k) ~=~ \nu_0(b_1 \cdot x_1, \ldots, b_k \cdot x_k) \mper
\]
Then $\Prob{x \sim \nu}{f(x) = 1} \ge 1-\sqrt{\eps}$
and  for the corresponding moment matrix $\zeta(\nu)$, we have
$\zeta(\nu)(i,j) = \ip{b_i \cdot \tu_i,b_j \cdot \tu_j}$ and $\zeta(\nu)(0,i) = \ip{b_i \cdot \tu_i,
  \voneempty}$ for all $i,j \in [k]$. From the definition of the vectors $\vone{i}$ and the matrix
$\zeta_C$ above, we have that
\[
\zeta_C ~=~ (1-\delta) \cdot \zeta(\nu) + \delta \cdot \mathbb{I} \mper
\]
However, $\zeta_C$ does not lie in the body $\C_{\delta}(f)$ since $\nu$ is not entirely supported on
$f^{-1}(1)$. We thus, consider the distribution $\nu'$, which is $\nu$ conditioned on the output
being in $f^{-1}(1)$. Also, we define the matrix $\zeta_C' \in \C_{\delta}(f)$ as
\[
\zeta_C' ~\defeq~ (1-\delta) \cdot \zeta(\nu') + \delta \cdot \mathbb{I} \mper
\]
Since $\nu$ satisfies $C$ with probability at least $1-\sqrt{\eps}$, we have
$\norm{\nu - \nu'}_1 = O(\sqrt{\eps})$. Also, this gives that
$\norm{\zeta_C-\zeta_C'}_{\infty} \leq O(\sqrt{\eps})$.
By Lemma \ref{payoff-continuous:lm}, we
have for $\zeta_C$ and $\zeta_C'$ as above
$$  \left| \opay(\zeta_C, \psi) - \opay(\zeta_C', \psi) \right| ~=~ O_{k,d,\delta}(\sqrt{\eps}) \mper$$
%%
%\[
%\abs{
%\Ex{\y_1, \ldots, \y_k \sim \cal{N}_d(\zeta_C)}{\sum_{S \sub [k] \atop S \neq
%    \emptyset} \hf(S) \cdot \inparen{\prod_{i \in S} \psi(\y_i)}}
%-
%\Ex{\y_1, \ldots, \y_k \sim \cal{N}_d(\zeta_C')}{\sum_{S \sub [k] \atop S \neq
%    \emptyset} \hf(S) \cdot \inparen{\prod_{i \in S} \psi(\y_i)}}
%} = O_{k,d,\delta}(\sqrt{\eps}) \mper
%\]
%%
We now analyze $\round_{\psi}(\Phi) = \Ex{C \in \Phi}{\round_{\psi}(C)}$.
Let $\Phi'$ denote the instance
restricted to the constraints $C$ such that $\sdpopt(C) \ge 1-\sqrt{\eps}$.  Since
$\Prob{C \in \Phi}{\sdpopt(C) \geq 1-\sqrt{\eps}} \geq 1-\sqrt{\eps}$, we have that
\[
\abs{\round_{\psi}(\Phi) - \round_{\psi}(\Phi')} ~=~ O(\sqrt{\eps}) \mper
\]
Finally, to relate the above to the value of one of the games $\G_{p,q}$, we let $\lambda$ be the
distribution on $\C_{\delta}(f)$ obtained by sampling a random $C \in \Phi'$ and taking the matrix
$\zeta_C'$. Using the above, we get that
\begin{align*}
\ExpOp_{\psi \sim \Gamma_q} \left[ \round_{\psi}(\Phi) - \rho(f) \right]
&~\ge~ \ExpOp_{\psi \sim \Gamma_q}   \left[ \round_{\psi}(\Phi') - \rho(f) \right]      - O(\sqrt{\eps}) \\
&~=~ \ExpOp_{\psi \sim \Gamma_q} \left[ \Ex{C \sim \Phi'}{\round_{\psi}(C) - \rho(f)}  \right] - O(\sqrt{\eps}) \\
&~=~ \ExpOp_{\psi \sim \Gamma_q} \left[ \Ex{C \sim \Phi'}{\opay(\zeta_C, \psi)}  \right] - O(\sqrt{\eps}) \\
&~\ge~ \ExpOp_{\psi \sim \Gamma_q} \left[ \Ex{C \sim \Phi'}{\opay(\zeta_C', \psi)}  \right] - O_{k,d,\delta}(\sqrt{\eps}) \\
&~=~ \ExpOp_{\psi \sim \Gamma_q} \left[ \Ex{\zeta \sim \lambda}{\opay(\zeta, \psi)} \right] - O_{k,d,\delta}(\sqrt{\eps})  \\
%&~\ge~ \ExpOp_{\psi \sim \Gamma_q} \left[ \ExpOp_{C \sim \Phi'} ~\Ex{\y_1, \ldots, \y_k \sim \cal{N}_d(\zeta_C')}{\sum_{S \sub [k] \atop S \neq
%    \emptyset} \hf(S) \cdot \inparen{\prod_{i \in S} \psi(\y_i)}}\right] - O_{k,d,\delta}(\sqrt{\eps}) \\
%&~=~ \ExpOp_{\psi \sim \Gamma_q} \left[ \ExpOp_{\zeta \sim \lambda} ~\Ex{\y_1, \ldots, \y_k \sim
%    \cal{N}_d(\zeta)}{\sum_{S \sub [k] \atop S \neq
%    \emptyset} \hf(S) \cdot \inparen{\prod_{i \in S} \psi(\y_i)}}\right] - O_{k,d,\delta}(\sqrt{\eps})  \\
&~=~ \opay(\lambda,\Gamma_q) - O_{k,d,\delta}(\sqrt{\eps})
\mper
\end{align*}
The above almost looks like the pay-off for our game, except for the fact that the distribution
$\lambda$ may not belong to the set of distributions $\Rp$ for any $p \in \N$. However, the
sets $R_p$ get arbitrarily dense in $\C_{\delta}(f)$ as $p$ increases. Also the probabilities for
distributions in $\Rp$ are allowed to be multiples of $1/2^p$, which gets arbitrarily small as $p$
increases. Hence for any $\beta_0 > 0$, we could have chosen large enough $p$ beforehand so that
%Since $\lambda$ is supported on finitely many $\zeta \in \C_{\delta}(f)$, for any
%$\beta_0 > 0$, it is possible to find a large enough $p \in \N$ and
there is a distribution $\lambda_p \in
\Rp$ such that:
\begin{itemize}
\item[-] There exists a map from  the support of $\lambda$ to that of $\lambda_p$.
\item[-] For each $\zeta$ in the support of $\lambda$, let $\zeta'$ denote its image according to
  the above map. Then $\norm{\zeta - \zeta'}_{\infty} \leq \beta_0$ and the probabilities
  $\lambda(\zeta)$ and $\lambda_p(\zeta')$ differ by at most $\beta_0$.
\end{itemize}
Also, since by Lemma \ref{payoff-continuous:lm} the function $\opay(\zeta,\psi_q)$ is $O_{k,d,\delta}(1)$-Lipschitz
in the argument $\zeta$, we have that for a sufficiently small choice of $\beta_0 > 0$,
\[
\abs{\opay(\lambda,\Gamma_q) - \opay(\lambda_p, \Gamma_q)} \leq \beta \mper
\]
%
% By Lemma \ref{payoff-continuous:lem} the pay-off function function above is continuos in the
% argument $\zeta'$. Since  $\zeta' \in \C_{\delta}(f)$ and the sets $R_p$ are dense in
% $\C_{\delta}(f)$, we can choose a sufficiently large $p$ (depending only on $q$, $k$ and
% $\beta$, but not on $\Gamma_q$) and a matrix $\zeta'' \in R_p$ such that
% \[
% \abs{
% \Ex{\psi \sim \Gamma_q  \atop z_1, \ldots, z_k \sim \cal{N}(\zeta')}{\sum_{S \sub [k] \atop S \neq
%     \emptyset} \hf(S) \cdot \inparen{\prod_{i \in S} \psi(z_i)}}
% -
% \Ex{\psi \sim \Gamma_q  \atop z_1, \ldots, z_k \sim \cal{N}(\zeta'')}{\sum_{S \sub [k] \atop S \neq
%     \emptyset} \hf(S) \cdot \inparen{\prod_{i \in S} \psi(z_i)}}
% } ~\le~ \beta \mper
% \]
Recall that Equation \eqref{eqn:algo-start} gave us
\[
\opay(\lambda_p, \Gamma_q) ~\ge~  L - 2\beta \mper
\]
Combining all the above inequalities, we have that
\[
\ExpOp_{\psi}\left[ \round_{\psi}(\Phi) - \rho(f) \right]
~\ge~ \opay(\lambda_p,\Gamma_q) - \beta - O_{k,d,\delta}(\sqrt{\eps})
~\ge~ L - 3\beta - O_{k,d,\delta}(\sqrt{\eps}) \mper
\]
% and hence expected number of constraints satisfied by the rounding strategy is at least
% \[
% \Ex{C \in \Phi}{\round(C)}
% ~\ge~ (1-\sqrt{\eps}) \cdot \inparen{\rho(f) + L - 2\beta - O(\sqrt{\eps})}
% ~\ge~ \rho(f) + L - 2\beta - O(\sqrt{\eps}) \mper
% \]
Choosing $\beta \leq L/16$ and $\eps =  o_{k,d,\delta}(L^2)$ gives that
$\ExpOp_{\psi} \left[   \round_{\psi}(\Phi) - \rho(f)  \right]   \geq L/2$ as claimed.
\end{proof}

%%%%%%%%%%%%%%%%%%%%%%%%%%%%%%%%%%%%%%%%%%%%%%%%%%%%%%%%%%%
% L = 0 implies existence of measure with desired properties
%%%%%%%%%%%%%%%%%%%%%%%%%%%%%%%%%%%%%%%%%%%%%%%%%%%%%%%%%%%
\subsection{A Characterization of Predicates with $L = 0$}

We are now left with the case: $L = 0$ (since we always have that $L \ge 0$). We will show that the
condition $L = 0$ implies the existence of a probability measure $\Lambda$ on $\C(f)$ satisfying
Equation~\ref{main:eqn} in  Theorem~\ref{sdp:thm}.
% This will give that when such a measure \emph{does not
% exist} (Case 2 in Theorem \ref{sdp:thm}), then we must have $L > 0$ and can thus apply Theorem
% \ref{sdp-algo:thm}.
We next prove the following.

\begin{theorem}\label{sdp-measure:thm}
If $L=0$, then there exists a probability measure $\Lambda$ on $\C(f)$
such that for all $t \in [k]$, and a uniformly random choice of $S$ with  $|S|=t$,
$\pi : S\to S$ and $b\in \pmone^{S}$, the following signed measure on $(t+1) \times (t+1)$ matrices:
\begin{equation*}
\Lambda^{(t)} ~:=~
\ExpOp_{|S| = t} ~\ExpOp_{\pi : S \to S} ~\Ex{b \in \pmone^{|S|}}{ \hf(S) \cdot  \inparen{\prod_{i\in S} b_i}
  \cdot \Lambda_{S,\pi,b}}
\end{equation*}
is identically zero.
\end{theorem}

We refer to a measure $\Lambda$ which satisfies the above condition, as a {\deffont vanishing measure}.
We will obtain this measure by considering limits of the various strategies for \Dev in the games
$\G_{p,q}$. We first consider the limit for each $p$ as $q \to \infty$.

\begin{lemma}\label{optdis:lm}
For each $p \in \N$, there exists a limiting distribution $\Lambda_p$ over $\Rp$
such that for every $q$ and $\Gamma_q$,
\[ \pay(\Lambda_p,\Gamma_q) ~\le~ r_p \mcom\]
where $r_p = \lim_{q \to \infty} \cal{V}(p,q)$.
\end{lemma}
\begin{proof}
For any row $p$ we have a sequence of distributions $\inbraces{\Lambda_{p,q}}_{q \in \N}$
such that for all $q$ and $\Gamma_q$,
\[
\pay(\Lambda_{p,q}, \Gamma_q) ~\le~ \cal{V}(p,q) ~\le~ r_p
\]
where the second inequality used the fact that the numbers $\cal{V}(p,q)$ are non-decreasing in $q$
(see Lemma~\ref{lmt:lm}).
Also, for a fixed $p$, each $\Lambda_{p,q}$ can be viewed as a vector in $[0,1]^{\abs{\Rp}}$ where
$\Rp$ is the class of distributions over $R_p$ with probabilities being integer multiples of $1/2^p$
(and thus $\abs{\Rp} \leq (2^p+1)^{\abs{R_p}}$).
Hence by the Bolzano-Weierstrass Theorem,
the sequence $\inbraces{\Lambda_{p,q}}_{q \in \N}$ has a convergent subsequence
with a limit point, which we take to be $\Lambda_p$.
Since each strategy $\Gamma_q$ can also be
viewed as a strategy $\Gamma_{q'}$ for any $q' \ge q$, we have that
$\pay(\Lambda_{p,q'}, \Gamma_q) \le r_p$. Taking the limit as $q' \to
\infty$ according to the above convergent subsequence, we have that for all $q$ and $\Gamma_q$,
\[
\pay(\Lambda_{p}, \Gamma_q) ~=~ \lim_{q' \to \infty}\pay(\Lambda_{p,q'}, \Gamma_q) ~~\le~~  r_p \mper
\]
\end{proof}

%
%For the remainder of this section, we will consider the expression of the pay-off function without
%the absolute value. For $\Lambda$ which is a \emph{distribution over distributions} on
%$\C_{\delta}(f)$ and for $\Gamma$ which is a distribution over functions $\psi: \R^d \to \values$,
%we define $\eval(\Lambda,\Gamma)$ as the pay-off expression without the absolute values.
%%
%\begin{align*}
%\eval(\Lambda,\Gamma) &~\defeq~
%\ExpOp_{\lambda \sim \Lambda} ~\ExpOp_{\psi \sim \Gamma}
%~\ExpOp_{\zeta \sim \lambda} ~\Ex{\y_1, \ldots, \y_k \sim \cal{N}_d(\zeta)}{\sum_{S \neq
%      \emptyset} \HF(S) \cdot \prod_{i\in S}\psi(\y_i)}  \mper
%\end{align*}
%%
%As before, for a fixed function $\psi$, we use $\eval(\Lambda,\psi)$ to denote the expression where
%we omit the outer expectation over $\Gamma$. The expression $\eval(\zeta,\psi)$ is defined similarly
%for a fixed $\zeta \in \C_{\delta}(f)$.
%%
%Note that for $\Lambda_p$ as given by Lemma \ref{optdis:lm} and any $\Gamma_q$, we have that
%%
%\[
%\abs{\eval(\Lambda_p,\Gamma_q)} ~\le~
%\ExpOp_{\lambda \sim \Lambda_p} ~\ExpOp_{\psi \sim \Gamma_q}
%~\abs{\ExpOp_{\zeta \sim \lambda} ~\Ex{\y_1, \ldots, \y_k \sim \cal{N}_d(\zeta)}{\sum_{S \neq
%      \emptyset} \HF(S) \cdot \prod_{i\in S}\psi(\y_i)} }
%~\le~ r_p \mper
%\]
%%

Since for the purpose of computing $\pay(\Lambda_p,\Gamma_q)$, we can merge the expectations over
$\lambda \sim \Lambda_p$ and $\zeta \sim \lambda$, we will now simply consider each $\Lambda_p$ to be
a probability measure over $R_p \sub \C_{\delta}(f)$.
We will obtain the desired probability measure $\Lambda$ by taking a limit of the measures
$\Lambda_p$ obtained above. However, since the measures $\Lambda_p$ are supported on sets $R_p$ with
growing size, we will need to be somewhat careful in taking the limit and will use the weak*
topology to do so.

Since $\lim_{p \to \infty} r_p = L = 0$, the function $\lim_{p \to \infty} \pay(\Lambda_p,\Gamma_q)
\leq 0$ for any $\Gamma_q$. In particular, $\lim_{p \to \infty} \pay(\Lambda_p,\psi_q) \leq 0$ for any $q
\in \N$ and function $\psi_q: \R^d \to \values$ which is constant on the cells of the partition
$\cal{P}_q$ and is 0 outside the box $\bbox$. We use this to prove the following lemma.
%\footnote{We remark that Lemma \ref{eval-zero:lm} is the main reason why we need  to have an absolute value in the
%pay-off of our games (and hence restrict ourselves to strong approximation resistance). Without the
%absolute value, we can only prove $\eval(\Lambda,\psi_q) \leq 0$ for all
%$\psi_q$, which does not suffice for our purpose.
%}.
%
\begin{lemma} \label{eval-zero:lm}
There exits a probability measure $\Lambda$ on $\C_{\delta}(f)$ such that for all $q \in \N$ and all
functions $\psi_q$, we have $\pay(\Lambda,\psi_q) \leq 0$.
\end{lemma}
\begin{proof}
Note that $\C_{\delta}(f)$ is a closed and bounded subset of $\R^{(k+1)^2}$ and is hence compact by
the Heine-Borel Theorem. Also, by Theorem \ref{wk*:thm}, we have
that the space of probability measures on $\C_{\delta}(f)$ is compact and metrizable in the weak*
topology.

From the compactness, we obtain that the infinite sequence $\inbraces{\Lambda_p}_{p \in \N}$ (viewed as a
sequence of probability measures on $\C_{\delta}(f)$) has a convergent subsequence with a limit
point, say $\Lambda$. By the definition of weak* topology, we have that for any continuous function
$h: \C_{\delta}(f) \to \R$, taking a limit over the above subsequence, we get
\[
\lim_{p \to \infty} \int h(\zeta) d \Lambda_p(\zeta) ~=~ \int h(\zeta) d \Lambda(\zeta) \mper
\]
Also, note that by Lemma \ref{payoff-continuous:lm} the function
$\pay(\zeta,\psi)$ is $O_{k,d,\delta}(1)$-Lipschitz continuous, when viewed as a function of
$\zeta$. Hence, taking limits according to the above subsequence, we get that for all $q \in \N$ and
functions $\psi_q$
\[
0
~\geq~ \lim_{p \to \infty} \Ex{\zeta \sim \Lambda_p}{\pay(\zeta,\psi_q)}
~=~ \Ex{\zeta \sim \Lambda}{\pay(\zeta,\psi_q)}
~=~ \pay(\Lambda,\psi_q)
\]
as claimed.
\end{proof}
To show that this implies the properties claimed in Theorem \ref{sdp-measure:thm} for the limiting
measure $\Lambda$, we think of the function %$\pay(\Lambda,\psi)$ defined as
\[
\pay(\Lambda,\psi)
~=~
\ExpOp_{\zeta \sim \Lambda}
~\Ex{\y_1,\ldots,\y_k \sim \cal{N}_d(\zeta)}{\sum_{S \neq \emptyset} \hf(S) \cdot \prod_{i \in S}
  \psi(\y_i) \cdot  \indic_q\left( \{ \y_i | i \in S\}\right) } \mcom
\]
as a degree-$k$ ``multi-linear polynomial'' in the (infinite
set of) variables $\psi(\y)$ for all $\y \in \R^d$.
The intuition is that since the polynomial stays upper-bounded by $0$
for all ``assignments'' $\psi$ to the variables, all its ``coefficients''
must be zero (see Lemma \ref{multi-linear:lem}).

Of course, the above is not a formal argument since the number of variables is
infinite. To formalize this, we define the following quantity, which plays the role of the
``coefficient'' for the term $\prod_{i = 1}^t \psi(\y_i)$ for $t \leq k$. Note that in the expression
for $\pay(\Lambda,\psi)$, the term $\prod_{i=1}^t \psi(\y_i)$ can arise for any $S \sub [k]$ with $|S| = t$
\ie $(\y_1,\ldots, \y_t)$ can be any ordering of any subset of size $t$ for the points $\z_1,\ldots,
\z_k$
which we sample for computing the pay-off. Also, we can also get a term involving $\psi(\y_i)$ if
$\y_i = - \z_j$ for some $j \in [k]$, since we have the constraints $\psi(-\z) = -\psi(\z)$. Taking
these into account, we define the following.

\begin{definition}
We  define $\theta^{(t)} : (\R^d)^t \to \R$ on formal variables $\{\y_1,...,\y_t\}$ as follows:
\begin{equation}
\theta^{(t)}(\y_1,\ldots,\y_t) ~\defeq~
\sum_{|S| = t} \ExpOp_{\pi: [t] \to [t]} ~\ExpOp_{b \in \pmone^t} \Ex{\zeta \sim\Lambda}{\hf(S) \cdot
  \inparen{\prod_{i=1}^t b_i} \cdot \gauss{t,d}{(\y_1,\ldots, \y_t)}{\zeta_{S,\pi,b}}} \mcom
\end{equation}
where $\gauss{t,d}{\cdot}{\zeta_{S,\pi,b}}$ is the joint density of $t$ correlated Gaussians in
$\R^d$, with different coordinates being independent and the moments for each coordinate
given by the appropriate submatrix $\zeta_S$ of $\zeta$ permuted according to $\pi$ and modified
according to the signs specified by $b$.  Also, $\Lambda$ is the limiting measure as above.
\end{definition}
The following properties follow easily from the definition of the  function $\theta^{(t)}$.
\begin{claim}\label{theta-invariance:clm}
For all $t \in [k]$ and for all $(\y_1,\ldots,\y_t) \in (\R^d)^t$, we have that
\begin{itemize}
\item[-] For all permutations $\pi': [t] \to [t]$, ~$\theta^{(t)}(\pi' (\y_1,\ldots,\y_t)) =
  \theta^{(t)}(\y_1,\ldots,\y_t)$.
\item[-] For all $b' \in \pmone^t$, ~$\theta^{(t)}(b_1'\y_1,\ldots,b_t'\y_t) = \inparen{\prod_{i=1}^t
    b_i'} \cdot \theta^{(t)}(\y_1,\ldots,\y_t)$
\end{itemize}
\end{claim}
\begin{proof}
By definition of $\zeta_{S,\pi,b}$, we have
$\gauss{t,d}{(\y_1,\ldots, \y_t)}{\zeta_{S,\pi,b}} = \gauss{t,d}{\pi(b_1\y_1,\ldots,
  b_t\y_t)}{\zeta_{S}}$. We can then write
\[
\theta^{(t)}(\y_1,\ldots,\y_t) ~=~
\sum_{|S| = t} \ExpOp_{\pi: [t] \to [t]} ~\ExpOp_{b \in \pmone^t} \Ex{\zeta \sim\Lambda}{\hf(S) \cdot
  \inparen{\prod_{i=1}^t b_i} \cdot \gauss{t,d}{\pi(b_1\y_1,\ldots, b_t\y_t)}{\zeta_{S}}} \mper
\]
Since the expression already involves expectation over all permutations $\pi$ of each tuple,
replacing $(\y_1,\ldots,\y_t)$ by $\pi(\y_1,\ldots,\y_t)$ does not change the value of the function.
Similarly, for any $b' \in \pmone^t$, we get
\begin{align*}
\theta^{(t)}(b_1'\y_1,\ldots,b_t'\y_t)
&~=~
\sum_{|S| = t} \ExpOp_{\pi: [t] \to [t]} ~\ExpOp_{b \in \pmone^t} \Ex{\zeta \sim\Lambda}{\hf(S) \cdot
  \inparen{\prod_{i = 1}^t b_i} \cdot \gauss{t,d}{\pi(b_1b_1'\y_1,\ldots, b_t b_t'\y_t)}{\zeta_{S}}} \\
&~=~
\sum_{|S| = t} \ExpOp_{\pi: [t] \to [t]} ~\ExpOp_{b \in \pmone^t} \Ex{\zeta \sim\Lambda}{\hf(S) \cdot
  \inparen{\prod_{i = 1}^t b_i b_i'} \cdot \gauss{t,d}{\pi(b_1\y_1,\ldots, b_t \y_t)}{\zeta_{S}}} \mcom
\end{align*}
which equals $ \inparen{\prod_{i=1}^t b_i'} \cdot \theta^{(t)}(\y_1,\ldots,\y_t)$ as claimed.
\end{proof}
The next claim shows that the functions $\theta^{(t)}$ indeed provide the right notion of
``coefficients'' when we think of the function $\pay(\Lambda,\psi)$
as a polynomial in the values $\psi(\z)$.
\begin{claim}\label{payoff-theta:clm}
Let $\Lambda$ be the measure as above and let $\psi = \psi_q: \R^d \to \values$ be an odd
function.
Then,
\[
\pay(\Lambda,\psi) ~=~ \sum_{t=1}^k \int_{(\R^d)^t} \theta^{(t)} (\y_1,\ldots,\y_t) \cdot
\inparen{\prod_{i=1}^t \psi(\y_i)} \cdot \indic_q\left( \{ \y_i | i \in [t] \}\right) \ d\y_1\ldots d\y_t \mper
\]
\end{claim}
\begin{proof}
Since $\psi$ is  measurable and the integral above is bounded,
we will freely switch the order of integrals in the
argument below. We have
% We begin with the expression for the function $\eval(\Lambda,\psi)$.
%
\begin{align*}
&~\pay(\Lambda,\psi) \\
&~=~ \ExpOp_{\zeta \sim \Lambda}
~\Ex{\z_1,\ldots,\z_k \sim \cal{N}_d(\zeta)}{\sum_{S \neq \emptyset} \hf(S) \cdot \inparen{\prod_{i \in S}
  \psi(\z_i)} \cdot \indic_q\left( \{ \z_i | i \in S\}\right) } \\
&~=~ \sum_{t=1}^k\sum_{|S| = t} \hf(S) \cdot \ExpOp_{{\zeta \sim \Lambda}}
~\Ex{\y_1,\ldots,\y_t \sim \cal{N}_d(\zeta_S)}{\inparen{\prod_{i=1}^t \psi(\y_i)} \cdot \indic_q\left( \{ \y_i | i \in [t]\}\right)  } \\
&~=~ \sum_{t=1}^k\sum_{|S| = t} \hf(S) \cdot \ExpOp_{{\zeta \sim \Lambda}}
~\int_{(\R^d)^t} \gauss{t,d}{(\y_1,\ldots,\y_t)}{\zeta_S} \cdot \inparen{\prod_{i=1}^t \psi(\y_i)} \cdot \indic_q\left( \{ \y_i | i \in [t]\}\right) ~d\y_1
\ldots d\y_t
\end{align*}
Symmetrizing the expression over the sign flips $b \in \{-1,1\}^t$,
\begin{align*}
&~\pay(\Lambda,\psi) \\
&~=~ \sum_{t=1}^k\sum_{|S| = t} \hf(S) \cdot \ExpOp_{b \in \pmone^t} \ExpOp_{{\zeta \sim \Lambda}}
~\int_{(\R^d)^t} \gauss{t,d}{(b_1 \y_1,\ldots, b_t \y_t)}{\zeta_S} \cdot \\
& \hspace{6.5cm}
\inparen{\prod_{i=1}^t \psi( b_i \y_i)  } \cdot \indic_q\left( \{ b_i \y_i | i \in [t]\}\right) ~d\y_1
\ldots d\y_t \\
&~=~ \sum_{t=1}^k\sum_{|S| = t} \hf(S) \cdot \ExpOp_{b \in \pmone^t} \ExpOp_{{\zeta \sim \Lambda}}
~\int_{(\R^d)^t} \gauss{t,d}{(b_1 \y_1,\ldots, b_t \y_t)}{\zeta_S} \cdot \\
& \hspace{6.5cm}
\inparen{\prod_{i=1}^t b_i} \cdot \inparen{\prod_{i=1}^t \psi( \y_i)   } \cdot \indic_q\left( \{ \y_i | i \in [t]\}\right)  ~d\y_1
\ldots d\y_t \mcom
\end{align*}
where the last equality used the fact that $\psi$ is odd. Finally, we note that the term
$\prod_{i=1}^t \psi(\y_i)$ can arise from any permutation of the tuple $(\y_1,\ldots,\y_t)$. We thus
re-write the expression above as
\begin{align*}
&~\pay(\Lambda,\psi) \\
&= \sum_{t=1}^k \int_{(\R^d)^t} \sum_{|S| = t}  \ExpOp_{\pi: [t] \to [t]
    \atop b \in \pmone^t} ~\Ex{\zeta \sim \Lambda}{ \hf(S) \cdot \inparen{\prod_{i=1}^t b_i} \cdot
\gauss{t,d}{\pi(b_1 \y_1,\ldots, b_t \y_t)}{\zeta_S}} \cdot \\
& \hspace{6.5cm} \inparen{\prod_{i=1}^t \psi( \y_i)} \cdot \indic_q\left( \{ \y_i | i \in [t]\}\right) \ d\y_1 \ldots d\y_t \\
&= \sum_{t=1}^k \int_{(\R^d)^t} \sum_{|S| = t}  \ExpOp_{\pi: S \to S
    \atop b \in \pmone^S} \Ex{\zeta \sim \Lambda}{ \hf(S) \cdot \inparen{\prod_{i \in S} b_i} \cdot
\gauss{t,d}{(\y_1,\ldots, \y_t)}{\zeta_{S,\pi,b}}} \cdot \\
& \hspace{6.5cm} \inparen{\prod_{i=1}^t \psi(\y_i)} \cdot \indic_q\left( \{ \y_i | i \in [t]\}\right) \ d\y_1 \ldots d\y_t \\
&= \sum_{t=1}^k \int_{(\R^d)^t} \theta^{(t)} (\y_1,\ldots,\y_t) \cdot
\inparen{\prod_{i=1}^t \psi(\y_i)}\cdot \indic_q\left( \{ \y_i | i \in [t]\}\right) \  d\y_1\ldots d\y_t \mcom
\end{align*}
as claimed.
\end{proof}

We next show that $\theta^{(t)}$ is a ``nice'' function. For this we shall need to use the fact
that $\Lambda$ is a measure over $\C_{\delta}(f)$, and that matrices in
$\C_{\delta}(f)$ have each eigenvalue at least $\delta$.

\begin{lemma}\label{smth:lm}
For all $t\in[k]$, $\theta^{(t)}$ is bounded \ie $\norm{\theta^{(t)}}_\infty \le O_{k,d,\delta}(1)$,
and it is $O_{k,d,\delta}(1)$-Lipschitz.
\end{lemma}
\begin{proof}
We first argue that $\theta^{(t)}$ is bounded. The Gaussian density
$\gauss{t,d}{\cdot}{\zeta_{S,\pi,b}}$ is at most $\frac{1}{(2\pi)^{td/2} \abs{\Sigma}^{d/2}}$ where
$\Sigma$ is the covariance matrix associated with $\zeta_{S,\pi,b}$ with $\Sigma_{ij} = \zeta_{S,\pi,b}(i,j) -
\zeta_{S,\pi,b}(0,i) \cdot \zeta_{S,\pi,b}(0,j)$, and $\abs{\Sigma}$ denotes the determinant. Since $\zeta \in
\C_{\delta}(f)$, all the eigenvalues of $\Sigma$ are at least $\delta$ and hence $\abs{\Sigma} \geq
\delta^t$. Also, since $\abs{\hf(S)} \leq 1$, we get
\[
\norm{\theta^{(t)}} ~\le~ \binom{k}{t} \cdot \frac{1}{(2\pi)^{td/2}
  \cdot \delta^{td/2}} ~\le~ \frac{1}{\delta^{kd/2}} \mper
\]
Let $\Sigma$ be the covariance matrix as above and $\mu$ be the vector of means with $\mu_i
~=~ \zeta_{S,\pi,b}(0,i)$. Also, for $l \in [d]$, let $\y^{(l)} \in \R^t$ denote the vector
$((\y_1)_l,\ldots,(\y_t)_l)$ obtained by taking the $l^{th}$ coordinates of $\y_1,\ldots,\y_t$. The
Gaussian density $\gauss{t,d}{(\y_1,\ldots,\y_t)}{\zeta_{S,\pi,b}}$ can then be written as
\[
\gauss{t,d}{(\y_1,\ldots,\y_t)}{\zeta_{S,\pi,b}}
~=~
\prod_{l = 1}^d \inparen{ \frac{1}{(2\pi)^{t/2}\abs{\Sigma}^{1/2}} \cdot \exp\inparen{-\frac12 \cdot
 (\y^{(l)} - \mu)^T \Sigma^{-1} (\y^{(l)}-\mu) }} \mper
\]
The density is a function on $\R^{dt}$. The gradient on the coordinates corresponding to $\y^{(l)}$
can be written as
\[
\insquare{\nabla \inparen{\gauss{t,d}{\y}{\zeta_{S,\pi,b}}}}_{l} ~=~
\frac{1}{\inparen{2 \pi}^{td/2} \cdot \abs{\Sigma}^{d/2}} \cdot \inparen{-\Sigma^{-1} (\y^{(l)}-\mu)} \cdot \prod_{l=1}^d\exp\inparen{-\frac12 \cdot
  (\y^{(l)}-\mu)^T \Sigma^{-1} (\y^{(l)}-\mu)} \mper
\]
Since $\Sigma^{-1}$ is positive semidefinite, we can define a matrix $\Sigma^{-1/2}$. Also, since
$\Sigma$ has eigenvalues at least $\delta$, we have
$\norm{\Sigma^{-1} (\y^{(l)}-\mu)} \leq \frac{1}{\sqrt{\delta}} \cdot \norm{\Sigma^{-1/2} (\y^{(l)}-\mu)}$.
Using this, we can bound the norm of gradient in the coordinates corresponding to $\y^{(l)}$ as
\begin{align*}
\norm{\insquare{\nabla \inparen{\gauss{t,d}{\y}{\zeta_{S,\pi,b}}}}_l}
&~=~
\frac{1}{\inparen{2 \pi}^{td/2} \cdot \abs{\Sigma}^{d/2}} \cdot \norm{\Sigma^{-1} (\y^{(l)}-\mu)} \cdot
\prod_{l=1}^d\exp\inparen{-\frac12 \cdot \norm{\Sigma^{-1/2} (\y^{(l)}-\mu)}^2} \\
&~\le~
\frac{\norm{\Sigma^{-1/2} (\y^{(l)}-\mu)}}{\inparen{2 \pi}^{td/2} \cdot \abs{\Sigma}^{d/2} \cdot \sqrt{\delta}}
\cdot \exp\inparen{-\frac12 \cdot \sum_{l=1}^d\norm{\Sigma^{-1/2}
    (\y^{(l)}-\mu)}^2} \mper
\end{align*}
This bounds the norm of the gradient as
\begin{align*}
\norm{\nabla \inparen{\gauss{t,d}{\y}{\zeta_{S,\pi,b}}}}^2
&~\le~
\frac{\sum_{l=1}^d \norm{\Sigma^{-1/2} (\y^{(l)}-\mu)}^2}{\inparen{2 \pi}^{td} \cdot \abs{\Sigma}^{d} \cdot \delta}
\cdot \exp\inparen{-\sum_{l=1}^d\norm{\Sigma^{-1/2}
    (\y^{(l)}-\mu)}^2} \\
&~\le~ \frac{1}{\inparen{2 \pi}^{td} \cdot \abs{\Sigma}^{d} \cdot \delta} \mcom
\end{align*}
where we used the fact that the function $x \cdot \exp(-x)$ is bounded above by 1. Using the
above, we obtain a bound on the gradient of $\theta^{(t)}$ as
\[
\norm{\nabla \theta^{(t)}} ~\le~ \binom{k}{t} \cdot \frac{1}{(2\pi)^{td/2}
  \cdot \delta^{(td+1)/2}}
~\le~ \frac{1}{\delta^{(kd+1)/2}} \mper
\]
Hence, $\theta^{(t)}$ is $C$-Lipschitz, with $C \leq  (1/\delta)^{(kd+1)/2}$.
\end{proof}
Using the above properties and the fact that $\pay(\Lambda,\psi_q) \leq 0$ for all $q \in \N$ and all
functions $\psi_q$, we can in fact show that the functions $\theta^{(t)}$ must in fact be
identically zero on the entire box $(\bbox)^t$.
\begin{lemma}\label{theta-zero:lm}
For all $t \in [k]$ and all $\y_1,\ldots,\y_t \in \bbox$, we have $\theta^{(t)}(\y_1,\ldots,\y_t) = 0$.
\end{lemma}
\begin{proof}
Let $H$ denote the space $[0,1] \times [-1,1]^{d-1}$. By Claim \ref{theta-invariance:clm}, we only
need to show $\theta^{(t)}(\y_1,\ldots,\y_t) = 0$ for all $\y_1,\ldots, \y_t \in H$, since changing
the sign of any input $\y_i$ only changes the sign of $\theta^{(t)}$. Also, by Claim
\ref{payoff-theta:clm}, we have that for any odd function $\psi_q : \R^d \to \values$, which is 0
outside $\bbox$,
\begin{align*}
\pay(\Lambda,\psi_q)
&~=~
\sum_{t=1}^k \int_{(\R^d)^t} \theta^{(t)} (\y_1,\ldots,\y_t) \cdot
\inparen{\prod_{i=1}^t \psi_q(\y_i)}  \cdot \indic_q\left( \{ \y_i | i \in [t]\}\right) \ d\y_1\ldots d\y_t \\
&~=~
\sum_{t=1}^k \int_{(\bbox)^t} \theta^{(t)} (\y_1,\ldots,\y_t) \cdot
\inparen{\prod_{i=1}^t \psi_q(\y_i)}\cdot \indic_q\left( \{ \y_i | i \in [t]\}\right) \  d\y_1\ldots d\y_t \\
&~=~
\sum_{t=1}^k 2^t \cdot  \int_{H^t} \theta^{(t)} (\y_1,\ldots,\y_t) \cdot
\inparen{\prod_{i=1}^t \psi_q(\y_i)} \cdot \indic_q\left( \{ \y_i | i \in [t]\}\right) \ d\y_1\ldots d\y_t \mper
\end{align*}
The second equality above used the fact that $\psi_q$ is 0 outside $\bbox$.  The last equality
used that by Claim \ref{theta-invariance:clm} and the fact that $\psi_q$ is odd, we have
for any $b \in \pmone^t$
\[
\theta^{(t)}(b_1 \y_1, \ldots, b_t \y_t) \cdot \inparen{\prod_{i=1}^t \psi_q(b_i \y_i)}
~=~
\theta^{(t)}(\y_1, \ldots, \y_t) \cdot \inparen{\prod_{i=1}^t \psi_q(\y_i)} \mper
\]
Recall that for each $q \in \N$ the functions $\psi_q$ are constant on the cells of the partition
$\cal{P}_q$ which divides $\bbox$ in $2^{(q+1)d}$ equal-sized boxes. By the above expression for
$\pay(\Lambda,\psi_q)$ and Lemma \ref{eval-zero:lm}, we have that for any such function $\psi_q$
\[
\sum_{t=1}^k 2^t \cdot  \int_{H^t} \theta^{(t)} (\y_1,\ldots,\y_t) \cdot
\inparen{\prod_{i=1}^t \psi_q(\y_i)}  \cdot \indic_q\left( \{ \y_i | i \in [t]\}\right) \  d\y_1\ldots d\y_t ~\leq~ 0 \mper
\]
%%%%%%%%%%%%%%%%%%%%%%%%%%%%
\newcommand{\Pt}{\cal{P}^{(t)}}
%%%%%%%%%%%%%%%%%%%%%%%%%%%%

The partition $\cal{P}_q$ induces a partition $\Pt$ on $H^t$ such that
$\prod_{i=1}^t \psi_q(\y_i)$ is constant on each cell of the partition $\Pt$. We will use
$w \in \Pt$ to denote a cell of this partition.
Also, note that the cell $w$ can be written as $(w_1, \ldots, w_t)$, where each
$w_i$ denotes a cell in $\cal{P}_q$.

We define the function $\bar{\theta}^{(t)}$, which is $\theta^{(t)}$
averaged over each cell of $\Pt$ (which has volume $2^{-qdt}$)
\[
\bar{\theta}^{(t)}(w) ~\defeq~ (2^{qd})^t
\cdot \int_{\y' \in w} \theta^{(t)}(\y_1', \ldots, \y_t') ~d\y_1' \ldots d\y_t' \mper
\]
Also, since $\prod_{i=1}^t \psi_q(\y_i)$ is constant on each $w$, we will use $\prod_{i=1}^t
\psi_q(w_i)$ to denote its value over the cell $w$.
Using the above, we get
\[
\pay(\Lambda,\psi_q)
~=~
\sum_{t = 1}^k 2^t \cdot \sum_{w \in \Pt} 2^{-qdt} \cdot \bar{\theta}^{(t)}(w_1,\ldots,w_t)
\cdot \inparen{\prod_{i = 1}^t \psi_q(w_i)} \cdot \indic_q\left( \{ w_i | i \in [t]\}\right)
~\leq~
0 \mcom
\]
for all functions $\psi_q$. Here the indicator function denotes the event that the cells $w_1,\ldots,w_t$ are distinct.
We thus assume henceforth that the cells $w_1,\ldots, w_t$ that occur in our expressions are always distinct.
Since each $\psi_q$ is defined by $2^{(q+1)d}/2$ values, corresponding to
the cells of $\cal{P}_q$ in $H$, the above can be viewed as a degree-$k$ multi-linear polynomial in
$2^{(q+1)d}/2$ variables. It is crucial that the polynomial is multi-linear and this is guaranteed because the cells
$w_1,\ldots,w_t$ are distinct.
Note that $\prod_{i=1}^t \psi_q(w_i)$ can arise from any permutation
of the tuple $(w_1,\ldots,w_t)$. Since $\theta^{(t)}$ is invariant under the permutation of its
inputs by Claim \ref{theta-invariance:clm}, the coefficient of $\prod_{i=1}^t \psi_q(w_i)$ is
\[
% 2^t \cdot 2^{-qdt} \cdot \sum_{\pi': [t] \to [t]} \bar{\theta}^{(t)}(\pi'(w_1,\ldots,w_t))
% ~=~
2^t \cdot 2^{-qdt} \cdot t! \cdot \bar{\theta}^{(t)}(w_1,\ldots,w_t) \mper
\]
%where $C_{t,w}$ is the number of permutations of the tuple $(w_1,\ldots,w_t)$.

%
We have that the above multi-linear polynomial over $\R$ is upper bounded by zero for all assignments to its variables from the
set $\values = \{-1,0,1 \}$. Applying Lemma \ref{multi-linear:lem},  the polynomial above must be identically zero and hence
\[
\forall t \in [k], ~~\forall (w_1,\ldots,w_t) \in \Pt  \suchthat w_1,\ldots,w_t ~\mbox{are distinct},  \qquad \bar{\theta}^{(t)}(w_1,\ldots,w_t) ~=~ 0 \mper
\]
Each cell of the partition $\Pt$ is a box in $\R^{dt}$ with each side having length $2^{-q}$. Since
$\bar{\theta}^{(t)}(w_1,\ldots,w_t)$ is the average of $\theta^{(t)}$ over the box corresponding to
$(w_1,\ldots,w_t)$ and $\theta^{(t)}$ is $O_{k,d,\delta}(1)$-Lipschitz by Lemma \ref{smth:lm}, we
have that for some constant $C_{k,d,\delta}$
\[
\forall t \in [k], ~~\forall (\y_1,\ldots,\y_t) \in H^t \qquad \abs{\theta^{(t)}(\y_1,\ldots,\y_t)}
~\le~ \frac{C_{k,d,\delta}}{2^q} \mper
\]
Note that the above holds also for $\y_1,\ldots,\y_t$ that are not necessarily in distinct cells, since by the Lipschitz condition, it suffices
that each $\y_i$ is close to a cell $w_i'$ such that the cells $w_1', \ldots, w_t'$ are distinct.

Finally, since the above holds for all $q \in \N$, we must have that
$\theta^{(t)}(\y_1,\ldots,\y_t) = 0$ for all $(\y_1,\ldots,\y_t) \in H^t$ and hence for all
$(\y_1,\ldots,\y_t) \in (\bbox)^t$.
\end{proof}

%%%%%%%%%%%%%%%%%%%%%%%%%%%%%%%%%%%%%%%%%%%%%%%%%%%%%%%%%%%%%%%%%%%%%%%%%%%%
\newcommand{\lt}{\Lambda^{(t)}}
\newcommand{\ltt}{\tilde{\Lambda}^{(t)}}
\newcommand{\Ext}{\ExpOp_{|S|=t} ~\ExpOp_{\pi:[t]\to[t]} ~\ExpOp_{b \in \pmone^t}}
\newcommand{\bfone}{\mathbf{1}}
\newcommand{\bfzero}{\mathbf{0}}
%%%%%%%%%%%%%%%%%%%%%%%%%%%%%%%%%%%%%%%%%%%%%%%%%%%%%%%%%%%%%%%%%%%%%%%%%%%%
For a set $S$ with $|S| = t$, permutation $\pi: [t] \to [t]$ and $b \in \pmone^t$, let
$\Lambda_{S,\pi,b}$ denote the projection of $\Lambda$ to $(t+1) \times (t+1)$ matrices as defined
in Section \ref{sec:prelims}. We define the following signed measure on space of $(t+1) \times
(t+1)$ matrices
\[
\lt ~\defeq~  \ExpOp_{|S| = t} ~\ExpOp_{\pi: [t] \to [t]} ~\Ex{b \in \pmone^t}{\hf(S) \cdot
  \inparen{\prod_{i=1}^t b_i} \cdot \Lambda_{S,\pi,b}}
\]
Lemma \ref{theta-zero:lm} immediately gives the following. Note that the integration below is over
$\zeta' \sim \lt$ and the tuple $\y_1,\ldots,\y_t$ is fixed.
\begin{claim} \label{integral-signed:clm}
For all $t \in [k]$ and for all $\y_1,\ldots,\y_t \in \bbox$, we have
\[
\int \gauss{t,d}{(\y_1,\ldots,\y_t)}{\zeta'} d\lt(\zeta') ~=~ 0 \mper
\]
\end{claim}
\begin{proof}
We start by expanding the expression for $\theta^{(t)}$.
\begin{align*}
\theta^{(t)}(\y_1,\ldots,\y_t)
&~=~
\binom{k}{t} \cdot \Ext \int \hf(S) \cdot
\inparen{\prod_{i=1}^t b_i} \cdot \gauss{t,d}{(\y_1,\ldots,\y_t)}{\zeta_{S,\pi,b}} d\Lambda(\zeta) \\
&~=~
\binom{k}{t} \cdot \Ext \int \hf(S) \cdot
\inparen{\prod_{i=1}^t b_i} \cdot \gauss{t,d}{(\y_1,\ldots,\y_t)}{\zeta'} d\Lambda_{S,\pi,b}(\zeta')
\\
&~=~
\binom{k}{t} \cdot \int \gauss{t,d}{(\y_1,\ldots,\y_t)}{\zeta'} d\lt(\zeta') \mper
\end{align*}
The claim follows by using that $\theta^{(t)}(\y_1,\ldots,\y_t) = 0$ for all $\y_1,\ldots,\y_t
\in \bbox$ by Lemma \ref{theta-zero:lm}.
\end{proof}
From the claim we get that the integral of $\gauss{t,d}{(\y_1,\ldots,\y_t)}{\zeta'}$ with respect to
the signed measure $\lt$ is zero for all $\y_1,\ldots,\y_t \in \bbox$. We will use it to show that
the integral of all continuous functions must be zero with respect to $\lt$ and hence $\lt$ must
itself be identically zero. However, we will need to modify $\lt$ a little to prove this.

We begin by considering the expression for $\gauss{t,d}{(\y_1, \ldots, \y_t)}{\zeta'}$.
We note that there is a bijection between the matrices $\zeta'$ and the pairs
$(\Sigma,\mu)$, where $\mu \in \R^t$ is a vector of means with $\mu_i = \zeta'(0,i)$ and $\Sigma$ is
the $t \times t$ covariance matrix with $\Sigma_{ij} = \zeta'(i,j) - \mu_i \cdot \mu_j$.
Also, since $\zeta' = \zeta_{S,\pi,b}$ for some $\zeta \in \C_{\delta}(f)$, we have that $\Sigma$ is
an invertible matrix with each eigenvalue at least $\delta$. We shall use $M$ to denote the matrix
$\Sigma^{-1}$ which has all eigenvalues at most $1/\delta$.
Also, as before, for vectors $\y_1,\ldots,\y_t \in \R^d$, and for $l \in [d]$, we use $\y^{(l)} \in
\R^t$ to denote the vector consisting of the $l^{th}$ coordinates of $\y_1,\ldots,\y_t$.
We can then write
\begin{align*}
&\gauss{t,d}{(\y_1,\ldots,\y_t)}{\zeta'} \\
&~=~
\frac{1}{(2\pi)^{td/2} \cdot \abs{\Sigma}^{d/2}} \cdot \exp \inparen{-\frac12 \cdot \sum_{l=1}^d
  (\y^{(l)}-\mu)^T M (\y^{(l)} - \mu)} \\
&~=~
\frac{1}{(2\pi)^{td/2} \cdot \abs{\Sigma}^{d/2}} \cdot \exp \inparen{-\frac12 \sum_{i,j=1}^t M_{ij}
  \ip{\y_i,\y_j} - \frac{d}{2} \sum_{i,j=1}^t M_{ij} \mu_i \mu_j  + \sum_{i,j=1}^t M_{ij} \mu_j
  \ip{\y_i,\bfone}} \\
&~=~
\gauss{t,d}{(\bfzero,\ldots,\bfzero)}{\zeta'} \cdot \exp \inparen{-\frac12 \sum_{i,j=1}^t M_{ij}
  \ip{\y_i,\y_j} + \sum_{i,j=1}^t M_{ij} \mu_j
  \ip{\y_i,\bfone}} \mcom
\end{align*}
where $\bfone \in \R^d$ denotes the vector $(1,\ldots,1)$ and $\bfzero \in \R^d$ denotes the vector
$(0,\ldots,0)$.

We will try to argue that for $d \geq k+1$, the values
$\inbraces{\ip{\y_i,\y_j}}_{i,j \in [t]}$ and $\inbraces{\ip{\y_i,\bfone}}_{i \in [t]}$ are
``independent enough'' so that if the integral of $\gauss{t,d}{(\y_1,\ldots,\y_t)}{\zeta'}$ with
respect to $\lt$ vanishes for all $\y_1,\ldots,\y_t \in \bbox$, then $\lt$ vanishes.
However, the values
$\inbraces{\ip{\y_i,\y_j}}_{i,j \in [t]}$ and $\inbraces{\ip{\y_i,\bfone}}_{i \in [t]}$ cannot vary
completely independently, since they are required to form a positive semidefinite matrix.
To handle
this, we define the variables (for $\beta > 0$ to be chosen later)
\begin{equation}\label{XandZ:eqn}
X_{ij} = \left\{
\begin{array}{ll}
\ip{\y_i,\y_j} & ~\text{if}~ i \neq j \\
\ip{\y_i,\y_i} - \beta & ~\text{if}~ i = j
\end{array}
\right.
\quad
\text{and}
\quad
Z_i = \ip{\y_i,\bfone} \mper
\end{equation}
Let $N$ denote the vector $\Sigma^{-1} \mu = M \mu$. We can then write
\[
\gauss{t,d}{(\y_1,\ldots,\y_t)}{\zeta'}
~=~
\gauss{t,d}{(\bfzero,\ldots,\bfzero)}{\zeta'} \cdot \exp \inparen{-\frac{\beta}{2} \cdot Tr(M)}
\cdot \exp \inparen{-\frac12 ( M \bullet X) + \ip{N,Z}} \mcom
\]
where $M \bullet X$ denotes the Frobenius inner product of the two matrices.
%%%%%%%%%%%%%%%%%%%%%%%%%%%%%%%%%%%%%%%
\newcommand{\gt}{g^{(t)}_{\beta}}
%%%%%%%%%%%%%%%%%%%%%%%%%%%%%%%%%%%%%%%

Note that there is a bijection between the pairs $(M,N)$ and the pairs $(\Sigma,\mu)$, and hence
also between the pairs $(M,N)$ and the matrices $\zeta'$. We can then view the expression
$\gauss{t,d}{(\bfzero,\ldots,\bfzero)}{\zeta'} \cdot \exp \inparen{-\frac{\beta}{2} \cdot Tr(M)}$ as
a function of the pair $(M,N)$, say $\gt(M,N)$. Also viewing the Gaussian density as a
function of the pair $(M,N)$, we can write
\[
\gauss{t,d}{(\y_1,\ldots,\y_t)}{(M,N)}
~=~
\gt(M,N) \cdot \exp \inparen{-\frac12 ( M \bullet X) + \ip{N,Z}} \mper
\]
Finally, note that the bijection from the pairs $(\Sigma,\mu)$ to the pairs $(M,N)$ is a
\emph{continuous map}, since both the maps $\Sigma \mapsto \Sigma^{-1}$ and $\mu \mapsto \Sigma^{-1} \mu$ are
continuous on the space of matrices $\Sigma$ with each eigenvalue at least $\delta$. Also, the
bijection from matrices $\zeta'$ to the pairs $(\Sigma,\mu)$ is continuous. Thus, the bijection from
matrices $\zeta'$ to the pairs $(M,N)$ is continuous and hence maps measurable sets to measurable
sets. Hence, we can also view the signed measure $\lt$ as a signed measure on the
pairs $(M,N)$.

We say that a pair $(X,Z)$ for $X \in \R^{t \times t}$ and $Z \in \R^{t}$
is {\deffont $(\beta,d)$-realizable} if there exist $\y_1,\ldots,\y_t \in
\bbox$ such that the values $X_{ij}$ and $Z_i$ satisfy the relation in Equation \ref{XandZ:eqn}.
From the above discussion and
Claim \ref{integral-signed:clm}, we have that for all $(\beta,d)$-realizable pairs $(X,Z)$
\[
\int \gt(M,N) \cdot \exp \inparen{-\frac12 ( M \bullet X) + \ip{N,Z}} d\lt(M,N) ~=~ 0 \mper
\]
Note that $\gt(M,N)$ is a positive valued function of the pair $(M,N)$. Using this we define the
signed measure $\ltt$ as
\[
\ltt ~\defeq~ \lt \cdot \gt \mper
\]
Formally, for every set $A$ (of pairs $(M,N)$) in the underlying $\sigma$-algebra, we define
\[
\ltt(A) ~\defeq~ \int \indicator{A}(M,N) \cdot \gt(M,N) ~d\lt(M,N)
\]
This operation indeed defines a new signed measure if $\gt$ is a continuous non-negative function
(see Exercise 7 in Chapter 3 of \cite{met} for example). The required conditions on $\gt$ are
easily proved.
\begin{claim}\label{g-properties:clm}
The function $\gt$ is a positive and continuous function of the pairs $(M,N)$, and is bounded above
by a constant $C_{k,d,\delta}$.
\end{claim}
\begin{proof}
Let $\zeta'(M,N)$ denote the moment matrix corresponding to $(M,N)$. Recall that the function $\gt$
was defined as
\[
\gt(M,N)
~=~
\gauss{t,d}{(\bfzero,\ldots,\bfzero)}{\zeta'(M,N)} \cdot \exp\inbraces{-\frac{\beta}{2} \cdot Tr(M)} \mper
\]
Note that the Gaussian density $\gauss{t,d}{(\bfzero,\ldots,\bfzero)}{\zeta'(M,N)}$
is a continuous function of the matrix $\zeta'$ and hence also of the
pair $(M,N)$. Also, it is positive and bounded above by
$\frac{1}{(2\pi)^{td/2}} \cdot \abs{M}^{d/2} \leq (1/\delta)^{td/2}$.
Also, $\exp\inbraces{-\frac{\beta}{2} \cdot Tr(M)}$ is a positive and continuous function of $M$ and
is bounded above by 1. Hence, their product $\gt$ is also positive, continuous and bounded as claimed.
\end{proof}
From the definition of $\ltt$, we have
that for all $(\beta,d)$-realizable pairs $(X,Z)$
\[
\int \exp\inparen{-\frac12 ( M \bullet X) + \ip{N,Z}} d\ltt(M,N) ~=~ 0 \mper
\]
The following claim shows that the class of $(\beta,d)$-realizable pairs is sufficiently rich.
\begin{claim}\label{realizable:clm}
Let $X \in \R^{t \times t}$ be a matrix and $Z \in \R^t$ be a vector such that
\[
\forall i,j \in [t]~~\abs{X_{ij}} \le \frac{\beta}{t+1}
\qquad \text{and} \qquad
\forall i \in [t]~~ \abs{Z_i} \le \frac{\beta}{t+1} \mper
\]
Then the pair $(X,Z)$ is $(\beta,d)$-realizable for $d \geq k+1$ and $\beta \leq 1/2$.
\end{claim}
\begin{proof}
Consider the $(t+1) \times (t+1)$ matrix $Y$ defined as $Y_{00} = d$, $Y_{0i} = Y_{i0} = Z_i$ and
$Y_{ii} = X_{ii} + \beta$ for $i \geq 1$,  and $Y_{ij} = X_{ij}$ for $i \neq j$ when $i,j \ge
1$. The matrix is diagonally dominant and is hence positive semidefinite, when $X$ and $Z$ are as
above.

Thus, there exist vectors $\y_0,\y_1,\ldots,\y_t \in \R^d$ when $d \ge t+1$, such that
$Y_{ij} = \ip{\y_i, \y_j}$. Also, we have $\norm{\y_0}^2 = Y_{00} = d$ and hence we can assume
(by applying a rotation if necessary) that $\y_0 = \bfone$.
Finally, we also have
\[
\norm{\y_i}^2 ~=~ Y_{ii} ~=~ X_{ii} + \beta ~\le~ \frac{\beta}{t+1} + \beta ~\le~ 1
\]
when $\beta \leq 1/2$. Thus, all the vectors $\y_i$ have $\norm{\y_i} \le 1$ and lie in
$\bbox$. By definition of the matrix $Y$, the pair $(X,Z)$ satisfies the relation in Equation
\ref{XandZ:eqn} and is hence $(\beta,d)$-realizable.
\end{proof}
Hence, all the variables $X_{ij}$ and $Z_i$ are allowed to vary in a radius of $\beta/(t+1)$ and the
above integral is zero for all values of these variables. We can expand the integral as a power
series in these variables, and then argue that all its coefficients must be zero within the radius
of convergence. However, it will be convenient to re-write the function
$\exp\inbraces{-\frac12(M \bullet X) + \ip{N,Z}}$ slightly differently, before expanding it as a
power series.

Note that since the matrices $M$ and $X$ are symmetric, the variable $X_{ij}$ actually appears twice
in $X$ when $i \neq j$, and thus it's coefficient in $-\frac12 (M \bullet X)$ is $-M_{ij}$ when $i
\neq j$ and $-M_{ii}/2$ when $i = j$. We re-write the $K = \binom{t}{2} + 2t$ variables corresponding to
$(X,Z)$ as the vector $\bfw = (w_1,\ldots,w_K)$ and their coefficients as $\bfa = (a_1,\ldots,a_K)$.
As before, the map from $(M,N)$ is a continuous bijection and thus, we can view $\ltt$ as a measure on
the coefficient vectors $\bfa$. From Claim \ref{realizable:clm}, we have
\begin{equation}\label{integral-zero:eqn}
\int \exp(\ip{\bfa,\bfw}) ~d\ltt(\bfa) = 0 \qquad \forall \bfw \in \insquare{-\frac{\beta}{t+1},
  \frac{\beta}{t+1}}^K \mper
\end{equation}
The following bound on the coefficients will be useful.
%%%%%%%%%%%%%%%%%%%%%%%%%%%%%%%%%%%%%
\newcommand{\bfe}{{\bf e}}
\newcommand{\bfr}{{\bf r}}
%%%%%%%%%%%%%%%%%%%%%%%%%%%%%%%%%%%%%
%
\begin{claim}\label{coefficient-bound:clm}
Let $\bfa = (a_1,\ldots,a_K)$ be as above. Then $\abs{a_i} \leq \frac{t}{\delta}$ for each $i \in [K]$.
\end{claim}
\begin{proof}
The coefficients for the variables $X_{ij}$ are $-M_{ij} = -\Sigma^{-1}_{ij}$ when $i \neq j$. Let
$\bfe_i$ denote the $i^{th}$ unit vector in the standard basis for $\R^t$. Then
\[
\abs{\Sigma^{-1}_{ij}} ~=~ \ip{\bfe_i, \Sigma^{-1} \bfe_j} ~\le~ \norm{\Sigma^{-1} \bfe_j} ~\le~
\frac{1}{\delta} \mper
\]
Similarly, the coefficient for $X_{ii}$, which equals $-\Sigma^{-1}_{ii}/2$ is bounded in absolute
value by $\frac{1}{2\delta}$. Finally, the coefficient for $Z_i$ equals $N_i = \inparen{\Sigma^{-1}
  \mu}_i$ and is bounded as
\[
\abs{\inparen{\Sigma^{-1} \mu}_i} ~\le~ \norm{\Sigma^{-1} \mu} ~\le~ \frac{1}{\delta} \cdot
\norm{\mu} ~\le~ \frac{t}{\delta} \mcom
\]
where the bound on $\norm{\mu}$ uses that its each coordinate $\mu_i$ is in $[-1,1]$.
\end{proof}
We shall expand the function $\exp(\ip{\bfa,\bfw})$ as a power series and integrate each term
separately to obtain a formal series $\cal{S}(\bfw)$. To write the series, it will be convenient to
use the multi-index notation.
Let $\bfr = (r_1,\ldots,r_K) \in (\Z_+)^K$ denote a multi-index. Let $\bfa^{\bfr}$ denote the term
$\prod_{i=1}^K a_i^{r_i}$ and define $\bfw^{\bfr}$ similarly. Let $\abs{\bfr}$ denote $\sum_{i=1}^K
r_i$ and let $(\bfr)!$ denote $\prod_{i=1}^K (r_i !)$.
Then we can write
\[
\exp(\ip{\bfa,\bfw})
~=~
\sum_{r=0}^{\infty}\frac{(\ip{\bfa,\bfw})^r}{r!}
~=~
\sum_{\bfr \in \Z_+^K} \frac{\bfw^{\bfr} \cdot \bfa^{\bfr}}{(\bfr)!} \mper
\]
We define the series
\[
\cal{S}(\bfw) ~=~ \sum_{\bfr \in \Z_+^K} \frac{\bfw^{\bfr}}{(\bfr)!} \cdot \int \bfa^{\bfr} ~d\ltt(\bfa) \mper
\]
We next show that this formal series converges everywhere. Using the convergence,
we can equate it to the integral in Equation \ref{integral-zero:eqn}. The fact that the integral is
zero in a box around the origin will then yield the desired conclusion.
\begin{claim}
Let the vectors $\bfa = (a_1,\ldots,a_K)$ and the measure $\ltt$ be as above. Then the series
\[
\cal{S}(\bfw) ~=~ \sum_{\bfr \in \Z_+^K} \frac{\bfw^{\bfr}}{(\bfr)!} \cdot \int \bfa^{\bfr} ~d\ltt(\bfa)
\]
is absolutely convergent for all $\bfw \in \R^K$.
\end{claim}
\begin{proof}
We bound the absolute value of the integral $\int \bfa^{\bfr} ~d\ltt(\bfa)$ for each $\bfr \in
\Z_+^K$. By claim \ref{coefficient-bound:clm}, $\abs{a_i} \le \frac{t}{\delta}$ for each coordinate
$a_i$ of $\bfa$. We then have
\[
\abs{\int \bfa^{\bfr} ~d\ltt(\bfa) }
~=~
\abs{\int \bfa^{\bfr} \cdot \gt(\bfa) ~d\lt(\bfa)}
~\leq~
\sup_{\bfa} \inparen{\abs{\bfa^{\bfr}} \cdot \gt(\bfa)} \cdot\int d|\lt|(\bfa) \mper
\]
Note that here we have used notation $|\lt|$, which is used to refer to a \emph{positive} measure
corresponding to $\lt$, which is given by the Hahn decomposition theorem for signed measures. By the
decomposition theorem, any signed measure $\nu$ can be written as $\nu_+ - \nu_{-}$, where $\nu_+$
and $\nu_{-}$ are positive measures supported on disjoint measurable sets, say $P$ and $N$
respectively.  Then $|\nu|$ is used to refer to the measure $(\nu_+ + \nu_{-})$. The inequality
above follows immediately by considering this decomposition.

Also, if $\Lambda_0$ is a finite  linear combination of positive measures \ie $\Lambda_0 = \sum_{i}
c_i \Lambda_i$, then using the above decomposition, we can say that $|\Lambda_0| \leq \sum_i
\abs{c_i} \Lambda_i$. By the definition of $\lt$ as a linear combination of positive measures,
we can now bound the integral as
\[
\int d|\lt|(\bfa) ~\leq~ \Ext \insquare{\abs{\hf(S)} \cdot \abs{\prod_{i=1}^t b_i} \cdot \int
  d\Lambda_{S,\pi,b}(\bfa)} ~\le~ 1 \mcom
\]
since each $\Lambda_{S,\pi,b}$ is a probability measure.
Using the bound on the coefficients $a_i$, we have that $\abs{\bfa^{\bfr}} \leq
(t/\delta)^{\abs{\bfr}}$. Also, by Claim \ref{g-properties:clm}, we have that $\gt \leq
C_{k,d,\delta}$. Thus, we get
\[
\abs{\int \bfa^{\bfr} ~d\ltt(\bfa) }
~\le~
C_{k,d,\delta} \cdot \inparen{\frac{t}{\delta}}^{\bfr} \mper
\]

For $\bfw = (w_1,\ldots, w_K) \in \R^K$, let $\bfw_+ = (\abs{w_1},\ldots,\abs{w_K})$ be the vector
of absolute values of all the entries of $\bfw$.
To show that $\cal{S}(\bfw)$ is absolutely convergent, we need to
show that the series $\cal{S}'(\bfw)$, obtained by replacing each term of $\cal{S}(\bfw)$ by its
absolute value, is
convergent.
We can write
\begin{align*}
\cal{S}'(\bfw)
% ~=~
% \sum_{\bfr \in \Z_+^K} \frac{\bfw^{\bfr}}{(\bfr)!} \cdot \int \bfa^{\bfr} ~d\ltt(\bfa)
~=~
\sum_{\bfr \in \Z_+^K} \abs{\frac{\bfw^{\bfr}}{(\bfr)!}} \cdot \abs{\int \bfa^{\bfr} ~d\ltt(\bfa) }
&~=~
\sum_{\bfr \in \Z_+^K} \frac{\bfw_+^{\bfr}}{(\bfr)!} \cdot \abs{\int \bfa^{\bfr} ~d\ltt(\bfa) } \\
&~\le~
\sum_{\bfr \in \Z_+^K} \frac{\bfw^{\bfr}_+}{(\bfr)!} \cdot C_{k,d,\delta} \cdot
\inparen{\frac{t}{\delta}}^{\abs{\bfr}} \\
&~=~
C_{k,d,\delta} \cdot \exp\inparen{\frac{t}{\delta} \cdot \sum_{i=1}^K \abs{w_i}} \mper
\end{align*}
The last equality above used the fact that for all $\bfx \in \R^K$, the series
$\sum_{\bfr \in \Z_+^K} \frac{\bfx^{\bfr}}{(\bfr)!}$ converges to $\exp\inparen{\sum_{i=1}^K x_i}$.
\end{proof}
Thus, we know that for $\bfw \in \insquare{-\frac{\beta}{t+1}, \frac{\beta}{t+1}}^K$, the
series $\cal{S}(\bfw)$ always converges to zero.
We shall use this to show that all the
  coefficients of the series must be zero, which in turn implies that the signed measure $\ltt$ must
  be identically zero.
The following lemma finishes the proof.
%%%%%%%%%%%%%%%%%%%%%%%%%%%%%%%%%%%%%%%%%%%%%%%
\newcommand{\tl}{\tilde{\Lambda}}
%%%%%%%%%%%%%%%%%%%%%%%%%%%%%%%%%%%%%%%%%%%%%%%
\begin{lemma}\label{measure-zero:lm}
Let $\tl$ be a signed measure on vectors $\bfa = (a_1,\ldots,a_K)$ contained in a compact set
$X \sub \R^K$, such that the series
\[
\cal{S}(\bfw) ~=~ \sum_{\bfr \in \Z_+^K} \frac{\bfw^{\bfr}}{(\bfr)!} \cdot \int \bfa^{\bfr} ~d\tl(\bfa)
\]
in the variables $w_1,\ldots,w_K$ converges and is identically zero for $\abs{w_i} \leq \tau$. Then
$\tl = 0$.
\end{lemma}
\begin{proof}
Since the series converges for all $\bfw \in [-\tau,\tau]^K$, $\cal{S}(\bfw)$ defines a real
analytic function for all $\bfw \in [-\tau,\tau]^K$. Since the function is identically zero in
$[-\tau,\tau]^K$, all its derivatives at $\bfw = (0,\ldots,0)$ must be zero. By comparing
coefficients of the above series with the Taylor expansion, we get that
\[
\int \bfa^{\bfr} ~d\tl(\bfa) = 0 \qquad \forall \bfr \in \Z_+^K \mper
\]
Thus, for all polynomials $P$ in the variables $(a_1,\ldots,a_K)$, we have that
$\int P(\bfa) ~d\tl(\bfa) = 0$. By the the Stone-Weierstrass theorem, we know that for any
continuous function $h: X \to \R$, there is a sequence of polynomials $\inbraces{P_r}_{r \in \N}$,
which converges to $h$. By the dominated convergence theorem for integrals over signed measures, we
have that
\[
\int h(\bfa) ~d\tl(\bfa) ~=~ \lim_{r \to \infty} \int P_r(\bfa) ~d\tl(\bfa) ~=~ 0 \mper
\]
Finally, we use the (uniqueness part of) Riesz Representation Theorem (see Chapter 13 in
\cite{met}), which says for a compact metric space $X$
and  two signed measures
$\Lambda_1$ and $\Lambda_2$ defined on $X$, if
$\int h(\bfa) ~d\Lambda_1(\bfa) = \int h(\bfa) ~d\Lambda_2(\bfa)$
for all continuous functions $h:X \to \R$, then $\Lambda_1 = \Lambda_2$. Using this theorem, we
conclude that $\tl = 0$.
\end{proof}
The above lemma gives that the signed measure $\ltt$ must be identically zero. Note that to
apply the lemma, we use the fact that the space of the vectors $\bfa$ is compact. This follows from
the fact that the space of the matrices $\zeta'$ is compact and a continuous map preserves
compactness.

However, we are
interested in the measure $\lt$, and we have $\ltt = \gt \cdot \lt$. The following lemma shows that
then we must in fact have that $\lt = 0$.
\begin{lemma}
Let $\Lambda_1$ and $\Lambda_2$ be two signed measures on a compact metric space $X$
such that $\Lambda_2 = g \cdot \Lambda_1$
for a strictly positive and bounded continuous function $g$. Then if $\Lambda_2$ is identically
zero, so is $\Lambda_1$.
\end{lemma}
\begin{proof}
We consider the integral of any continuous function $h : X \to \R$ with respect to
$\Lambda_1$. Note that since $g$ is strictly positive and $X$ is compact, $g$ is also bounded below
by some absolute constant.
Using the fact that $g$ is positive and bounded, we can write
\[
\int h ~d\Lambda_1 ~=~ \int \frac{h}{g} \cdot g ~d\Lambda_1 ~=~ \int \frac{h}{g} ~d\Lambda_2 \mper
\]
Since $h$ and $g$ are both continuous and $g$ is positive, the function $\frac{h}{g}$ is continuous
and hence measurable. Thus, we obtain that for every continuous function $h$,
\[
\int h ~d\Lambda_1 ~=~ \int \frac{h}{g} ~d\Lambda_2 ~=~ 0\mper
\]
Again, by the (uniqueness aspect of) Riesz Representation Theorem as in the proof of Lemma
\ref{measure-zero:lm}, this implies that $\Lambda_1 = 0$.
\end{proof}
Since the function $\gt$ is strictly positive, bounded and continuous by Claim \ref{g-properties:clm},
the previous claim implies that the measure $\Lambda$ on $\C_{\delta}(f)$ is
such that for each $t \in [k]$, the signed measure
\[
\lt ~=~ \Ext \insquare{\hf(S) \cdot \inparen{\prod_{i=1}^t b_i} \cdot \Lambda_{S,\pi,b}}
\]
is identically zero \ie $\Lambda$ is a vanishing measure.
However, we need to establish the existence of a vanishing measure on $\C(f)$. The
following claim shows that the existence of such measures on $\C(f)$ and $\C_{\delta}(f)$ are equivalent.
\begin{claim}\label{vanishing-measure:clm}
There exists a vanishing probability measure $\Lambda$ on $\C_{\delta}(f)$ if and only if there
exists a vanishing probability measure $\Lambda'$ on $\C(f)$.
% There exists a probability measure $\Lambda$ on $\C_{\delta}(f)$ such that for each $t \in [k]$,
% the above signed measure vanishes, if and only there exists a probability measure $\Lambda'$ on
% $\C_{\delta}(f)$ such that for each $t \in [k]$, the signed measure
% \[
% \Lambda'^{(t)} ~=~ \Ext \insquare{\hf(S) \cdot \inparen{\prod_{i=1}^t b_i} \cdot \Lambda_{S,\pi,b}'}
% \]
% is identically zero.
\end{claim}
\begin{proof}
By definition of the body $\C_{\delta}(f)$, we have that for every $\zeta \in \C_{\delta}(f)$, the
matrix
\[
\zeta' ~=~ \frac{\zeta - \delta \cdot \mathbb{I}_{k+1}}{1-\delta}
\]
is in $\C(f)$, where $\mathbb{I}_{k+1}$ denotes the $(k+1) \times (k+1)$ identity matrix. The above
map defines a continuous bijection from $\C_{\delta}(f)$ to $\C(f)$, and thus maps measurable sets to
measurable sets. Thus, we can define a measure $\Lambda'$ on $\C(f)$ where for any measurable set
$A' \sub \C(f)$, we take $\Lambda'(A') = \Lambda(A)$ for $A$ which is the inverse image of $A'$ under
the above map. Note that if $\zeta' \in \C(f)$ is the image of $\zeta \in \C_{\delta}(f)$ under the
above map, then we also have for any $S, \pi$ and $b$ that
\[
\zeta'_{S,\pi,b} ~=~ \frac{\zeta_{S,\pi,b} - \delta \cdot \mathbb{I}_{|S|}}{1-\delta} \mper
\]
Thus, we also have that for every measurable set $A$ of $(|S|+1)\times(|S|+1)$ matrices, and its
image $A'$ that $\Lambda_{S,\pi,b}(A) = \Lambda'_{S,\pi,b}(A')$. Since $\lt$ is a linear combination of
the measures $\Lambda_{S,\pi,b}$ for $|S|=t$, and $\Lambda'^{(t)}$ is an identical linear
combination of measures $\Lambda'_{S,\pi,b}$, $\lt$ being identically zero implies that
$\Lambda'^{(t)}$ must also be identically zero.

For the reverse direction, we consider the inverse map
$\zeta = (1-\delta) \cdot \zeta' + \delta \cdot \mathbb{I}_{k}$, which is also continuous. The
rest of the argument is the same as above.
\end{proof}

\subsection{The Integrality Gap Instances}
We now show that for a predicate $f$, if there exists a vanishing probability measure $\Lambda$ on
$\C(f)$, then there exists an infinite family of
$\maxkcsp(f)$ instances such that the SDP has optimum value $1-o(1)$,
while the value of any integer assignment lies in $[\rho(f)-o(1), \rho(f) + o(1)]$.

First, we give a description of our instance family in the
continuous setting and then sketch how to discretize it.
The advantage of this description is that the soundness and
completeness of the instance are far easier to analyze than
in the discrete setting, while continuity properties ensure that
the results translate to the discrete setting as well. To ensure the continuity of various functions
defined on the matrices $\zeta$, we will instead work with  a vanishing measure $\Lambda$ defined on
$\C_{\delta}(f)$ (for some small $\delta > 0$) instead of $\C(f)$. By Claim
\ref{vanishing-measure:clm}, the existence of vanishing measures on $\C(f)$ and $\C_{\delta}(f)$ are
equivalent.

Our set of literals will be the set of all points in $\R^d$,
where the variable represented by the point $-\y$ is treated as
the negation of the variable represented by the point $\y \in \R^d$.
The set of constraints will be given by all $k$-tuples of
points in $\R^d$. We think of the constraints being generated as follows: we pick a $\zeta$
according to $\Lambda$, choose a $k$-tuple of points $(\y_1,\ldots, \y_k)$ according to
$\cal{N}_{d}(\zeta)$,  and impose the constraint $f(\y_1,\ldots,\y_k)$.
Thus, given a $k$-tuple $(\y_1,\ldots,\y_k)$, the ``weight'' of the constraint $f(\y_1,\ldots,\y_k)$
is $\Ex{\zeta \sim \Lambda}{\gauss{k,d}{(\y_1,\ldots,\y_k)}{\zeta}}$.

We remark that while it is convenient to think of the instance as above, the set of variables in
fact only corresponds to $\y \in H$, where $H$ is an arbitrary half-space of $\R^d$, say $H = \R_+
\times \R^{d-1}$. This is because the variable $-\y$ is supposed to be the negation of the variable
$\y$. This means that any ``assignment'' to the variables, must be an odd function on $\R^d$. Also,
we will need to be careful of the above while constructing the SDP solution.

\subsubsection*{Soundness}
%%%%%%%%%%%%%%%%%%%%%%%%%%%%%%%%%%%%%%%%%%%%%
\newcommand{\dyk}{d\y_1,\ldots,d\y_k}
\newcommand{\dyt}{d\y_1,\ldots,d\y_t}
\newcommand{\yk}{\y_1,\ldots,\y_k}
\newcommand{\yt}{\y_1,\ldots,\y_t}
%%%%%%%%%%%%%%%%%%%%%%%%%%%%%%%%%%%%%%%%%%%%%
Let $\psi:\R^d\to\{\pm 1\}$ be an odd function, which forms an assignment
to the variables of our continuous $\maxkcsp(f)$ instance.
In the continuous setting $\psi$ may not even be
measurable but such technical issues do not occur in a
discrete setting, which is our goal, and so we will ignore them for the analysis below.
We show that the fraction of constraints satisfied by any such assignment is $\rho(f)$.

\begin{lemma}\label{sdp-soundness:lm}
Let $\Phi$ be the instance as described above and let $\psi: \R^d \to \pmone$ be any measurable odd
function. Then the fraction of constraints satisfied by $\psi$, denoted by $\sat(\psi)$, is equal to
$\rho(f)$.
\end{lemma}
\begin{proof}
The objective value is given by:
\begin{align} \label{soundness:eqn}
\nonumber \sat(\psi)
&~=~
\ExpOp_{\zeta \sim \Lambda} ~\Ex{\yk \sim \cal{N}_d(\zeta)}{f(  \psi(\y_1),\ldots,\psi(\y_k)    )} \\
\nonumber &~=~
\rho(f) +
\ExpOp_{\zeta \sim \Lambda} ~\Ex{\yk \sim \cal{N}_d(\zeta)}{\sum_{S \sub [k] \atop S \neq \emptyset}
  \hf(S) \cdot \prod_{i \in S} \psi(\y_i)} \\
&~=~
\rho(f) + \opay(\Lambda,\psi)
\end{align}
By Claim \ref{payoff-theta:clm}, we can write $\opay(\Lambda,\psi)$ as (the claim is stated for the function $\pay(\cdot,\cdot)$, but it is
easily checked that it holds for the function $\opay(\cdot,\cdot)$ by ignoring the indicator function therein)
\[
\opay(\Lambda,\psi) ~=~ \sum_{t=1}^k \int_{(\R^d)^t} \theta^{(t)} (\y_1,\ldots,\y_t) \cdot
\inparen{\prod_{i=1}^t \psi(\y_i)}  d\y_1\ldots d\y_t \mcom
\]
where the function $\theta^{(t)}$ is defined as
\[
\theta^{(t)}(\y_1,\ldots,\y_t) ~\defeq~
\sum_{|S| = t} \ExpOp_{\pi: [t] \to [t]} ~\ExpOp_{b \in \pmone^t} \Ex{\zeta \sim\Lambda}{\hf(S) \cdot
  \inparen{\prod_{i=1}^t b_i} \cdot \gauss{t,d}{(\y_1,\ldots, \y_t)}{\zeta_{S,\pi,b}}} \mper
\]
However, since $\Lambda$ is a vanishing measure, it is easy to see that $\theta^{(t)}$ must be
identically zero for each $t$. This follows from writing $\theta^{(t)}$ as
\begin{align*}
\theta^{(t)}(\y_1,\ldots,\y_t)
&~=~
\sum_{|S| = t} \ExpOp_{\pi: [t] \to [t]} ~\ExpOp_{b \in \pmone^t} ~\Ex{\zeta' \sim\Lambda_{S,\pi,b}}{\hf(S) \cdot
  \inparen{\prod_{i=1}^t b_i} \cdot \gauss{t,d}{(\y_1,\ldots, \y_t)}{\zeta'}} \\
&~=~
\sum_{|S| = t} \ExpOp_{\pi: [t] \to [t]} ~\ExpOp_{b \in \pmone^t} \int \hf(S) \cdot
  \inparen{\prod_{i=1}^t b_i} \cdot \gauss{t,d}{(\y_1,\ldots, \y_t)}{\zeta'} d\Lambda_{S,\pi,b}(\zeta') \\
&~=~
\int \gauss{t,d}{(\y_1,\ldots, \y_t)}{\zeta'} ~d\Lambda^{(t)}(\zeta') \\
&~=~
0 \mper
\end{align*}
Hence, $\sat(\psi) = \rho(f)$ for every measurable and odd assignment $\psi$ to our continuous instance.
\end{proof}

\subsubsection*{Completeness}
We now demonstrate an SDP solution for the continuous instance, for which the value of the objective
is 1. However, it is not a valid solution to the relaxation in Figure \ref{fig:basic-sdp}, as some of the
SDP constraints will not be satisfied for each tuple of variables involved in a constraint, but only in
expectation over these variables (which is also the case with the continuous Gaussian version of
the Feige-Schechtman instance for MAX-CUT). However, as we discretize the instance, these
constraints will be satisfied upto a small error, with high probability over the participating tuple
of variables. We will be able to correct these errors later, without significantly affecting the
value of the SDP solution.

To construct the SDP solution, we need to specify a vector $\vtwoempty$, a vector $\vtwo{\y}{b}$
for each $\y \in \R^d$ and $b \in \pmone$, and a variable $\vartwo{(\yk)}{\alpha}$ for all $\yk \in
\R^d$ and $\alpha \in \pmone^k$, satisfying the conditions in Figure \ref{fig:basic-sdp}.
We take
the vector $\vtwoempty = \frac{1}{\sqrt{d}} \cdot \bfone$. We shall also define the vector
$\voneempty = \vtwoempty$ for the calculations below.
For each $\y \in \R^d$, we first define the following vectors.
\[
\vone{\y} = \frac{1}{\sqrt{d}} \cdot \y,
\qquad
\vtwo{\y}{1} = \frac{1}{2} \cdot (\voneempty + \vone{\y})
\quad \text{and} \quad
\vtwo{\y}{-1} = \frac{1}{2} \cdot (\voneempty - \vone{\y}) \mper
\]
Note that $\vtwo{-\y}{b} = \vtwo{\y}{-b}$ for any $\y \in \R^d$ and $b \in \pmone$, since $-\y$ is
simply the negation of the variable $\y$.

Before describing the values of the variables $\vartwo{(\yk)}{\alpha}$, we mention a subtle
issue. Note that we need to produce one such set of
variables \emph{for every constraint} in the CSP instance, and not just for every $k$-tuple of
variables. This means that if for some $\yk \in H$, there are two constraints of the form $f(\yk)$
and $f(-\yk)$, then we will produce \emph{two different sets of variables},
$\inbraces{\vartwo{(\yk)}{\alpha}}_{\alpha \in \pmone^k}$ and
$\inbraces{\vartwo{(-\yk)}{\alpha}}_{\alpha \in \pmone^k}$, corresponding to the same tuple $(\yk)$
of CSP variables.
Similarly for constraints where the tuple $(\yk)$ is generated according to two different matrices
$\zeta$ and $\zeta'$ in the support of $\Lambda$.
The only consistency conditions are the ones imposed through the inner products of
the corresponding vectors.

Since every constraint is uniquely described by a tuple $(\yk) \in \R^d$ and $\zeta \in \Lambda$,
we have a different set of variables $\inbraces{\vartwo{(\yk)}{\alpha}^{(\zeta)}}_{\alpha \in
  \pmone^k}$ for each $\zeta$ and $\yk \in \R^d$.
We now describe the value of the variable $\vartwo{(\yk)}{\alpha}^{(\zeta)}$ for all $\yk \in \R^d$ and
$\alpha \in \pmone^k$ and $\zeta \in \C_{\delta}(f)$.
Since the only ``actual variables'' correspond to $\y \in H$, some of the
elements $\y_i$ might be negations of actual variables $-\y_i \in H$. In that case we interpret
$\vartwo{(\y_1,\ldots,\y_i, \dots,\y_k)}{\alpha}^{(\zeta)}$ as $\vartwo{(\y_1,\ldots,-\y_i,
  \dots,\y_k)}{\alpha'}^{(\zeta)}$
(for the constraint corresponding to $(\yk)$ and $\zeta$),
where $\alpha'$ is $\alpha$ with the $i^{th}$ bit negated.

\Mnote{The above is a somewhat subtle issue. Let me know if the description here is too confusing.}
%

%%%%%%%%%%%%%%%%%%%%%%%%%%%%%%%%%%%%%%%%%%%%%%%%%%%
\newcommand{\bnu}{\bar{\nu}}
%%%%%%%%%%%%%%%%%%%%%%%%%%%%%%%%%%%%%%%%%%%%%%%%%%%
%
Recall that for each $\zeta \in \C_{\delta}(f)$, there exists a distribution $\nu$ supported on $f^{-1}(1)$, such
that $\zeta = (1-\delta) \cdot \zeta(\nu) + \delta \cdot \mathbb{I}_{k+1}$. Consider a distribution $\bnu$,
  which is $\nu$ with probability $1-\delta$ and uniform on $\pmone^k$ with probability
  $\delta$. Then, we have
\[
\zeta ~=~ (1-\delta) \cdot \zeta(\nu) + \delta \cdot \mathbb{I}_{k+1} ~=~ \zeta(\bnu) \mper
\]
We refer to (an arbitrary choice of) this distribution $\bnu$ for a given $\zeta \in \C_{\delta}(f)$
as $\bnu_{\zeta}$. For a constraint $f(\yk)$, the variable $\vartwo{(\yk)}{\alpha}^{(\zeta)}$ is then defined
as
\[
\vartwo{(\yk)}{\alpha}^{(\zeta)} ~=~ \bnu_{\zeta} (\alpha) \mcom
\]
where $\bnu_{\zeta}(\alpha)$ is the probability assigned to $\alpha$ by $\bnu_{\zeta}$. We now show
that this assignment has SDP value $(1-\delta)$ and satisfies the SDP constraints in expectation over
the tuples $(\yk)$.

\begin{lemma}\label{sdp-completeness:lm}
Let $\Phi$ be the continuous instance of $\maxkcsp(f)$ as described above. Then the SDP solution
given by the vectors $\vtwo{\y}{b}$ and the variables $\vartwo{(\yk)}{\alpha}^{(\zeta)}$ defined as
above has an objective value of $1-\delta$. Also, we have
\begin{itemize}
\item[-]
For all $i \in [k]$ and all $\zeta$ in the support of $\Lambda$,
$\Ex{\yk \sim\cal{N}_d(\zeta)}{\ip{\vtwo{\y_i}{1},\vtwo{\y_i}{-1}}} = 0$.
\item[-]
For all $i,j \in [k]$ with $i \neq j$, all $b,b' \in \pmone$, and all $\zeta$ in the support of $\Lambda$,
\[
\Ex{\yk \sim\cal{N}_d(\zeta)}{\ip{\vtwo{\y_i}{b}, \vtwo{\y_j}{b'}}}
~=~
\Ex{\yk \sim\cal{N}_d(\zeta)}{\sum_{\alpha \in \pmone^t \atop \alpha(i) = b, \alpha(j) = b'}
  \vartwo{(\yk)}{\alpha}^{(\zeta)} } \mper
\]
\end{itemize}
The remaining SDP conditions are satisfied for each constraint corresponding to a tuple $(\yk)$ and
matrix $\zeta$.
\end{lemma}
\begin{proof}
We first verify the SDP constraints. It is immediate from the definitions that we have
$\norm{\vtwoempty} = 1$, $\vtwo{\y}{1} + \vtwo{\y}{-1} = \vtwoempty$ for all $\y \in \R^d$, and
$\vartwo{(\yk)}{\alpha}^{(\zeta)} \ge 0$ for all $\yk \in \R^d$ and all $\alpha \in \pmone^k$. The remaining
two constraints will only be satisfied in expectation over the tuple $(\yk)$.

Consider the constraint $\ip{\vtwo{\y}{1},\vtwo{\y}{-1}} = 0$. With our definition of
vectors, we have
\[
\ip{\vtwo{\y}{1},\vtwo{\y}{-1}}
~=~
\frac14 \cdot \inparen{\norm{\voneempty}^2 - \norm{\vone{\y}}^2}
~=~
\frac14 \cdot \inparen{1 - \frac{1}{d} \cdot \norm{\y}^2},
\]
which is not always zero. However, for any $i \in [k]$, we have
\begin{align*}
\Ex{\yk \sim\cal{N}_d(\zeta)}{\ip{\vtwo{\y_i}{1},\vtwo{\y_i}{-1}}}
&~=~
\Ex{\yk \sim\cal{N}_d(\zeta)}{\frac14 \cdot \inparen{1-\frac{1}{d}
    \cdot \norm{\y_i}^2}} \\
&~=~
{\frac14 \cdot \inparen{1-\zeta(i,i)}} \\
&~=~ 0 \mper
\end{align*}
Thus, the constraint is satisfied in expectation over the tuples $(\yk)$ for each $\zeta$.
Similarly, for any
tuple $(\yk)$, $i,j \in [k], i \neq j$ and $b,b' \in \pmone$, we have the constraint
\[
\sum_{\alpha \in \pmone^t \atop \alpha(i) = b, \alpha(j) = b'} \vartwo{(\yk)}{\alpha}^{(\zeta)}
~=~
\ip{\vtwo{\y_i}{b}, \vtwo{\y_j}{b'}} \mper
\]
From the definition of the variables $\vartwo{(\yk)}{\alpha}^{(\zeta)}$, the left hand side equals
\begin{align*}
{ \Prob{z \sim \bnu_{\zeta}}{(z_i = b) \wedge (z_j = b')} }
&~=~
\Ex{z \sim \bnu_{\zeta}}{\inparen{\frac{1+(-1)^{b} \cdot z_i}{2}} \cdot
      \inparen{\frac{1+(-1)^{b'} \cdot z_j}{2}} } \\
&~=~
\frac14 \cdot \inparen{ 1+(-1)^b \cdot \zeta(0,i) + (-1)^{b'} \cdot \zeta(0,j) + (-1)^{b+b'}
    \cdot \zeta(i,j) } \mper
\end{align*}
Also, the right hand side equals
\begin{align*}
\ip{\vtwo{\y_i}{b}, \vtwo{\y_j}{b'}}
&~=~
\ip{
\inparen{\frac{\voneempty + (-1)^b \cdot \vone{\y_i}}{2}},
\inparen{\frac{\voneempty + (-1)^{b'} \cdot \vone{\y_j}}{2}}
} \\
&~=~
\frac14 \cdot \inparen{1 + \frac{(-1)^b}{d} \cdot \ip{\bfone,\y_i} + \frac{(-1)^{b'}}{d} \cdot
  \ip{\bfone,\y_j} + \frac{(-1)^{b+b'}}{d} \cdot \ip{\y_i,\y_j}} \mper
\end{align*}
Again, we have in expectation over the tuples $(\yk)$,
\begin{align*}
&\Ex{\yk \sim\cal{N}_d(\zeta)}{\ip{\vtwo{\y_i}{b}, \vtwo{\y_j}{b'}}} \\
&~=~
\Ex{\yk \sim\cal{N}_d(\zeta)}{ \frac14 \cdot \inparen{1 + \frac{(-1)^b}{d} \cdot \ip{\bfone,\y_i} + \frac{(-1)^{b'}}{d} \cdot
  \ip{\bfone,\y_j} + \frac{(-1)^{b+b'}}{d} \cdot \ip{\y_i,\y_j}} } \\
&~=~
{ \frac14 \cdot \inparen{1 + (-1)^b \cdot \zeta(0,i) + (-1)^{b'} \cdot \zeta(0,j)
+ (-1)^{b+b'} \cdot \zeta(i,j)} } \mper
\end{align*}
Thus, the SDP constraint is satisfied in expectation over the tuples $(\yk)$.
Finally, we verify that the above solution has an SDP value of $1-\delta$. The expression for the
SDP value can be written as
\begin{align*}
\ExpOp_{\zeta \sim \Lambda} ~\Ex{\yk \sim\cal{N}_d(\zeta)}{ \sum_{\alpha \in \pmone^k} f(\alpha)
  \cdot \vartwo{(\yk)}{\alpha}^{(\zeta)} }
&~=~
\ExpOp_{\zeta \sim \Lambda} ~\Ex{\yk \sim\cal{N}_d(\zeta)}{ \Prob{\alpha \sim
    \bnu_{\zeta}}{f(\alpha)=1} } \\
&~\ge~
\ExpOp_{\zeta \sim \Lambda} ~\Ex{\yk \sim\cal{N}_d(\zeta)}{ (1-\delta) } \mcom
\end{align*}
since $\bnu_{\zeta}$ is a convex combination of $\nu$ with probability $1-\delta$ and uniform on
$\pmone^k$ with probability $\delta$, and $\nu$ is supported on $f^{-1}(1)$.
\end{proof}

\subsubsection*{Discretization}
We now describe how to discretize the continuous instance described above. We first discretize the
body $\C_{\delta}(f)$ and replace it by a sufficiently dense set of points. The measure $\Lambda$
can then be replaced by a distribution $\Lambda'$ over these set of points.
Recall that the value of any
integer assignment $\psi$ to the continuous instance generated according to the measure $\Lambda$ is
$\rho(f) + \opay(\Lambda,\psi)$ as derived in Equation \ref{soundness:eqn}.
Since the function $\opay(\cdot,\psi)$ is continuous in the matrices
$\zeta$ (for $\zeta \in \C_{\delta}(f)$) by  Lemma \ref{payoff-continuous:lm}, replacing $\Lambda$
by $\Lambda'$ only affects the value of the assignment $\psi$ by $o(1)$. Hence, the value of each
assignment is in $[\rho(f) - o(1), \rho(f) + o(1)]$.

Next we restrict the set of constraints.
We say that a constraint on the tuple $(\yk)$ generated according to a
matrix $\zeta$ is $\eps$-good, if for all $i, j \in [k]$, we have
\[
\abs{\frac{1}{d} \cdot \ip{\y_i, \bfone} - \zeta(0,i)} ~\le~ \eps
\qquad \text{and} \qquad
\abs{\frac{1}{d} \cdot \ip{\y_i, \y_j} - \zeta(i,j)} ~\le~ \eps \mper
\]
We will restrict our set of constraints only to the set of $\eps$-good constraints, for a
sufficiently small $\eps$ to be fixed later.
Since the tuple $(\yk)$ is generated according to $\cal{N}_d(\zeta)$, we have that
$\ex{\frac1d \cdot \ip{\y_i, \bfone} = \zeta(0,i)}$ and $\ex{\frac{1}{d} \cdot \ip{\y_i, \y_j}} =
\zeta(i,j)$.
Hence for sufficiently large $d$,
the probability that a randomly generated constraint is \emph{not} $\eps$-good is $o(1)$ by
standard tail estimates on Gaussian variables. Thus, restricting our instance only to the set of
$\eps$-good constraints changes the value of all assignments only by $o(1)$.
Note that it follows from the proof of Lemma \ref{sdp-completeness:lm} that for any $\eps$-good constraint, we
will have for all $i,j \in [k]$ and $b,b' \in \pmone$
\[
\abs{\ip{\vtwo{\y_i}{1},\vtwo{\y_i}{-1}}} \le \eps
\quad \text{and} \quad
\abs{ \ip{\vtwo{\y_i}{b}, \vtwo{\y_j}{b'}} - \sum_{\alpha \in \pmone^t \atop \alpha(i) = b, \alpha(j) = b'}
  \vartwo{(\yk)}{\alpha}^{(\zeta)} } \le \eps \mper
\]

Finally, we discretize the set of variables. Since we only consider $\eps$-good
constraints, we have that for all participating tuples $(\yk)$ and all $i \in [k]$,
$\abs{\frac{1}{d} \cdot \ip{\y_i,\y_i} - \zeta(i,i)} \le \eps$ and hence $\norm{\y_i}^2 \in
\insquare{\inparen{1-\eps}\cdot{d}, \inparen{1+\eps}\cdot{d}}$. Thus, we can restrict ourselves
to a sufficiently dense set of points such that their squared distance from the origin is between
$\inparen{1-\eps}{d}$ and $\inparen{1+\eps}{d}$.
For each constraint on a tuple $(\yk)$, we collapse each
$\y_i$ to the nearest point in our set, which gives a finite set of constraints over a finite number
of variables. Since an assignment to the collapsed instance can also be thought of as an assignment
to the continuous instance (where $\psi$ is constant over each set of collapsed points), the value
of any assignment still remains in the range $[\rho(f) - o(1), \rho(f) + o(1)]$.

We define the vectors $\vtwo{\y}{b}$ and variables $\vartwo{(\yk)}{\alpha}^{(\zeta)}$ as before for
our new set of variables.
Since the
contribution of \emph{each constraint} to the SDP objective is at least $1-\delta$, the SDP value
still remains at least $1-\delta$.
Also, if the set of points is sufficiently dense, each vector only moves by a small amount (say
$o(\eps)$) and we still have that for every ($\eps$-good) constraint, for all $i,j \in [k]$ and $b,b' \in \pmone$
\[
\abs{\ip{\vtwo{\y_i}{1},\vtwo{\y_i}{-1}}} \le O(\eps)
\quad \text{and} \quad
\abs{ \ip{\vtwo{\y_i}{b}, \vtwo{\y_j}{b'}} - \sum_{\alpha \in \pmone^t \atop \alpha(i) = b, \alpha(j) = b'}
  \vartwo{(\yk)}{\alpha}^{(\zeta)} } \le O(\eps) \mper
\]
Thus, we have an SDP solution with value at least $1-\delta$, which satisfies the above inequalities
approximately and the rest of the SDP constraints exactly. At this point we can apply the
``surgery'' and ``smoothening'' procedures of Raghavendra and Steurer \cite{raghsteu}
(Lemmas 5.1 and 5.2), which transform an SDP solution satisfying the above constraints
approximately, to new solution for the basic SDP relaxation in Figure \ref{fig:basic-sdp}, while
only losing $O(\sqrt{\eps} \cdot k^2)$ in the SDP value.
%
%%%%%%%%%%%%%%%%%%%%%%%%%%%%%%%%%
\newcommand{\vt}[2]{\tilde{\bf v}_{(#1,#2)}}
\newcommand{\nt}{\tilde{\nu}}
%%%%%%%%%%%%%%%%%%%%%%%%%%%%%%%%%
Note that for an instance of $\maxkcsp(f)$, the variables $\vartwo{S_C}{\alpha}$ define a
distribution on the set $S_C$. Let this be denoted by $\nu_C$. The following is a combination of
Lemmas 5.1 and 5.2 from \cite{raghsteu}.
\begin{lemma}[\cite{raghsteu}]\label{lem:basic-correction}
Let $\Phi$ be an instance of $\maxkcsp(f)$ in $n$ (Boolean) variables such that there exist vectors
$\vtwo{i}{b}$ for all $i \in [n]$ and $b \in \pmone$, and distributions $\nu_C$ over $\pmone^{S_C}$
for all $C \in \Phi$, satisfying
\[
\abs{\ip{\vtwo{i}{1},\vtwo{i}{-1}}} \le \eps
~~\text{and}~~
\abs{ \ip{\vtwo{i}{b}, \vtwo{j}{b'}} - \Prob{x \sim \nu_C}{(x_i = b) \wedge (x_j = b')}} \le \eps
\quad \forall C \in \Phi,~ i,j \in S_C
\mper
\]
Then there exist vectors $\inbraces{\vt{i}{b}}_{i \in [n], b \in \pmone}$ and distributions
$\inbraces{\nt_C}_{C \in \Phi}$ such that
\[
\ip{\vt{i}{1},\vt{i}{-1}} = 0
~~\text{and}~~
\ip{\vt{i}{b}, \vt{j}{b'}} = \Prob{x \sim \nt_C}{(x_i = b) \wedge (x_j = b')}
\quad \forall C \in \Phi,~ i,j \in S_C \mper
\]
Also, we have that for all $i, b$, $\norm{\vtwo{i}{b} - \vt{i}{b}} = O(k^2 \cdot \sqrt{\eps})$ and
for all $C \in \Phi$, $\norm{\nu_C - \nt_C}_1 = O(k^2 \cdot \sqrt{\eps})$.
\end{lemma}
Choosing $\eps = O(\delta^2/k^4)$ and applying the above lemma, we obtain a solution to the SDP in
Figure \ref{fig:basic-sdp} with value at least $1-2\delta$.

\section{Proof of the LP Dichotomy Theorem}\label{lpthm:sec}

Note that a dichotomy theorem for SDPs need not imply a
similar dichotomy theorem for LPs. For example, 2LIN is
approximable (very well) via Goemans-Williamson
SDP but the same predicate appears approximation resistant to a super-constant
number of rounds of the
Sherali-Adams LP \cite{delavega, yury}.
Nevertheless, our characterization in Theorem~\ref{sdp:thm} can be used in a more or less
black-box fashion so as to yield a {\it syntactically} similar characterization in the
LP case.
The integrality gap construction however
 needs substantial work.
The Feige-Schechtman approach is not sufficient to construct
integrality gap instances for the Sherali-Adams LP, which
is our focus in this section. We overcome this difficulty by
generalizing the construction of de la Vega and Kenyon ~\cite{delavega}.
A noteworthy detail of our construction  is that our technique, even though it is probabilistic,
requires a more subtle argument for both completeness and soundness.
This is unlike many previous
constructions, which typically consider a  uniformly
random instance (or a minor modification of it) from the
family of all possible instances.

Recall that in Definition~\ref{cvxptppred:dfn}, we define a
moment matrix $\zeta$ consisting of the first and second moments of a distribution $\nu$
supported on $f^{-1}(1)$. The second moments also match with the inner products of the
SDP vectors. In the LP case, the LP solution only gives first moments. Still, we
are able to use a {\it dummy} ~setting for the second moments and reduce the LP
case to the SDP case! The dummy setting ensures that the corresponding covariances are zero
and hence the Gaussians with matching first and second moments are independent.

We describe this trick formally now. Given a predicate $f:\{-1,1\}^k\to\B$, recall  that
$\cal{D}(f)$ is the set of all probability distributions over $f^{-1}(1)$. We define
a compact body $\tilde{{\cal C}}(f)$ that replaces the role of the polytope ${\cal C}(f)$ before.

\begin{definition}
For $\nu \in \cal{D}(f)$, we let $\tilde{\zeta}(\nu)$ denote the
$(k+1) \times (k+1)$ symmetric {\deffont moment matrix}:
\begin{align*}
\forall i\in \{0\} \cup [k]:\ \tilde{\zeta}(i,i) &~=~ 1 \mcom \\
\forall i\in[k]:\ \tilde{\zeta}(0,i) &~=~ \Ex{x \sim \nu}{x_i} \mcom \\
\forall i,j\in[k], i\neq j:\ \tilde{\zeta}(i,j) &~=~ \tilde{\zeta}(0,i)\cdot\tilde{\zeta}(0,j) \mper
\end{align*}
Also, let $\tilde{\cal{C}}(f) \subseteq \R^{(k+1)\times (k+1)}$
denote the compact (but not necessarily convex) set of all such moment matrices:
\[ \tilde{\cal{C}}(f) ~:=~\{ \tilde{\zeta}(\nu) \suchthat \nu \in\cal{D}(f)\}. \]
\end{definition}

Note that if $g_1, \ldots, g_k$ are correlated Gaussians with $\ExpOp[g_i] = \tilde{\zeta}(0,i)$,
$\ExpOp[g_i^2] =1 $ and
$\ExpOp[g_i g_j] = \tilde{\zeta}(0,i)\cdot \tilde{\zeta}(0,j)$ for $i \not= j$,
then these are independent with given means.

The entire argument in Section \ref{sdpthm:sec} can be repeated as is except for two changes: firstly,
the body $\tilde{C}(f)$ is now used throughout the argument. Secondly, in Section \ref{Lge0rnd:sbs}, towards designing
an algorithm, the $k$-round Sherali-Adams LP is solved instead of the basic (SDP) relaxation. The
SDP solution enables us to generate a (global)
sequence of correlated Gaussians, one for every CSP variable,
so that for every CSP constraint $C$, the $k$ Gaussians corresponding to that constraint have
first and second moments given by $\zeta = \zeta(\nu(C))$, where $\nu(C)$ is the local distribution on
that constraint. In the LP case however, we only have access to (globally consistent) first moments (i.e. biases)
of the local distributions $\nu(C)$. But we
can still generate a (global) sequence of correlated Gaussians as before whose first and second
moments corresponding to the constraint $C$ are $\tilde{\zeta}(\nu(C)))$. These are simply independent
unit $\ell_2$-norm Gaussians with first moments equal to the biases computed by the LP!

As before, depending on the
value of the limit $L$, we get a dichotomy, i.e. the following analogs of Theorems~\ref{sdp-algo:thm}
and~\ref{sdp-measure:thm} respectively. When $L > 0$, the predicate is approximable
via a $k$-round Sherali-Adams LP.

\begin{theorem}\label{lp-algo:thm}
If $L >0$, then there exists a $k$-round LP
rounding algorithm such that given an instance $\Phi$ with $\lpopt(\Phi) \ge 1-\eps$
(for sufficiently small $\eps > 0$), we have
$\ExpOp_{\psi}{\left[ \round_{\psi}(\Phi) \right]} \geq \rho(f) + L/2$.
\end{theorem}

When $L=0$, as in Theorem \ref{sdp-measure:thm}, we get a measure $\Lambda$ on the body
$\tilde{C}(f)$ that is vanishing in the sense therein. However we note that since the second
moments are just dummy, we might as well restrict everything to the body ${\cal C}^*(f)$ that
is the projection of $\tilde{C}(f)$ onto the first moments (and thus ${\cal C}^*(f)$ is
simply the convex hull of $f^{-1}(1)$). Denoting the measure on ${\cal C}^*(f)$ so obtained by
$\Lambda^*$, we get:

%We note that obtaining Theorem~\ref{lp-algo:thm} from Theorem~\ref{sdp-algo:thm}
%does not require us to solve SDPs (only LPs). This is because the
%algorithm simply samples independent Gaussians instead of sampling
%correlated Gaussians, where the correlations were being specified
%by the inner product of SDP vectors.
%\Mnote{Is the above explanation sufficient? I do not have a short alternative.}

%Finally, in a slight abuse of notation, we ignore the second moments and consider $\zeta$ to be a
%vector in $k$ dimensions with only the first moments.
%We redefine $\cal{C}(f)$ as follows.
%\begin{definition}\label{lpcvxptppred:dfn}
%For $\nu \in \cal{D}(f)$, we let $\zeta(\nu)$ denote the
%$k$-dimensional vector for $\nu$ such that for
%$x\in\{-1,1\}^k$:
%%
%\begin{align*}
%\forall i\in[k]:\ \zeta(i) &~=~ \Ex{x \sim \nu}{x_i} \mper
%\end{align*}
%%
%Also, let $\cal{C}(f)$ denote the compact set of all such first
%moment vectors:
%%
%\[ \cal{C}(f) ~:=~\{\zeta(\nu) \suchthat \nu \in\cal{D}(f)\}. \]
%%
%\end{definition}
%\Mnote{This is inconsistent with the preliminaries. There we use $\eta(i)$. But here $\eta$ is used
%  for the noise in the distribution later.}

%The following analog of Theorem~\ref{sdp-measure:thm} follows without
%any changes to the SDP argument except for the modification to $\zeta$
%discussed above.

\begin{theorem}\label{lp-measure:thm}
If $L=0$, then there exists a probability measure $\Lambda^*$ on $\C^*(f)$
such that for all $t \in [k]$, and a uniformly random choice of $S$ with  $|S|=t$,
$\pi : [t]\to [t]$ and $b\in \pmone^{t}$, the following signed measure on
$t$-dimensional vectors,
\begin{equation}
\Lambda^{*,(t)} ~:=~
\ExpOp_{|S| = t} ~\ExpOp_{\pi : [t] \to [t]} ~\Ex{b \in \pmone^{t}}{ \hf(S) \cdot  \inparen{\prod_{i
    = 1}^t b_i}
  \cdot \Lambda_{S,\pi,b}^*}\label{mainlp2:eqn}
\end{equation}
is identically zero.
\end{theorem}

In the next section, we show how the existence of the measure $\Lambda^*$ leads to a
$(1-o(1), \rho(f) + o(1))$ integrality gap for a super-constant number of rounds of the
Sherali-Adams LP.

\subsection{The Integrality Gap Instance}
The integrality gap construction for the Sherali-Adams LP is rather different from that
for SDPs.
One important aspect of our construction is that unlike many previous
constructions, e.g. \cite{yury, delavega, schoen, madhur},
our construction requires a non-trivial proof
of {\em both} the soundness and completeness parts.
The proof of the soundness part is similar to
that in the SDP case and for the completeness part we generalize the
construction in de la Vega and Kenyon~\cite{delavega}.
A formal description of our instance follows.

Let $f:\{-1,1\}^k \mapsto \{0,1\}$ be any predicate with a measure $\Lambda^*$ as in
Theorem \ref{lp-measure:thm}.
Note that now $\C^{*}_{\delta}(f)$ is simply the body
$\{(1-\delta) \cdot \zeta ~|~ \zeta \in \C^{*}(f)\}$. Since it's just a scaling, there is a
vanishing measure over $\C^*(f)$ if and only if there is such a measure over $\C^*_{\delta}(f)$. We
will assume $\Lambda^*$ is over $\C^*_{\delta}(f)$ for reasons similar to the ones in the SDP integrality gap.
We will finally need to choose $\delta \geq \sqrt{\eps}$, for the constant $\eps$ below.

Fix a small enough $\eps > 0$ and let $s = \lceil \frac{1}{\epsilon} \rceil$.
Partition the interval $[0,1]$ into $s+1$ disjoint sets $I_0, I_1, \ldots, I_{s}$ where
$I_0  = \{ 0 \}$ and  $I_1,\ldots,I_s$ are contiguous equal length intervals that
partition $(0,1]$, each being open at its left endpoint and closed at the right endpoint.
% set of $s = \lceil \frac{1}{\eps} \rceil$ contiguous
%intervals $I_1,\ldots,I_s$\footnote{Each interval is closed at its minimum
%endpoint and open at its maximum endpoint. The last interval is closed
%at both endpoints.} % We have $\lceil \frac{1}{\eps}\rceil$ intervals in all.}.
For each interval $I_i$, we define a set (layer) of
$n$ variables $X_i$.  %:=\{x_{i_1},...,x_{i_n}\}$.
Thus the total number of variables in the
CSP instance is $(s+1)\cdot n$. Our constraints are generated
by the following algorithm. %, up to a small caveat that we will fix later.
\begin{enumerate}
\item Sample   %\footnote{Strictly speaking, it is not clear how to sample from an arbitrary measure
%    $\Lambda^*$. However, as in the case of the SDP integrality gap instance, we can approximate the
%  body $\C^{*}_{\delta}(f)$ by a sufficiently dense set of points, and consider $\Lambda^*$ to be
%a distribution over a finite set of points.}
$\zeta\sim\Lambda^*$.
\item For each $j\in[k]$,
 let $i_j(\zeta)$ denote the index of the interval that contains $|\zeta(j)|$.
Sample uniformly a variable $x_{i_j}$ from the set $X_{i_j(\zeta)}$.
\item If $\zeta(j) < 0$ then negate $x_{i_j}$. If $\zeta(j)=0$ then negate $x_{i_j}$ with probability
$\frac{1}{2}$.
%\begin{enumerate}
%\item If $\zeta(j) < 0$ then negate $x_i$.
%\item If $\zeta(i)=0$ then negate $x_i$ with probability $1/2$.
%\item Otherwise, leave $x_i$ as is.
%\end{enumerate}
\item We have sampled a $k$-tuple of literals. Introduce a constraint $f$ on these literals.
\item Repeat the above procedure $m=\Delta(\eps)\cdot n$ times independently (where
$\Delta(\eps)$ is a sufficiently large constant) and thus generate $m$ constraints.
\end{enumerate}

This completes the description of our CSP$(f)$ instance. % (up to a caveat as we said).

\Mnote{I replaced all $\opt$ below by $\sat(\psi)$.}

Let $\psi$ be any (global) $\{-1,1\}$-assignment to the above instance.
%Let us momentarily assume that case (3b) in our construction does not
%occur for a randomly chosen constraint in our instance. We will show
%how to take care of this special case after we deal with the more general
%case. For $\zeta \sim \Lambda^*$ and $j \in [k]$, let $i_j(\zeta)$ be the
%index of the interval that contains $|\zeta(j)|$. With this somewhat cumbersome notation,
Denoting the fraction of constraints satisfied by $\psi$ by $\sat(\psi)$, we note that
$\ex{\sat(\psi)}$ is equal to the
the probability that a randomly chosen constraint as above
is satisfied by $\psi$.
We prove that this probability is precisely $\rho(f)$.
We can write $\ex{\sat(\psi)}$, which equals the probability of satisfying a random constraint as
above, as:
\[\ex{\sat(\psi)} ~=~ \Ex{\zeta\sim\Lambda^*,\atop{x_{i_j} \in X_{i_j(\zeta) }   }}{f({\sf sign}(\zeta(1))\cdot\psi(x_{i_1}),...,{\sf sign}(\zeta(k))\cdot\psi(x_{i_k}))}.\]
Here the function ${\sf sign}(\cdot)$ is $-1$ if its argument is strictly negative,  $+1$ if its argument
is strictly positive and ${\sf sign}(0)=0$.
%Since the function is discontinuous, it is not clear that the above expression is even meaningful
%(note that $\Lambda^*$ is an arbitrary measure and the expectation is really an integral over
%$\Lambda^*$). We ignore this issue for now.
%Note that the variable $x_{i_j}$ is randomly chosen from the {\it layer} $X_{i_j(\zeta)}$.
Using the Fourier expansion of $f$,
\[\ex{\sat(\psi)} ~=~\Ex{\zeta\sim\Lambda^*,\atop{x_{i_j}\in X_{i_j(\zeta)}  }}{\sum_{S\subseteq[k]}\HF(S)\prod_{j\in S}\left({\sf sign}(\zeta(j))\cdot\psi(x_{i_j})\right)}.\]
Since $x_{i_j}$ is randomly chosen from the layer $X_{i_j(\zeta)}$, we can move the expectation over the
choice of $x_{i_j}$ inside and get
\[
\ex{\sat(\psi)} ~=~ \rho(f)+\sum_{t=1}^k\Ex{\zeta\sim\Lambda^*}{\sum_{|S|=t}\HF(S)\prod_{j\in
    S}\left({\sf sign}(\zeta(j))\cdot\Ex{\atop{x_{i_j} \in X_{i_j(\zeta)}
      }}{\psi(x_{i_j})}\right)}.
\]
The expectations inside are the average values of $\psi$ over the respective layers and hence in
$[-1,1]$.  Define a function $\tilde{\psi}: [-1,1] \mapsto [-1,1]$ that is
odd,  in particular $\tilde{\psi}(0)=0$ and for each $i \in [s]$, is
constant on the interval $I_i$ where it takes the value $\Ex{x_i \in X_i}{\psi(x_i)}$.
Thus the innermost expectation is really $\tilde{\psi}(|\zeta(j)|)$ and combining it with ${\sf sign}(\zeta(j))$ and using the oddness of $\tilde{\psi}$,
\begin{equation}\label{eqn:tildepsi}
\ex{\sat(\psi)} ~=~\rho(f)+\sum_{t=1}^k\Ex{\zeta\sim\Lambda^*}{\sum_{|S|=t}\HF(S)\prod_{j\in S}
\tilde{ \psi} (\zeta(j))
%\left({\sf sign}(\zeta(j))\cdot\Ex{\atop{x_{i_j} \in X_{i_j(\zeta)}   }}{\psi(x_{i_j})}\right)
}.
\end{equation}
%Assume for now that $\tilde{\psi}$ is continuous (which clearly it isn't, being piecewise constant)
%so that the above expectation is meaningful.
%Under this assumption,

We observe that for every $t \in [k]$, the expectation above vanishes.
This is because, up to a multiplicative factor of $\binom{k}{t}$, the expectation is same as
$$\Ex{\zeta\sim\Lambda^*}{ \ExpOp_{|S|=t} ~\ExpOp_{\pi: S\mapsto S}
~\Ex{b\in \{-1,1\}^S} { \HF(S) \left( \prod_{j\in S} b_j \right) \left( \prod_{j\in S}
\tilde{ \psi} ( b_j \zeta(\pi(j))) \right) } },  $$
which in turn is same as
 $$ \int ~\left( \prod_{j=1}^t  \tilde{\psi} (\zeta'(j)) \right)~  d\Lambda^{*, (t)} (\zeta').  $$
This integral vanishes since $\Lambda^{*, (t)}$ vanishes identically and we  are done.

\Mnote{Every $\sat(\psi)$ above needs to be replaced by $\ex{\sat(\psi)}$.}

Now we prove the soundness property of the CSP instance.
Since each constraint is picked independently, a Chernoff bound
implies that the probability that $\sat(\psi)$ is outside $[\rho(f)-\eps, \rho(f)+\eps]$,
for any fixed $\{-1,1\}$ assignment $\psi$, decays exponentially in $m$.
For large enough $\Delta(\eps)$, one may then take a  union bound over all $2^{(s+1)\cdot n}$ assignments
and obtain the following claim.
\begin{lemma}\label{soundlp:lm}
For every $\eps > 0$,  there exists a sufficiently large constant $\Delta(\eps)$ such that w.h.p.
over the choice of the \maxkcsp$(f)$ instance, it holds that
for every assignment $\psi$ to the instance, $\sat(\psi) \in [\rho(f)-\eps,\rho(f)+\eps]$.
\end{lemma}

Let $G$ denote the natural constraint vs variable bipartite graph
of our instance.
In other words, $G$ has a vertex for each constraint and each variable
and there is an edge between a constraint and a variable if and only
if the variable occurs in that constraint. Strictly speaking, $G$ is a multi-graph since
in a constraint, the same variable may appear twice or more.
We show that after deleting a small fraction of
vertices, $G$ has high girth, in particular eliminating cycles of length two, i.e. multiple edges.

\begin{lemma}\label{comb:lm}
The constraint vs variable graph $G$  has $(k\Delta)^{O(g)}$ cycles
of length at most $g$, in expectation.
\end{lemma}
\begin{proof}
%Suppose we are given a constraint graph $G$ from the above experiment.
Recall that the variable
vertices of $G$ correspond to the set $[n] \times \{0,1,\ldots,s\}$. We think of these as arranged
in an $n \times (s+1)$ array.
Suppose we contract the set of $s+1$ vertices in $j^{th}$ row
into a single vertex $x_j$ for $j \in [n]$.
We will get a %constraint
bipartite multi-graph $G'$ such that the
set of variables of each of the $m$ constraints is picked uniformly
from the set of variables $\{x_j:j\in[n]\}$.
Note that under this operation there exists a unique cycle of length
at most $g$ in $G'$ for every cycle of length at most $g$ in $G$.
Moreover, the probability of obtaining that cycle in $G'$ is the
at most the probability of obtaining that cycle in $G$.
Hence, it will suffice to bound the expected number of cycles
of length at most $g$ in $G'$.
We have reduced our problem to obtaining a bound on the girth of $G$
to the following combinatorial problem.

We have a random bipartite multi-graph $H:=(U,V)$, where the
edge set $E(H)$ is selected by independent sampling (with repetition)
of $k$ vertices from $V$ ($|V|=n$), for each of the $m$ vertices in $U$.
We need a bound on the expected number of cycles of length at most $g$.

Consider any cycle $C(h)$ of length $2h$ in $H$.
Half the vertices in $C(h)$ come from $U$ and half come from $V$.
The probability that a given vertex in $U$ and given vertex in $V$ have an edge between them is at
most $k/n$.
Therefore, the expected number of cycles of length exactly $2h$ in $H$ is
bounded by:
\begin{equation}
n^h \cdot (\Delta\cdot n)^h \cdot \inparen{\frac{k}{n}}^{2h} ~\le~ (k\Delta)^{O(h)}.\label{cyccount:eqn}
\end{equation}

The above is a geometric progression in $h$, since $k$ and $\Delta$ are
constants. Hence, the expected number of cycles of length at most $2h$
in $H$ is also bounded by $(k\Delta)^{O(h)}$.
\end{proof}

For $g = c \cdot \log n$ for a sufficiently small constant $c$ depending on $k$ and  $\Delta$,
we may delete $o(n)$ constraints from our instance so as to eliminate all cycles of length at most $g$.
This still preserves the property that for every assignment $\psi$ to the instance $\sat(\psi)
\in [\rho(f) - \eps, \rho(f)+\eps]$, possibly with a negligible change in parameter $\eps$ that we
ignore.  Moreover, a union bound implies that with high probability every vertex in our constraint
bigraph $G$ has bounded degree.
Therefore, Lemmas~\ref{soundlp:lm} and~\ref{comb:lm} imply
the following lemma.
\begin{lemma}\label{finallpinst:lm}
For all large enough $n$ and every $\eps>0$, there exists a  \maxkcsp$(f)$ instance  with $n$
variables and  $m=\Delta n$ constraints  such that its constraint vs
variable graph $G$  has girth  $\Omega(\log n)$, every vertex in $G$ has
bounded degree and every assignment to the instance satisfies between $[\rho(f)-\eps, \rho(f)+\eps]$
fraction of the constraints.
\end{lemma}
Also note that large girth in particular implies that any two constraints in our instance share at
most one variable.

For the remainder of this section, we assume that our $\maxkcsp(f)$
instance is given by some fixed constraint graph $G$, as in
Lemma~\ref{finallpinst:lm}.
Next, we need to show that the Sherali-Adams LP has an optimal solution with value
$1-o(1)$ for instance given by $G$.
Our task is to define locally consistent distributions over all subsets of variables of size at
most $r$ (we will finally be able to choose $r = \Omega(\log \log n)$).
To this end we will first define distributions which are
approximately consistent,
and then use a result by Raghavendra and Steurer~\cite{raghsteu2}
to make the distributions exactly consistent.

Recall that every constraint $C$ in our instance was generated using a $\zeta(C) \in
{\cal C}^*(f)$.
Let $\bnu(C)$ be a distribution on $f^{-1}(1)$ such that $\zeta(C) = \zeta(\bnu(C))$.
Note that $\bnu(C)$ is a distribution on the \emph{literals} involved in constraint $C$,
with the biases of the literals being $(\zeta(1),\ldots,\zeta(k))$.
If  a constraint $C$ is on variables in layers $i_1,\ldots,i_k$ respectively, then the
biases of these variables according to $\bnu(C)$ are
\[
(\abs{\zeta(1)},\ldots,\abs{\zeta(k)})
~=~
(p_{i_1}, \ldots, p_{i_k}) \mper
\]
respectively so that $p_{i_j}\in I_{i_j}$. The biases of the variables are always non-negative since
we negate the $j^{th}$ variable only if $\zeta(j) < 0$ (and with probability $1/2$ when $\zeta(j) = 0$).

The local distributions we define on sets of size $r$ will have the property that all variables in
the same layer $X_i$ have the same bias.
For each interval $I_i$ with $i \in \{0,\ldots,s\}$, choose an arbitrary point $t_i \in I_i$. We
will first modify the distributions $\bnu(C)$ such
 that  all the variables in layer $X_i$ have bias exactly $t_i$.
Since $p_{i_j} \in I_{i_j}$, we have $|p_{i_j} - t_{i_j}| \leq \eps$.
Thus we can change the biases of the variables as desired with a slight perturbation of the
distributions $\bnu(C)$. However this incurs a
slight loss in the completeness parameter: the resulting distribution $\bnu'(C)$ is now only
$(1-o(1))$-supported on $f^{-1}(1)$.

%%%%%%%%%%%%%%%%%%%%%%%%%%%%%%%%%%%%%%%
\newcommand{\sgn}{{\sf sign}}
%%%%%%%%%%%%%%%%%%%%%%%%%%%%%%%%%%%%%%%
\begin{claim}\label{bias-correction:clm}
Let the distribution $\bnu(C)$ be as above such that the biases for the literals in $C$ are given
by $(\zeta(1),\ldots, \zeta(k))$. Also, let $t_{i_1},\ldots,t_{i_k}$ as above be the desired biases
for the variables such that $\abs{t_{i_j} - \abs{\zeta(j)}} \leq \eps$.
Then there exists a distribution $\bnu'(C)$ on $\pmone^{k}$  such that
$\norm{\bnu(C)-\bnu'(C)}_1 = O(k \cdot \sqrt{\eps})$ and
\[
\forall j \in [k] \quad \Ex{z \sim \bnu'(C)}{z_j} = \sgn(\zeta(j)) \cdot t_{i_j} \mper
\]
Thus, the biases for the variables, when the
literals are sampled according to $\bnu'(C)$ are exactly $(t_{i_1},\ldots,t_{i_k})$ since the
$j^{th}$ variable is negated only if $\sgn(\zeta(j)) = -1$.
\end{claim}
\begin{proof}
Let $r_j = \sgn(\zeta(j)) \cdot t_{i_j}$ be the desired bias of the $j^{th}$ literal. Then,
$\abs{\zeta(j) - r_j} \leq \eps$ for all $j \in [k]$
We construct a sequence of distributions $\bnu_0,\ldots,\bnu_k$ such that $\bnu_0 = \bnu(C)$ and
$\bnu_k = \bnu'(C)$. In $\bnu_j$, the biases of the literals are $(r_1,\ldots,r_j,\zeta(j+1),\ldots,\zeta(k))$.

The biases in $\bnu_0$ satisfy the above by definition.
We think of the distributions over $z \in \pmone^k$.
We obtain $\bnu_{j}$ from $\bnu_{j-1}$ as,
\[
\bnu_{j} = (1-\tau_j) \cdot \bnu_{j-1} + \tau_j \cdot D_j \mcom
\]
where $D_j$ is the distribution in which all bits, except for the $j^{th}$ one, are set
independently according to their biases in $\bnu_{j-1}$. For the $j^{th}$ bit, we set it to
$\sgn(r_j-\zeta(j))$ (if $r_j-\zeta(j) = 0$, we can simply proceed with $\bnu_j = \bnu_{j-1}$).
The biases for all except for the $j^{th}$ bit are unchanged. For the $j^{th}$ bit, the bias now
becomes $r_j$ if
\[
r_j = (1-\tau_j) \cdot \zeta(j) + \tau_j \cdot \sgn(r_j-\zeta(j))
~\Longrightarrow~
\tau_j \cdot (\sgn(r_j-\zeta(j)) - r_j) = (1-\tau_j) \cdot (r_j - \zeta(j)) \mper
\]
Since $\zeta \in \C^*_{\delta}(f)$ for $\delta \geq \sqrt{\eps}$, we know that
$\abs{\sgn(r_j-\zeta(j)) - r_j} \geq O(\sqrt{\eps})$. Also, $\abs{r_j - \zeta(j))} \leq \eps$ by
assumption. Thus, we can choose $\tau_j = O(\sqrt{\eps})$ which gives that $\norm{\bnu_{j} -
  \bnu_{j-1}}_1 = O(\sqrt{\eps})$. The final bound then follows by triangle inequality.
\end{proof}
The distribution over the literals of $C$, given by the above claim also gives a distribution for
the variables in $S_C$.
We now refer to the distribution over $\pmone^{S_C}$ given by Claim \ref{bias-correction:clm} as
$\nu(C)$.
We will need to modify the distributions $\nu(C)$ a little further before we use them to define the
local distributions over sets of size $r$.

\begin{definition}
Given a constraint $C$ and $\eta >0$, let $U_C$ denote the following
distribution on $\pmone^{S_C}$
\[
U_C ~\defeq~ (1-\eta) \cdot \nu(C) + \eta \cdot U_k \mper
\]
where $U_k$ denotes the uniform distribution on $\pmone^k$.
For $\alpha$ a partial assignment to variables in $C$, let $U_{C,\alpha}$
denote the distribution $U_C$ conditioned according to $\alpha$.
\end{definition}

Recall that the distributions $\nu(C)$ are defined so that the variables in the layer $X_i$ have
bias exactly $t_i$. The following observation will be extremely useful.
\begin{remark}\label{bias:rem}
The bias of a variable in layer $X_i$ is exactly $(1-\eta) \cdot t_i$, when assigned according to
$U_C$, for any constraint $C$ containing that variable.
\end{remark}

%%%%%%%%%%%%%%%%%%%%%%%%%%%%%%%%%%%
\newcommand{\ball}{B^{(d)}}
\newcommand{\varg}{{\cal V}_G}
\newcommand{\constg}{{\cal C}_G}
\newcommand{\dist}{\mathsf{dist}_G}
\newcommand{\m}{m}
%%%%%%%%%%%%%%%%%%%%%%%%%%%%%%%%%%%

Let $\varg$ denote the set of variable vertices in the bipartite constraint-variable graph
$G$ and let $\constg$ be the set of constraint vertices.  Let $\dist(u,v)$ denote the shortest path
distance in $G$ between two vertices $u$ and $v$.
Given a set $S$ of variables in $G$ and an \emph{even} number $d \in \N$, we define
\[
\ball(S) ~:=~ \inbraces{u\in \varg \cup \constg ~:~ \dist(u,S) \le d} \mper
\]
We will choose $d$ to be sufficiently small so that $|\ball(S)| \leq girth(G)$ and hence the set
$\ball(S)$ is a forest.
Also, since $d$ is assumed to be even and $S \subseteq \varg$, the leaves of each component in
$\ball(S)$ are variable vertices in $G$.
Let $\ball(S) =\cup_i \ball(S_i)$, where each $\ball(S_i)$ is a maximal connected
component in $\ball(S)$. We now describe a probabilistic process, which will be
used to defined a probability distribution $\m_S$ on $\pm 1$ assignments to the set $S$. We will use
this process to generate a random assignment to all the variables in $\ball(S)$, and hence also in
$S$.

First, we fix an arbitrary ordering of all variables in $G$. This also gives an ordering of all the
constraints in $S$ (depending on the variables involved in each constraint). We generate an
assignment for all variables in $\varg \cap \ball(S)$. The assignment for each component $\ball(S_i)$ is
generated independently of the other components by the following process:

\begin{enumerate}
\item Pick the least variable $x \in S\cap \ball(S_i)$. If $x$ belongs to the layer
  $X_j$, assign it to be 1 with probability $(1+(1-\eta)\cdot t_j)/2$ and $-1$ with probability
  $(1-(1-\eta)\cdot t_j)/2$, so that the bias is $(1-\eta) \cdot t_j$.

\item Traverse $\ball(S_i)$ in a breadth-first manner, starting from the vertex corresponding to the
  least variable $x$ (and using the above ordering on variables and constraints).
\begin{itemize}
\item[-]
  When visiting a vertex corresponding to a constraint $C$, if $\alpha$ is the
  partial assignment to the variables assigned so far, generate an assignment for the remaining
  variables in $C$ according to $U_{C,\alpha}$.

\item[-]
  When visiting a vertex corresponding to a variable, its value is already assigned by its parent
  constraint-vertex. We simply proceed to its children, which are new constraint vertices.
\end{itemize}
%
% \item Let $C$ be the least constraint on $v$.
% Sample an assignment $\alpha$ from $\{\pm 1\}^k$ according to the distribution
% $U_C$ and assign it to $C$.
%  \item For every constraint $C'\in B_i(S)$ and $C'\cap C\neq\emptyset$
% sample an assignment $\beta$ from $U_{C',\alpha}$.
% \item Repeat steps 3 and 4 in a breadth first manner\footnote{The traversal
% is with respect to the fixed lexicographic ordering on all constraints.}
% on each constraint supported by variables in $B(S)$.
\end{enumerate}
Note that since $\ball(S_i)$ is a tree, when visiting a constraint vertex $C$ we will have \emph{at most
  one} of the variables in $C$ assigned before. We will assign the remaining variables
according to $U_C$ conditioned on the value of this one variable.

This process above defines a probability distribution $\m_S$ on the $\pm 1$
assignments to the variables in $B(S)$, and hence also on $\pmone^S$ as long as $\ball(S)$ is a
forest. We can obtain a bound on the size of such sets $S$ in terms of the girth and the degree of
the constraint graph $G$.

\begin{claim} \label{setsize:clm}
Let the girth of the constraint graph $G$ be equal to $g$ and let the degree of every vertex in $G$
be at most $D$. Then the distribution $\m_S$ is well-defined for all sets $S$ with $|S| < g / D^d$.
\end{claim}
\begin{proof}
Since the degree of every vertex at most $D$, we have that
\[
\abs{\ball(S)} ~\leq~ |S| \cdot D^d ~<~ g \mper
\]
Hence, we have that $\ball(S)$ is a forest and the distribution $\m_S$ is well-defined.
\end{proof}

% \begin{remark}
% For $|S|=o(\log n)$ the measure $m_S$ is well defined.
% \end{remark}
%
% We may also extend the assignment
% to the entire set of variables in our instance.
% Finally, for every variable $v\not\in B(S)$, define $m_S(v=1)=m_S(v=-1)=1/2$.
% \Mnote{If $v \notin B(S)$, then $v \notin S$. Why are we assigning it a probability in $m_S$.}

We need the following lemma to show that the objective value of our Sherali-Adams LP solution is
close to $1$.
\begin{lemma}\label{almsatsol:lm}
For every constraint $C$ supported on variables $S_C$,
the distribution $\m_{S_C}$, has at least $(1-\eta - O(k\sqrt{\eps}))$-fraction of its
probability mass on the accepting assignments of $C$.
\end{lemma}
\begin{proof}
Note that for any constraint $C$ at most 1 variable can be fixed by a partial assignment to some
other variables by our process for generating assignments.
At this point, we assign all variables in $C$ according to $U_C$
conditioned on the value of this one variable. Hence, the joint distribution of all the variables in
$C$ is always according to $U_C$.

Also, $U_C$ is obtained by taking $\nu(C)$ with probability $1-\eta$ and uniform with probability
$\eta$. By Claim \ref{bias-correction:clm}, $\nu(C)$ is $O(k\sqrt{\eps})$-close to a
distribution which corresponds to a point in $\C^{*}_{\delta}(f)$ and has mass at least $1-\delta$
over accepting assignments. Thus, $U_C$ has mass at least $1-\eta-\delta-O(k\sqrt{\eps})$ on
accepting assignments. Using $\delta = \sqrt{\eps}$ proves the bound.
%
% In order to prove the statement we need to show that if
% one of the variables in some constraint $C$ is fixed
% to $\pm 1$ by a partial assignment $\alpha$ then $U_C$
% still has at least $1-\eta$ of its probability mass
% on the accepting assignments of $C$.
% However, the only way $U_{C,\alpha}$ has less than $1-\eta$
% mass on accepting assignments of $C$ is if there are no
% accepting assignments in $f_C^{-1}(1)$, which are consistent
% with $\alpha$.
%
% Note that $\alpha$ fixes at most one variable of $C$, say $x_i$.
% This would imply that all members in $f_C^{-1}(1)$ are always
% $1$ (say) on variable $x_i$, while constraint $C'$, in the
% same component of $B(S)$, has all members of $f_{C'}^{-1}(1)$
% fixed to $0$ on $x_i$.
% Therefore, $x_i$ must appear in negated form in exactly
% one of $C$ or $C'$.
% However, this will give a contradiction. The definition
% of $\Lambda^*$ implies that it would be entirely supported
% on either $1$ or $-1$ (not both).
% Therefore, variable $x_i$, in our $\maxkcsp$ instance, will
% always either remain negated or not-negated in the entire
% instance. Therefore our statement follows.
\end{proof}
Note that the definition of $\m_S$ implicitly depends on the ordering of variables. The following
lemma shows that the distributions in fact \emph{do not} depend on the ordering.
\begin{lemma}\label{orddontmatter:lm}
Given a set $S \subseteq \varg$ and an ordering $\omega$ of all the variables in $\varg$, let
$\m_{S,\omega}$ denote the distribution $\m_S$ when defined according the ordering $\omega$. Then,
for any $\alpha \in \pmone^S$ and any two orderings $\omega$ and $\omega'$, we have that
\[
\m_{S,\omega} (\alpha)~=~ \m_{S,\omega'}(\alpha) \mper
\]
\end{lemma}
\begin{proof}
Since the distributions in different components of $\ball(S)$ are independent, it is sufficient to
prove the lemma for the case when $\ball(S)$ is a tree (instead of a forest). We will, in fact,
prove that the probability for any assignment $\beta \in \pmone^{\varg \cap \ball(S)}$ is the same
regardless of the ordering $\omega$. Since $S \subseteq \varg \cap \ball(S)$, this implies the lemma.

Let $\m_{S,\omega}(\beta)$ denote the probability of the assignment $\beta \in \pmone^{\varg \cap
  \ball(S)}$. Note that since the leaves of $\ball(S)$ must correspond to variables (since $d$ is
even), for each constraint $C \in \ball(S)$, we must have that $S_C \sub \varg \cap \ball(S)$, where
$S_C$ denotes the set of variables involved in the constraint $C$.
For $C \in \ball(S)$, let $\beta_{|C}$ denote $\beta$ restricted the set $S_C$.

We now compute the probability for the assignment $\beta$. Suppose that
at some intermediate step in the breadth first traversal for $m_{S,\omega}$
one has fixed an assignment $\beta' \in \{\pm 1\}^R$ for a set $R \sub \varg \cap \ball(S)$,
where $\beta_{|R}=\beta'$. Let $C$ be the next constraint-vertex visited by the traversal.
Using $\beta_1 \circ \beta_2$ to denote the concatenation of two assignments $\beta_1$ and
$\beta_2$, we have
\[
\m_{S,\omega}(\beta' \circ \beta_{|C}) ~=~ \m_{S,\omega}(\beta') \cdot U_{C,\beta'}(\beta_{|C}) \mcom
\]
where $U_{C,\beta'}(\beta_{|C})$ is the probability
that constraint $C$ gets an assignment $\beta_{|C}$
conditioned on the event that variables in $R$ were
assigned according to $\beta'$.

Since $B(S)$ is a tree, there is exactly one variable in $R$, say $x_j$, which is also present in
$C$ (this variable is the parent vertex of $C$). We can then write the above as
\[
\m_{S,\omega}(\beta'\circ\beta_{|C}) ~=~ \m_{S,\omega}(\beta') \cdot
\frac{ U_{C}(\beta_{|C})}{U_C(\beta_{|j})} \mper
\]
By Remark \ref{bias:rem}, the quantity $U_C(\beta_{|j})$ is independent of the constraint $C$ and
only depends on the variable $x_j$ and the assignment $\beta_{|j}$. Denoting the quantity by
$p_j(\beta)$, we can write the above expression as
\[
\m_{S,\omega}(\beta'\circ\beta_{|C}) ~=~ \m_{S,\omega}(\beta') \cdot
\frac{ U_{C}(\beta_{|C})}{p_j(\beta)} \mper
\]
%
%%%%%%%%%%%%%%%%%%%%%%%%%%
\newcommand{\degree}{\mathsf{deg}}
%%%%%%%%%%%%%%%%%%%%%%%%%%

We can now inductively simplify the expression for $\m_{S,\omega}(\beta)$. Let $x_{j_0}$ be the
first variable in $S$ according to the ordering $\omega$. Since we visit each
constraint exactly once, the numerator equals
\[
p_{j_0}(\beta) \cdot \prod_{C \in \constg  \cap \ball(S)}U_C(\beta_{|C}) \mper\]
Also, each variable $x_j \in \varg \cap \ball(S)$, except for
$x_{j_0}$, has exactly $\degree(x_j)-1$ children in the tree (where $\degree(x_j)$ denotes its degree in
the tree $\ball(S)$). Thus, the term $p_j(\beta)$ appears exactly $\degree(x_j) - 1$ times in the
denominator, for each $x_j \in \varg \cap \ball(S) \setminus \{ x_{j_0}\}$. The term
$p_{j_0}(\beta)$ appears $\degree(x_{j_0})$ times since all the neighbors of $x_{j_0}$ are its
children in the tree. Thus, we get
\[
\m_{S, \omega}(\beta) ~=~ \frac{\prod_{C \in \constg  \cap \ball(S)}U_C(\beta_{|C})}{\prod_{x_j \in
    \varg \cap \ball(S)} \inparen{p_j(\beta)}^{\degree(x_j)-1} } \mcom
\]
which is independent of the ordering $\omega$.
%
% Let $\cal{C}(B(S))$ be the set of constraints
% supported by the variables in $B(S)$.
% so that the numerator is always
% \[\prod_{C\in \cal{C}(B(S))} U_C(\beta_{|C}),\]
% regardless of the choice of $\omega$.
%
% The denominator, on the other hand, will now depend on $\omega$.
% In particular, the only difference between the orderings would be
% the following.
% For constraints $C'$ and $C$ (as above), if $C'$ is visited
% before $C$ in the breadth first traversal with respect to
% $\omega$ then the denominator will contain a term $U_C(\beta_{|j})$,
% but if $C$ is visited before $C'$ in a different breadth first
% traversal (say with respect to $\omega'$) then the denominator
% will contain the term $U_{C'}(\beta_{|j})$ instead.
%
% However, by Remark~\ref{samebias:rmk}, we have
% \[U_C(\beta_{|j})=U_{C'}(\beta_{|j}).\]
% Therefore, we can substitute the former with the latter in the
% denominator for $m_{S,\omega}$. Moreover, we can
% do this substitution for every pair of adjacent constraints
% $C$ and $C'$ when the order of traversal differs for
% $\omega$ and $\omega'$.
% Hence, $m_{S,\omega}(\beta)=m_{S,\omega'}(\beta)$
% for any assignment $\beta\in\{-1,1\}^{B(S)}$.
%
% Let $\alpha\in\pmone^T$ be a fixed assignment for
% $T\subseteq S$.
% The probability $m_{S,\omega}(\alpha)$ is obtained
% by summing over all $\beta$ consistent with $\alpha$ i.e.,
% \[\sum_{\beta\in\{\pm 1\}^{B(S)}\atop{\beta_{|T}=\alpha}} m_{S,\omega}(\beta).\]
% Therefore, $m_{S,\omega}(\alpha) = m_{S,\omega'}(\alpha)$.
% Hence, the statement follows.
\end{proof}

We now prove that the distributions $\m_S$ are \emph{locally consistent} \ie for any two sets $S_1$
and $S_2$, the distributions $\m_{S_1}$ and $\m_{S_2}$ agree on $S_1 \cap S_2$. It suffices to show
that for that for any two sets $S$ and $T$, with $S \sub T$, we have for all
any $\alpha \in \pmone^S$, $\m_S(\alpha) = \m_T(\alpha)$. Here $\m_T(\alpha)$ denotes the
probability that the variables in $S$ are assigned according to $\alpha$ in $\m_T$ when we
marginalize over the variables in $T \setminus S$. The distributions $\m_S$ will only satisfy this
approximately \ie we will be able to show that $\abs{\m_S(\alpha) - \m_T(\alpha)}$ is very
small. However, using a result of Raghavendra and Steurer \cite{raghsteu2}, we will be able to
correct the distributions $\inbraces{\m_S}$ to a family of distributions $\inbraces{\m_S'}$ such that
$\m_S'(\alpha) = \m_T'(\alpha)$ for all $\alpha$. We first prove the following.

\begin{lemma}[Approximate Local Consistency]\label{local-consistency:lm}
There exists a constant $c_0$ such that for any two sets $S \sub T \sub \varg$, with $|S| \leq |T|
\leq 2^{c_0 \cdot \eta d}$, we have
\[
\forall \alpha \in \pmone^S \qquad \abs{\m_S(\alpha) - \m_T(\alpha)} ~=~ 2^{-\Omega(\eta d)} \mcom
\]
when the distributions $\m_S$ and $\m_T$ are both well-defined.
\end{lemma}
\begin{proof}
Note that it suffices to prove the above for the case when $T = S \cup \{v\}$, since then by
triangle inequality we will have that for any $T$,
$\abs{\m_S(\alpha) - \m_T(\alpha)} \leq \abs{T \setminus S} \cdot 2^{-\Omega(\eta d)} = 2^{-\Omega(\eta d)}$.

Since $\m_T$ is well defined, $\ball(T)$ must be a forest in the graph $G$. Also,
the distributions in different components of $\ball(T)$ are independent and the components in
$\ball(T)$ which do not contain $v$ are identical in $\ball(S)$. Hence, the distribution over them
would be identical according to $\m_S$ and $\m_T$. Thus, it suffices to consider the case when
$\ball(T)$ is a tree (\ie we restrict ourselves to the component $\ball(T_i)$ of $\ball(T)$ which
contains $v$).

Note that even though $\ball(T)$ is assumed to be a tree, we could still have that
$\ball(S)$ is a forest with more than one components, which get connected in $\ball(T) = \ball(S
\cup \{v\})$. We first consider a simple special case when $\ball(S)$ \emph{is also a tree}.

\paragraph{Case 1: $\ball(T)$ is a tree and $\ball(S)$ is also a tree.}
In this case, since $\ball(S) \sub \ball(T)$, we must have that any edge $(u,v)$ which is present in
$\ball(S)$ must also be present in $\ball(T)$. Thus, the vertices in $\ball(T) \setminus \ball(S)$
must form a collection of subtrees of the tree $\ball(S)$.
By Lemma \ref{orddontmatter:lm}, we can assume that the distributions $\m_T$ and $\m_S$ are
defined with the same starting vertex in $S$. Since the distribution $\m_T$ is defined by a
breadth-first traversal of the tree $\ball(T)$, the probability of any assignment to the vertices in
$\ball(S)$ will remain unchanged even if we remove the subtrees corresponding to the vertices in
$\ball(T) \setminus \ball(S)$. Thus, in this case, we have that
\[
\forall \alpha \in \pmone^S \quad \m_S(\alpha) ~=~ \m_T(\alpha) \mper
\]

\paragraph{Case 2: $\ball(T)$ is a tree but $\ball(S)$ is a forest with more than one components.}
In this case, we might have $\ball(S) = \cup_{i=1}^t \ball(S_i)$, where the components
$\ball(S_i)$ are disconnected in $\ball(S)$, but become connected in $\ball(T)$.
Thus, the distributions over the different components
will be independent according $\m_S$ but will become correlated when we consider
$\m_T$.

However, recall that in the distribution $U_C$, with probability $\eta$ we assign the variables in
$S_C$ according to the uniform distribution on $\pmone^{S_C}$. This breaks the correlation between any two
variables in the constraint $C$. Since for any $i \neq j$, $\ball(S_i)$ and $\ball(S_j)$ are
disconnected, any path between $S_i$ and $S_j$ in $\ball(T)$ must have length at least $d$. We will use this to
show that the correlation between the variables in $S_i$ and $S_j$ must be small since (with high
probability) at some constraint $C$ along the path, we must assign the variables in $C$ uniformly.

However, the above intuition is slightly incorrect since in defining the distributions $\m_S$, we do
not assign all the variables of the constraint together, but assign $k-1$ variables conditioned on
one variable which is the parent of $C$ in the tree.
The following claim shows that even when assigning the variables in a constraint $C$, conditioned on the
value of one of its variables, we break the correlation between the variables with probability at
least $\eta/2$ \ie even the conditional distribution can be viewed as being a convex combination of
the uniform distribution and some other distribution.
\begin{claim}
Let $C$ be a constraint and let $j \in S_C$ be the index of  a variable involved in $C$. Let $\beta \in \pmone$ be
an assignment to $x_j$. Then the distribution $U_{C,\beta}$ on $\pmone^{S_C \setminus\{j\}}$ can be
written as
\[
U_{C,\beta} ~=~ \inparen{1-\frac{\eta}{2}} \cdot \m_C^{(\beta)} ~+~ \inparen{\frac{\eta}{2}} \cdot U_{k-1} \mcom
\]
where $\m_C^{(\beta)}$ is a distribution on $\pmone^{S_C \setminus\{j\}}$ that depends on $\beta$ and
$U_{k-1}$ denotes the uniform distribution on $\pmone^{S_C \setminus\{j\}}$.
\end{claim}
\begin{proof}
Let $p_{\beta}$ denote the probability that $x$ is assigned the value $\beta$ according to the
distribution $U_C$.
For any assignment $\beta' \in \pmone^{S_C \setminus\{j\}}$, we can write
\begin{align*}
U_{C,\beta}(\beta')
~=~ \frac{U_C(\beta \circ \beta')}{p_{\beta}}
&~=~ \frac{(1-\eta) \cdot \nu_C(\beta \circ \beta') + \eta \cdot 2^{-k}}{p_{\beta}} \\
&~=~ \frac{(1-\eta) \cdot \nu_C(\beta \circ \beta') + \eta \cdot 2^{-k}}{p_{\beta}} - \eta \cdot
2^{-k} + \inparen{\frac{\eta}{2}} \cdot 2^{k-1}\\
&~=~ \frac{(1-\eta) \cdot \nu_C(\beta \circ \beta') + \eta \cdot (1-p_{\beta})\cdot 2^{-k}}{p_{\beta}} +
\inparen{\frac{\eta}{2}} \cdot U_{k-1}(\beta') \mper
\end{align*}
Let $T_{\beta}(\beta')$ to denote the first term above. We can say that $T_{\beta}(\beta') =
(1-\eta/2) \cdot \m_C^{(\beta)}(\beta')$ for some distribution $\m_{\beta}$ if $T_{\beta}(\beta') \geq 0$
for all $\beta'$ and $\sum_{\beta'}T_{\beta}(\beta') = 1-(\eta/2)$. The condition $T_{\beta}(\beta')
\geq 0$ follows from observing that both the terms in the numerator of $T_{\beta}(\beta')$ are
non-negative. The second condition follows from noting that
\[
\sum_{\beta'} T_{\beta}(\beta')
~=~ \sum_{\beta'} \inparen{U_{C,\beta}(\beta') - \inparen{\frac{\eta}{2}} \cdot U_{k-1}(\beta')}
~=~ 1 - \frac{\eta}{2} \mper
\]
This gives a distribution $\m_C^{(\beta)}$ on $\pmone^{S_C \setminus\{j\}}$ such that $U_{C,\beta}
= (1-(\eta/2)) \cdot \m_C^{(\beta)}+ (\eta/2) \cdot U_{k-1}$.
\end{proof}
Thus, in the definition of the distributions $\m_S$, the process of assigning the remaining $k-1$
variables in a constraint $C$ conditioned on an assignment $\beta$ to one of the variables, can be
viewed as assigning them from the distribution $\m_C^{(\beta)}$ with probability $1-(\eta/2)$ and from
$U_{k-1}$ with probability $\eta/2$. We can equivalently view the definition of the distribution
$\m_S$ as first making the choice for every $C \in \constg \cap \ball(S)$, whether conditioned on
the parent of $C$, the rest of the variables in
$C$ will be assigned according to $\m_C^{(\beta)}$ (which happens with probability $1-(\eta/2)$)
or according to $U_{k-1}$ (which happens with probability $\eta/2$).

For $S \sub \varg$,  we say an edge from a constraint $C \in \constg \cap \ball(S)$ to a variable
$x_{j'} \in \varg \cap \ball(S)$  is \emph{broken} in $\m_S$, if $x_{j'}$ is assigned according to
$U_{k-1}$, conditioned on some other variable $x_j$ which is the parent of $C$ in $\ball(S)$. Note
that conditioned on the event that the edge from $C$ to $x_{j'}$ is broken, the distributions of $x_j$
and $x_{j'}$ are independent. This is because for any assignment $\beta$ to $x_j$, $x_{j'}$ is
assigned uniformly in $\pmone$.

We say that a path between two variables $x_j, x_{j'} \in \varg \cap \ball(S)$ is broken if some
edge in the path is broken. Note that if the length of the path is $\ell$, then there $\ell/2$ edges
going from a constraint vertex to a variable vertex, and hence the path is broken with probability
at least $1-(1-\eta/2)^{\ell/2}$.
Also, we have as before that conditioned on the path between $x_j$ and $x_{j'}$ being broken, the
distributions of $x_{j}$ and $x_{j'}$ are independent.

Recall that we are considering the case when $\ball(S) = \cup_{i=1}^t \ball(S_i)$ is a forest but
$\ball(T) = \ball(S \cup \{v\})$ is a tree. Note that even though $v \notin S$, we can still have $v
\in \ball(S)$. We first present the argument for the case when this does not happen.

%%%%%%%%%%%%%%%%%%%%%%%%%%%%%%%%%%%%%%%%%%%
\newcommand{\cale}{{\cal E}}
%%%%%%%%%%%%%%%%%%%%%%%%%%%%%%%%%%%%%%%%%%%

\begin{itemize}
\item[-] \textbf{Case 2a:} $\mathbf{v \notin \ball(S)}$. Since $v \notin \ball(S)$,
any path from $v$ to $S_i$ for $i \in [t]$ must have length at least $d$. To
  analyze the distribution in this case, we first assume by Lemma \ref{orddontmatter:lm} that the
  starting vertex for defining the distribution is $v$. We define the following event for the distribution $\m_T$
\[
\cale ~:=~ \inbraces{\forall i \in [t], ~\text{all paths in $\ball(T)$ from $v$ to $S_i$ are
    broken}} \mper
\]
Since $\ball(T)$ is a tree, there is exactly one path from $v$ to a node in $S_i$. The probability
that the path is not broken is at most $(1-\eta/2)^{d/2}$. Hence,
\[
\prob{\bar{\cale}} ~\leq~ |S| \cdot (1-\eta/2)^{d/2} ~\leq~ 2^{-\Omega(\eta d)} \mcom
\]
for $|S| = O(\eta d)$. Also, if $\m_T(\alpha ~|~ \cale)$ denotes the probability of the assignment
$\alpha \in \pmone^S$ given than the event $\cale$ happens, then we can write
\begin{align*}
\m_T(\alpha)
&~=~ \prob{\cale} \cdot \m_T(\alpha ~|~ \cale) ~+~ \prob{\bar{\cale}} \cdot \m_T(\alpha
~|~ \bar{\cale}) \\
&~=~ \m_T(\alpha ~|~ \cale) ~\pm~ 2^{-\Omega(\eta d)} \mcom
\end{align*}
where we use $a = b \pm c$ to denote $\abs{a-b} \leq c$. Since the distribution of vertices
separated by a broken path is independent, conditioned on the event $\cale$, the distribution for
the sets $S_1, \ldots, S_r$ must be independent. Let $\alpha_i$ denote the restriction of the
assignment $\alpha$ to the set $S_i$. We then have by the independence that
\[ \m_T(\alpha ~|~ \cale) ~=~ \prod_{i=1}^t \inparen{\m_T(\alpha_i ~|~ \cale)}  \mper\]
Conditioned on the event $\cale$, the assignment for the set $S_i$ in the distribution $\m_T$ is
defined by considering a subtree of $\ball(T)$ which does not include any vertices from $S_j$ for $j
\neq i$. The distribution on this subtree will be identical if we instead define the assignment
according to the distribution $\m_{S_i \cup \{v\}}$ conditioned on all paths from $v$ to $S_i$ in
$\ball(S_i \cup \{v\})$ being broken. Let $\cale_i$ denote this event. Then, since $\prob{\cale_i}
\geq 1-2^{-\Omega(\eta d)}$, we have,
\[
\m_T(\alpha_i ~|~ \cale)
~=~ \m_{S_i \cup \{v\}}(\alpha_i ~|~ \cale_i)
~=~ \m_{S_i \cup \{v\}}(\alpha_i) ~\pm~ 2^{-\Omega(\eta d)} \mper
\]
However, $\ball(S_i)$ is a tree and $\ball(S_i \cup \{v\})$ is also a tree and hence by Case 1 we
have $\m_{S_i \cup \{v\}}(\alpha_i) = \m_{S_i}(\alpha_i)$. Combining the above and using the fact
that the components $\ball(S_i)$ are disconnected in $\ball(S)$, we get
\[
\m_T(\alpha)
~=~ \prod_{i=1}^t \inparen{\m_{S_i}(\alpha_i)} ~\pm~ 2^{-\Omega(\eta d)}
~=~ \m_S(\alpha) ~\pm~ 2^{-\Omega(\eta d)} \mper
\]
\item[-] \textbf{Case 2b:} $\mathbf{v \in \ball(S)}$. Without loss of generality, let $v \in
  \ball(S_1)$. The treatment for this case is almost identical except that we need to treat the set $S_1$
  more carefully since the paths from $v$ to $S_1$ may now be short. We now define the event $\cale$
  as
\[
\cale ~:=~ \inbraces{\forall i \in \{2,\ldots,t\}, ~\text{all paths in $\ball(T)$ from $v$ to $S_i$ are
    broken}} \mper
\]
As before, conditioned on the event $\cale$, the distributions on different sets $S_i$ are
independent and we can write
\[
\m_T(\alpha ~|~ \cale)
~=~ \m_T(\alpha ~|~ \cale) ~\pm~ 2^{-\Omega(\eta d)}
~=~ \prod_{i=1}^t \inparen{\m_T(\alpha_i ~|~ \cale)}  ~\pm~ 2^{-\Omega(\eta d)} \mper
\]
Again, we have that conditioned on the event $\cale$, the distribution for the assignment $\alpha_i$
is defined by considering a subtree not containing any set $S_j$ for $i \neq j$. For $i \geq 2$, letting $\cale_i$
denote the event that all paths between $v$ and $S_i$ in $\ball(S_i \cup \{v\})$ are broken, we
have
\[
\prod_{i=1}^t \inparen{\m_T(\alpha_i ~|~ \cale)}
~=~ \m_{S_1 \cup \{v\}}(\alpha_1) \cdot \prod_{i=2}^t \inparen{\m_{S_i \cup \{v\}}(\alpha_i ~|~
  \cale_i)} \mper
\]
As before $\m_{S_i \cup \{v\}}(\alpha_i ~|~ \cale_i)  = \m_{S_i \cup \{v\}}(\alpha_i) \pm
2^{-\Omega(\eta d)}$ for all $i \geq 2$ and $\m_{S_i \cup \{v\}}(\alpha_i) = \m_{S_i}(\alpha_i)$ for
all $i \in [t]$ by Case 1. Combining the above, we again get
\[
\m_T(\alpha)
~=~ \prod_{i=1}^t \inparen{\m_{S_i}(\alpha_i)} ~\pm~ 2^{-\Omega(\eta d)}
~=~ \m_S(\alpha) ~\pm~ 2^{-\Omega(\eta d)} \mper
\]
\end{itemize}
\end{proof}

Now that we have a family $\inbraces{\m_S}_{|S| \leq r}$ of approximately locally consistent
probability distributions, we can use it to define locally consistent distribution
$\inbraces{\m_S'}_{|S|\leq t}$ using a result by Raghavendra and Steurer ~\cite{raghsteu2}.
\begin{lemma}[\cite{raghsteu2}]\label{makecons:lm}
Let $\inbraces{m_S:\{-1,1\}^S\to\R_+}_{|S|\le r}$ be a family of probability distributions
such that for all $S \subseteq T$ and $\alpha\in\{-1,1\}^S$:
\[
\left|m_S(\alpha) -m_T(\alpha)\right| ~\le~ \eps_0.
\]
Then there exists a family of probability distributions $\inbraces{m_S' : \{-1,1\}^S\to\R_+}_{|S|\le r}$
such that for all $ S\subseteq T$ and $\alpha\in\{-1, 1\}^S$:
\[
m_S'(\alpha) ~=~ m_T'(\alpha) \mcom
\]
and for all $S$ with $|S|\le r$, we have $\norm{m_S-m_S'}_1\le O(2^r \cdot \eps_0)$.
\end{lemma}

Therefore, using Lemmas~\ref{local-consistency:lm}, \ref{makecons:lm},
and \ref{almsatsol:lm}, we can now prove the following theorem,
which also completes our proof of Theorem~\ref{lp:thm}.
\begin{theorem}
Let $f:\pmone^k \to \B$ be a predicate such that there exists a vanishing measure $\Lambda^*$ on
$\C^*(f)$. Then, for every $\eps > 0$, there is a constant $c_{\eps} > 0$, such that for all
large enough $n$, there exists an instance of $\maxkcsp(f)$ on $n$ variables satisfying the following:
\begin{itemize}
\item[-] For any integral assignment $\psi$, the fraction of the constraints satisfied is in the range
$[\rho(f)-\eps, \rho(f)+\eps]$.
\item[-] The optimum for the linear program obtained by $c_{\eps} \cdot \log\log n$ rounds of the
  Sherali-Adams hierarchy is at least $1 - O(k\cdot \sqrt{\eps})$.
\end{itemize}
\end{theorem}
\begin{proof}
The proof follows simply from appropriate choices for the parameters $\eta, d, r$ and $\delta$.
Using Lemma \ref{finallpinst:lm} we obtain an instance such that the constraint graph $G$ has girth
$g = O(\log n)$, degree $D = O_{\eps}(1)$ and such that the fraction of constraints satisfied by
any integral assignment $\psi$ is between $\rho(f) - \eps$ and $\rho(f) + \eps$.

Using Claim \ref{setsize:clm}, we can define the distributions $\m_S$ for all sets of size at most
$r$, when $r \cdot D^d \leq g$. Setting $\delta = \sqrt{\eps}$ and $\eta = k\cdot\sqrt{\eps}$, we
get from Lemma \ref{almsatsol:lm} that the distributions $\m_S$ achieve LP value
$1-O(k\sqrt{\eps})$. Taking the error $2^{-\Omega(\eta d)}$ from Lemma \ref{local-consistency:lm} to be
$\eps_0$ in Lemma \ref{makecons:lm}, and $r$ to be such that $2^{r} \cdot \eps_0 = O(\sqrt{\eps})$, we get that
the LP value achieved by the the distributions $\m_S'$ is at least $O(1-k\sqrt{\eps})$. Since we
only need that $r \cdot D^d \leq g$ for defining the distributions and $2^r \cdot 2^{-\Omega(\eta
  d)} = O(\sqrt{\eps})$ for bounding the error, we can choose both $d$ and $r$ to be
$\Omega_{\eps}(\log g) = \Omega_{\eps}(\log \log n)$.
\end{proof}

\section*{Acknowledgements}
The first author would like to thank Per Austrin and Johan H\aa stad for many discussions over the
years on the topic of approximation resistance.

\bibliographystyle{plain}
\bibliography{characterization}
\end{document}